\documentclass[onecolumn]{IEEEtran}
\usepackage{multirow}
\usepackage{multicol}
\usepackage{times}
\usepackage{times,amssymb,amsmath,amsfonts,float,color,cite,bbm,mathrsfs,float,stmaryrd,amsbsy,mathbbol}
\usepackage{amsfonts}
\usepackage{latexsym}
\usepackage{theorem}
\usepackage{graphicx}
\usepackage{subcaption}
\usepackage{enumitem}
\usepackage{algorithm,algpseudocode}
\usepackage[normalem]{ulem}
\usepackage{xcolor}

\newtheorem{theorem}{Theorem}[section]
\newtheorem{lemma}[theorem]{Lemma}
\newtheorem{claim}[theorem]{Claim}

\newtheorem{corollary}[theorem]{Corollary}
\newtheorem{definition}[theorem]{Definition}

\newtheorem{example}[theorem]{Example}
\theoremstyle{definition}
\newtheorem{remark}[theorem]{Remark}

\renewcommand{\thefigure}{\thesection.\arabic{figure}}
\usepackage{amsmath}

\allowdisplaybreaks

\makeatletter
\renewcommand{\@endtheorem}{\endtrivlist}
\makeatother

\newcommand\remove[1]{}

\makeatletter
\renewcommand{\thefigure}{{\@arabic\c@figure}}
\renewcommand{\fnum@figure}{{\bf Figure\,\thefigure}}
\makeatother

\newcommand\nc\newcommand
\nc{\cA}{\mathcal{A}}\nc{\cB}{\mathcal{B}}\nc{\cC}{\mathcal{C}}\nc{\cD}{\mathcal{D}}
\nc{\cE}{\mathcal{E}}\nc{\cF}{\mathcal{F}}\nc{\cG}{\mathcal{G}}\nc{\cH}{\mathcal{H}}
\nc{\cI}{\mathcal{I}}\nc{\cJ}{\mathcal{J}}\nc{\cK}{\mathcal{K}}\nc{\cL}{\mathcal{L}}
\nc{\cM}{\mathcal{M}}\nc{\cN}{\mathcal{N}}\nc{\cO}{\mathcal{O}}\nc{\cP}{\mathcal{P}}
\nc{\cQ}{\mathcal{Q}}\nc{\cR}{\mathcal{R}}\nc{\cS}{\mathscr{S}}\nc{\cT}{\mathcal{T}}
\nc{\cU}{\mathcal{U}}\nc{\cV}{\mathcal{V}}\nc{\cW}{\mathcal{W}}\nc{\cX}{\mathcal{X}}
\nc{\cY}{\mathcal{Y}}\nc{\cZ}{\mathcal{Z}}

\nc{\bba}{\mathbf{a}}\nc{\bbb}{\mathbf{b}}\nc{\bbc}{\mathbf{c}}\nc{\bbd}{\mathbf{d}}
\nc{\bbe}{\mathbf{e}}\nc{\bbf}{\mathbf{f}}\nc{\bbg}{\mathbf{g}}\nc{\bbh}{\mathbf{h}}
\nc{\bbi}{\mathbf{i}}\nc{\bbj}{\mathbf{j}}\nc{\bbk}{\mathbf{k}}\nc{\bbl}{\mathbf{l}}
\nc{\bbm}{\mathbf{m}}\nc{\bbn}{\mathbf{n}}\nc{\bbo}{\mathbf{o}}\nc{\bbp}{\mathbf{p}}
\nc{\bbq}{\mathbf{q}}\nc{\bbr}{\mathbf{r}}\nc{\bbs}{\mathbf{s}}\nc{\bbt}{\mathbf{t}}
\nc{\bbu}{\mathbf{u}}\nc{\bbv}{\mathbf{v}}\nc{\bbw}{\mathbf{w}}\nc{\bbx}{\mathbf{x}}
\nc{\bby}{\mathbf{y}}\nc{\bbz}{\mathbf{z}}

\nc{\bbA}{\mathbf{A}}\nc{\bbB}{\mathbf{B}}\nc{\bbC}{\mathbf{C}}\nc{\bbD}{\mathbf{D}}
\nc{\bbE}{\mathbf{E}}\nc{\bbF}{\mathbf{F}}\nc{\bbG}{\mathbf{G}}\nc{\bbH}{\mathbf{H}}
\nc{\bbI}{\mathbf{I}}\nc{\bbJ}{\mathbf{J}}\nc{\bbK}{\mathbf{K}}\nc{\bbL}{\mathbf{L}}
\nc{\bbM}{\mathbf{M}}\nc{\bbN}{\mathbf{N}}\nc{\bbO}{\mathbf{O}}\nc{\bbP}{\mathbf{P}}
\nc{\bbQ}{\mathbf{Q}}\nc{\bbR}{\mathbf{R}}\nc{\bbS}{\mathbf{S}}\nc{\bbT}{\mathbf{T}}
\nc{\bbU}{\mathbf{U}}\nc{\bbV}{\mathbf{V}}\nc{\bbW}{\mathbf{W}}\nc{\bbX}{\mathbf{X}}
\nc{\bbY}{\mathbf{Y}}\nc{\bbZ}{\mathbf{Z}}

\nc{\bfA}{\mathbb{A}}\nc{\bfB}{\mathbb{B}}\nc{\bfC}{\mathbb{C}}\nc{\bfD}{\mathbb{D}}
\nc{\bfE}{\mathbb{E}}\nc{\bfF}{\mathbb{F}}\nc{\bfG}{\mathbb{G}}\nc{\bfH}{\mathbb{H}}
\nc{\bfI}{\mathbb{I}}\nc{\bfJ}{\mathbb{J}}\nc{\bfK}{\mathbb{K}}\nc{\bfL}{\mathbb{L}}
\nc{\bfM}{\mathbb{M}}\nc{\bfN}{\mathbb{N}}\nc{\bfO}{\mathbb{O}}\nc{\bfP}{\mathbb{P}}
\nc{\bfQ}{\mathbb{Q}}\nc{\bfR}{\mathbb{R}}\nc{\bfS}{\mathbb{S}}\nc{\bfT}{\mathbb{T}}
\nc{\bfU}{\mathbb{U}}\nc{\bfV}{\mathbb{V}}\nc{\bfW}{\mathbb{W}}\nc{\bfX}{\mathbb{X}}
\nc{\bfY}{\mathbb{Y}}\nc{\bfZ}{\mathbb{Z}}

\nc{\sA}{\mathsf{A}}\nc{\sB}{\mathsf{B}}\nc{\sC}{\mathsf{C}}\nc{\sD}{\mathsf{D}}
\nc{\sE}{\mathsf{E}}\nc{\sF}{\mathsf{F}}\nc{\sG}{\mathsf{G}}\nc{\sH}{\mathsf{H}}
\nc{\sI}{\mathsf{I}}\nc{\sJ}{\mathsf{J}}\nc{\sK}{\mathsf{K}}\nc{\sL}{\mathsf{L}}
\nc{\sM}{\mathsf{M}}\nc{\sN}{\mathsf{N}}\nc{\sO}{\mathsf{O}}\nc{\sP}{\mathsf{P}}
\nc{\sQ}{\mathsf{Q}}\nc{\sR}{\mathsf{R}}\nc{\sS}{\mathsf{S}}\nc{\sT}{\mathsf{T}}
\nc{\sU}{\mathsf{U}}\nc{\sV}{\mathsf{V}}\nc{\sW}{\mathsf{W}}\nc{\sX}{\mathsf{X}}
\nc{\sY}{\mathsf{Y}}\nc{\sZ}{\mathsf{Z}}

\nc{\bi}{\mathbb{i}}

\newcommand{\mathset}[1]{\left\{#1\right\}}
\newcommand{\abs}[1]{\left|#1\right|}

\newcommand{\parenv}[1]{\left( #1 \right)}
\newcommand{\sparenv}[1]{\left[ #1 \right]}

\nc{\set}[1]{\llbracket #1 \rrbracket}
\nc{\bmat}[1]{\begin{bmatrix} #1 \end{bmatrix}}

\newcommand{\bal}[1]{\begin{align}\label{#1}}
\newcommand{\eal}{\end{align}}

\renewcommand{\le}{\leqslant}
\renewcommand{\leq}{\leqslant}
\renewcommand{\ge}{\geqslant}
\renewcommand{\geq}{\geqslant}

\renewcommand{\Bbb}{\mathbb}

\newcommand{\Cref}[1]{Co\-ro\-lla\-ry\,\ref{#1}}

\renewcommand{\Bbb}{\mathbb}

\newcommand{\N}{{\Bbb N}}

\newcommand{\R}{{\Bbb R}}
\newcommand{\Z}{{\Bbb Z}}
\newcommand{\E}{{\Bbb E}}

\newcommand{\cAS}{\cA^{\star}}

\newcommand{\fr}{\mathsf{fr}}
\newcommand{\ct}{\mathsf{ct}}
\newcommand{\bfr}{\mathsf{\mathbf{fr}}}
\newcommand{\bct}{\mathsf{\mathbf{ct}}}
\newcommand{\bbeta}{\boldsymbol{\beta}}

\newcommand{\0}{\mathbf{0}}
\newcommand{\1}{\mathbf{1}}
\newcommand{\bo}{\mathbbm{1}}
\nc{\vt}{\vartheta}
\nc{\vtp}{\vartheta_{\bbP}}
\nc{\vtpk}{\vartheta_{\bbP,k}}

\DeclareMathOperator{\supp}{Supp}

\DeclareMathOperator{\var}{Var}

\outer\def\proclaim #1. #2\par{\medbreak
 \noindent{\bf#1.\enspace}{\sl#2\par}%
 \ifdim\lastskip<\medskipamount \removelastskip\penalty55\medskip\fi}

\begin{document}

\title{{Central Limit Theorem for Mutation Systems}}

\author{\IEEEauthorblockN{Liav Koram} \and\hspace*{.5in} \IEEEauthorblockN{Ohad Elishco} 
}

\maketitle

{\renewcommand{\thefootnote}{}\footnotetext{

\vspace{-.2in}
 
\noindent\rule{1.5in}{.4pt}

{ 
L. Koram and O. Elishco are with the School of Electrical and Computer Engineering, Ben-Gurion University of the Negev, 
Beer Sheva, 8410501, Israel. Email: liavk@post.bgu.ac.il, elishco@gmail.com. 
This research was supported by the Israel Science Foundation (Grant No. 1789/23).
}
\vspace{-.1in}
}}


\begin{abstract} 
DNA-based storage has emerged as a promising alternative to traditional data storage methods, offering unmatched advantages
in data density, longevity, and sustainability. 
Two main approaches have developed: in-vitro storage, where information is synthesized in controlled environments, and in-vivo storage, where data is embedded within an organism’s DNA for enhanced confidentiality and protection. 
While in-vivo DNA storage provides unique advantages, it faces significant challenges from mutations, including duplications, deletions, and substitutions, which cause sequence evolution over time. 
Thus, in-vivo systems experience continuous sequence alterations that increase length and change composition, making error correction particularly challenging.

We study the asymptotic behavior of mutation systems, which model the probabilistic evolution of sequences over a finite alphabet, and are central to the analysis of in vivo DNA-based data storage. 
Building upon prior works that established the limit of empirical $k$-tuple frequencies, we characterize the stochastic fluctuations around these values by establishing a Central Limit Theorem (CLT). 
Our approach leverages the spectral properties of the $k$-substitution matrix to project the centered count vectors, allowing us to approximate the system via a martingale difference sequence, and then verifying the classical martingale CLT conditions. 
In addition, we explicitly derive the limiting covariance matrix. 

\end{abstract}

\section{Introduction}\label{Sec:Intro}
The study of mutation systems, which describe the evolution of sequences over a finite alphabet through repeated applications of a probabilistic mutation law, has gained increasing prominence in various scientific and technological domains. These systems, characterized by a finite alphabet, an initial word, and a mutation rule defined by probability measures and a random function, offer a mathematical framework for understanding processes where sequences undergo random transformations. 
The dynamic nature of these systems, where words evolve step by step through the selection of a random position and the application of the mutation rule, makes them relevant to fields ranging from evolutionary biology to DNA-based data storage \cite{lou2019evolution,elishco2024longterm,mundy2004origin,international2001initial,clelland1999hiding,heider2007dna,jupiter2010dna,liss2012embedding,shipman2017crispr}. 

The prospect of utilizing DNA as a medium for digital data storage has spurred significant research interest due to its remarkable capacity for information density and long-term stability \cite{balado2012capacity,wong2003organic,yazdi2017portable}. In particular, the concept of in-vivo DNA storage, where data is encoded within the DNA of a living organism, presents both opportunities and challenges (for the feasibility of in-vivo DNA storage see \cite{shipman2017crispr}). A key challenge in this context is the inherent susceptibility of DNA to mutations, which can introduce errors and potentially corrupt the stored information. 

Research on \emph{in vivo} DNA storage generally focuses on two interconnected aspects. The first is the design of error-correcting codes capable of combating mutations, with a primary focus on duplications \cite{kovavcevic2018asymptotically,kovacevic2022maximum,lenz2018bounds,lenz2017bounds,tang2021error,lenz2019duplication,jain2017duplication,mahdavifar2017asymptotically,sala2017exact,tang2021error,tang2020single,yu2024duplication,liu2025explicit,yohananov2025coding,sun2026palindromic}. The second aspect concerns the (long-term) properties and behavior of mutation systems. Analyzing these, sheds light on the reliability and efficiency of DNA data storage, as well as on other applications where sequential data undergoes probabilistic changes \cite{jain2017capacity,ben2022reverse,lou2019evolution,elishco2019entropy,farnoud2015capacity}.

To gain deeper insights into the evolution of mutation systems, a valuable approach involves analyzing the frequency and distribution of $k$-tuples within the sequences generated by these processes \cite{farnoud2019estimation,lou2019evolution,elishco2019entropy,ben2022reverse}. 
A $k$-tuple, in this context, is defined as a contiguous subword of length $k$. By studying these $k$-tuple frequencies, we may uncover fundamental aspects of the compositional changes and the potential for long-term stability within the sequences evolving under the mutation law. Indeed, \cite{lou2019evolution} has indicated that the frequency of $k$-tuples can be used to bounds the entropy of a mutation system. 

As a first step in this line of inquiry, previous research has focused on establishing the limit composition of a sequence under various mutation models. 
Using stochastic approximation methods, \cite{lou2019evolution} investigated the evolution of $k$-tuple frequencies and entropy in systems undergoing duplication and substitution mutations, establishing almost sure convergence. 
Later, \cite{elishco2024longterm} studied empirical $k$-tuple frequencies in somewhat more general mutation systems, establishing a limit in probability under mild assumptions. 
While these results provide a basic understanding of the eventual deterministic state towards which a system tends, convergence alone (whether almost sure or in probability) does not characterize the stochastic fluctuations around this average or the rate at which convergence occurs.

To further understand the asymptotic behavior of $k$-tuples in mutation systems, it is natural to draw inspiration from the rich literature on limit theorems for general stochastic processes \cite{jacod2013limit}, and specifically on urn models \cite{mahmoud2008polya}. 
Urn models, which involve the probabilistic drawing and replacement of balls of different colors, share conceptual similarities with mutation systems. 
Both evolve through a sequence of stochastic steps where each transition depends on the system's current state. 
However, while urn models primarily track the aggregate frequency of individual symbols, mutation systems focus on $k$-tuples, meaning the exact sequential order of the symbols is strictly taken into account.

In \cite{smythe1996central}, the author provides central limit theorems for both the composition of the urn (the number of balls of each type) and the number of times balls of each type are drawn. These theorems typically rely on the spectral properties (eigenvalues and eigenvectors) of the system's generating matrix. Given that the previous result on the limit in probability for $k$-tuples also hinges on the spectral properties of a matrix (the $k$-substitution matrix), the techniques and results from \cite{smythe1996central} on urn models offer potentially valuable insights and a framework for investigating a central limit theorem for $k$-tuples in mutation systems. 
Different replacement rules in urn models might provide parallels to how various mutation laws (e.g., substitution, duplication) could influence the central limit theorem for $k$-tuples. 

Building upon \cite{elishco2024longterm}, the natural next step in understanding the stochastic behavior of $k$-tuple frequencies in mutation systems is to seek a central limit theorem (CLT). Establishing such a theorem would provide a more complete characterization of the system's long-term dynamics by describing the limiting distribution of the fluctuations around the expected frequencies. 
A CLT can also offer crucial information about the rate at which the distribution of these frequencies converges to its limiting form. 
Furthermore, the power of a CLT lies in its ability to enable statistical inference, such as the construction of confidence intervals and the formulation of hypothesis tests regarding the distribution of $k$-tuples within the evolving sequences. 
Determining whether these fluctuations are asymptotically normal is essential for making accurate predictions about sequence compositions. 
Therefore, this paper establishes a central limit theorem for $k$-tuple frequencies. 
Drawing inspiration from the spectral techniques used in urn models \cite{smythe1996central,elishco2024longterm}, we extend the matrix-based probability limit analysis by leveraging martingale limit theorems.

The remainder of this paper is structured as follows. 
Section~\ref{sec:pre} provides preliminaries, including notation, definitions, and relevant prior results. 
Section~\ref{sec:main_res} presents the problem setup, an outline of the proof, and the paper's main result. 
In Section~\ref{sec:approx}, we construct an approximation using a martingale sequence. 
The central limit theorem for this martingale sequence is proven in Section~\ref{sec:MCLT}. 
Section~\ref{sec:cov_matrix} details the calculation of the covariance matrix. 
The non-diagonalizable case is addressed in Section~\ref{sec:non_diag}, and Section~\ref{sec:conc} concludes the paper. 

\section{Preliminaries}\label{sec:pre}
In this section, we introduce basic notation and definitions used throughout the paper. 
We also review known results on mutation systems, along with several new corollaries that will be utilized in our analysis. 
For consistency, we adopt the notation from~\cite{elishco2024longterm}. 

Let $\cA=\mathset{a_0,\dots,a_{d-1}}$ be an ordered, finite alphabet of size $|\cA|=d$. 
A word (or a string) $w=w_0w_1\dots w_{n-1}$ over $\cA$ is a finite concatenation of symbols from $\cA$. 
We use $|\cdot|$ to indicate both a set size and the length of a word. 
For $k\in \N$, we use $\cA^k$ to denote all $k$-length words, which also referred to as $k$-tuples, over $\cA$. 
We denote by $\cAS$ the set of all finite words over $\cA$. 

For an integer $n\in\N$, we denote $[n]:=\mathset{0,1,\dots, n-1}$ and for $I\subseteq \Z$ and $r\in \Z$ we write 
$r+I:=\mathset{r+j ~:~ j\in I}$. 
For a word $w\in \cA^n$ and for $I\subseteq [n]$, we denote by $w_{[I]}$ the restriction of $w$ to the positions in $A$. 
For $i,j\in [n]$, we also use $w_i^j$ to denote $w_{i+[j-i]}=w_iw_{i+1}\dots w_j$. 
If $i>j$ then $w_i^j$ denotes the empty word.

For $u,v\in\cAS$, $uv$ denotes the concatenation of $u$ and $v$, and $u^n$ represents the concatenation of $u$ with itself $n$ times. 
We say that $u$ is a subword (or substring) of $v$ if $|u|\leq |v|$ and there exists $i\in [|v|]$ such that $v_{i+[|u|]}=u$, 
with coordinates taken modulo $|v|$.

We denote by $\ct_v(u)$ the number of times $u$ appears in $v$ as a subword (the count of $u$), i.e., 
\[\ct_v(u):=\sum_{i=0}^{|v|} \bo_{[u=v_{i+[|u|]}]}\]
where coordinates are taken modulo $|v|$ and 
\[\bo_{[a=b]}=\begin{cases} 1 & \text{ if } a=b \\ 0 & \text{ otherwise}\end{cases}\] 
is the indicator function. 
We denote by $\fr_v(u)$ the empirical frequency (empirical probability) of $u$ in $v$, that is,
\[\fr_v(u)=\begin{cases} 0 & \text{ if } |u|>|v| \\
\frac{1}{|v|}\ct_v(u) & \text{ if } |u|\leq |v|\end{cases}.\]

For a given $k\in\N$, the count vector $\bct^{(k)}_v\in\R^{d^{k}}$ has coordinates corresponding to $k$-tuples. For any $u\in\cA^{k}$, $\bct^{(k)}_v(u)$ represents the count of occurrences of $u$ in $v$. When $k$ is evident or does not affect the result, we simplify it to $\bct_v$. Similarly, $\bfr^{(k)}_v$ denotes the vector of empirical frequencies, and when $k$ is clear from the context (or does not affect the result) we write $\bfr_v$. 

We now work our way to define a mutation system. 
For a set $A$, let $\bbP$ denote a distribution over $A$. The \textbf{support} of $\bbP$, denoted as $\supp(\bbP)$, is a subset of $A$ containing the arguments with positive probability, i.e., $\supp(\bbP)=\mathset{a\in A ~:~ \bbP(a)>0}$.  
\begin{definition} 
\label{def:mutation_law}
    For an alphabet $\cA$ of size $d$, a \textbf{mutation law} is a pair $(\vt,\bbP)$, where $\bbP=\mathset{\bbP_{a_t}}_{t\in [d]}$ is a set of $d$ non-degenerate, finitely supported probability measures on $\cAS$, and $\vt:\cA\to \cAS$ is a (random) function such that 
    for every $a\in\cA$, $\vt(a)$ is a random word chosen according to $\bbP_{a}$. 

    If there exists $\tau \in \R$ such that for every $a\in \cA$, $\E\parenv{\abs{\vt(a)}}=\tau$, we say that $(\vt,\bbP)$ is an 
    \textbf{average $\tau$-mutation law}. Meaning, the average length of the word that replaces any symbol is $\tau$. 
\end{definition}

Let $w\in\cA^m$ be a word of length $m$ and consider a mutation law $(\vt,\bbP)$. A \textbf{mutation step} (or round) on $w$ is performed by selecting a position $i\in [m]$ uniformly at random, and then applying the random function $\vt$ on the symbol in this position $\vt(w_i)$, 
resulting in a random word $w_0\dots w_{i-1} \vt(w_i) w_{i+1}\dots w_{n-1}$. 
With a slight abuse of notation, we denote the obtained (random) word by $\vt(w)$.
We consider mutation systems which are processes that evolve in time according to a mutation law.  
\begin{definition}
    A \textbf{mutation system} is defined by $S=\parenv{\cA,w,(\vt,\bbP)}$, where $\cA$ is a finite alphabet, 
    $S(0)=w\in \cAS$ is a non-empty starting word, and $(\vt,\bbP)$ is a mutation law defined over $\cA$. 
    We call $S$ an \textbf{average $\tau$-mutation system} if $(\vt,\bbP)$ is an \textbf{average $\tau$-mutation law}. 

    A mutation system evolves according to mutation steps and is defined recursively. For $n\in \N$, set $S(n)=\vt(S(n-1))=\vt^n(w)$, where $\vt^n(w)=\vt(\vt(\dots( \vt(w))\dots ))$ is the random word obtained by applying $n$ mutation steps, starting with the word $w$.
\end{definition}

Let $v_1, v_2, \dots, v_{d^k} \in \cA^k$ be all the $k$-tuples in lexicographic order. We mark the following expressions with respect to $S(n)$. 
\begin{enumerate}
\item $Y_n=\abs{S(n)}$ denotes the (random) length of $S(n)$. 
\item $\bct_n = \bct_{S(n)} := \bmat{
        \ct_{S(n)}(v_1) & \ct_{S(n)}(v_2) & \dots & \ct_{S(n)}(v_{d^k})}^T$ denotes the count (column) vector of $S(n)$. 
\item $\bfr_n = \bfr_{S(n)} := \bmat{
        \fr_{S(n)}(v_1) & \fr_{S(n)}(v_2) & \dots & \fr_{S(n)}(v_{d^k})}^T$ denotes the frequency (column) vector of $S(n)$.
\end{enumerate}

The following example appears in \cite{elishco2024longterm} and serves as a basic example for a mutation system. 
\begin{example}[See \cite{elishco2024longterm}]
\label{ex:run_ex1}
Let $\cA=\mathset{0,1}$ be the binary alphabet, let the starting word be $w=01$ and assume $\vt$ is defined as follows 
\[\vt(0)=\begin{cases} 1& \text{w.p. } \frac{1}{2}\\ 00&\text{w.p. } \frac{1}{2}\end{cases}, \;  \vt(1)=\begin{cases} 0& \text{w.p. } \frac{1}{2}\\ 11&\text{w.p. } \frac{1}{2}\end{cases}.\] 
This mutation system is an average $(3/2)$-mutation system. 
To calculate the probability that $\vt(w)=11$, we need to choose the first symbol in the mutation step (probability $1/2$), and then the symbol needs to be transformed into $1$ (probability $1/2$). Overall, $\Pr(\vt(w)=11)=1/4$. 
\end{example}

Since $\bbP$ is assumed to be finitely supported, in every mutation step, only a finite number of symbols (and thus, $k$-tuples) is added. 
To capture the needed properties of $k$-tuples evolution in a mutation system, we use the substitution matrix defined in \cite{elishco2024longterm}. 
\begin{definition} 
\label{def:k-sub_mat}
    Let $S$ be an average $\tau$-mutation system over the alphabet $\cA=\mathset{a_0,\dots,a_{d-1}}$, and let $k\in\N$. 
    The $k$-\textbf{substitution matrix} $\bbM^{(k)}:=\bbM^{(k)}_{(\vt,\bbP)}$ is a real $d^k\times d^k$ matrix, with non-negative entries. 
    The coordinates in $\bbM^{(k)}$ correspond to $k$-tuples, ordered according to the lexicographic order. 
    For $u=(u_0 u_1 \dots u_{k-1}),v=(v_0 v_1 \dots v_{k-1})\in\cA^k$, the $(u,v)$ entry $\bbM^{(k)}_{u,v}$ is  
    \begin{align}
    \label{eq:k_sub_mat2}
        \bbM^{(k)}_{u,v}&:= \sum_{l\in \N}\sum_{\eta\in\cA^l}\sum_{t\in [d]}\Pr(\vt(a_t)=\eta)\parenv{\sum_{j=1}^{k-1}\bo_{[v_j=a_t]}\bo_{[(v_0^{j-1}\eta v_{j+1}^{k-1})_{[k]}=u]} +\sum_{j=0}^{l-1}\bo_{[v_0=a_t]}\bo_{[\eta_j^{l-1} v_1^{k-l+j}=u]}}
    \end{align} 
    where for $i\leq j$, $v_i^j=(v_i\dots v_j)$, and $(v_0^{j-1} \eta v_{j+1}^{k-1})_{[k]}$ is the word obtained by taking the first $k$ symbols from $v_0^{j-1} \eta v_{j+1}^{k-1}$.
    When $k$ is clear from the context, we write $\bbM$ instead of $\bbM^{(k)}$. Simply put, $\bbM^{(k)}_{u,v}$ is the sum of means of the number of $k$-tuples of type $u$ added to the mutation system, assuming the mutation step occurred on one of the symbols in a $k$-tuple $v$. The sum is over all the symbols in $v$.
\end{definition}

In the above definition, we assume that $k\geq l$, or in words, we assume that $k$ is larger than the random word $\eta$ that we insert. 
The same definition holds for smaller $k$, if we define $v_1^r$ for $r<1$ as the empty word and replace $[\eta_j^{l-1}v_1^{k-l+j}=u]$ with $[(\eta_j^{l-1}v_1^{k-l+j})_{[k]}=u]$. 

\begin{example}
\label{ex:run_ex2} 
Let us continue with Example \ref{ex:run_ex1}. 
Recall that $w=01$ and that 
\[\vt(0)=\begin{cases} 1& \text{w.p. } \frac{1}{2}\\ 00&\text{w.p. } \frac{1}{2}\end{cases}, \;  \vt(1)=\begin{cases} 0& \text{w.p. } \frac{1}{2}\\ 11&\text{w.p. } \frac{1}{2}\end{cases}.\] 
Calculating the substitution matrix for $k=2,3$ we obtain 
\begin{align*}
\bbM^{(2)}&= \bmat{\frac{3}{2}&1&\frac{1}{2}&0\\ \frac{1}{2}&1&0&\frac{1}{2}\\ \frac{1}{2}&0&1&\frac{1}{2}\\0&\frac{1}{2}&1&\frac{3}{2}}
\\
\bbM^{(3)}&= \bmat{2&\frac{3}{2}&\frac{1}{2}&0&\frac{1}{2}&0&0&0\\ \frac{1}{2}&1&\frac{1}{2}&1&0&\frac{1}{2}&0&0\\ \frac{1}{2}&0&1&\frac{1}{2}&0&0&\frac{1}{2}&0\\0&\frac{1}{2}&1&\frac{3}{2}&0&0&0&\frac{1}{2}\\
\frac{1}{2}&0&0&0&\frac{3}{2}&1&\frac{1}{2}&0\\
0&\frac{1}{2}&0&0&\frac{1}{2}&1&0&\frac{1}{2}\\
0&0&\frac{1}{2}&0&1&\frac{1}{2}&1&\frac{1}{2}\\
0&0&0&\frac{1}{2}&0&\frac{1}{2}&\frac{3}{2}&2\\}.
\end{align*} 
To demonstrate, we show the explicit calculation of the top-left entry in $\bbM^{(2)}$, i.e., the $(00,00)$ coordinate. 
To do so, we would look at each symbol in $v$, look at its possible mutations, and count the number of tuples of type $u$ we can find that begin in the first symbol in $v$. In our case, for the sake of the example $u,v=00$.

Looking at the first symbol in $v$, we will only find tuples of type $u$ if $\vt(v_0)=00$:
\[\underbrace{0}_{v_0}0\xrightarrow{\vt}\underbrace{00}_{\vt(v_0)}0\]
In $000$ we can count the tuple $u=00$ twice. For the second symbol in $v$, we will only find tuples of type $u$ if $\vt(v_1)=00$:
\[0\underbrace{0}_{v_1}\xrightarrow{\vt}0\underbrace{00}_{\vt(v_1)}\]
As we count only tuples that begin with $v_0$, we will count only 1 tuple of type $00$, summing it up we calculate:
\[\bbM^{(2)}_{00,00} = \Pr(\vt(0)=00)\times2 + \Pr(\vt(0)=00)\times1 = \frac32\]
thus we found the first value of $\bbM^{(2)}$ we can similarly find all the rest of its values, and do the same for $\bbM^{(3)}$.
\end{example}
 
Throughout, we assume that $\bbM^{(k)}$ is \textbf{irreducible}, thus satisfying the following theorem for non-negative irreducible matrices. 
\begin{theorem}\cite[Section 8.4]{horn2012matrix}
\label{th:irreducible_mat_properties} 
Let $\bbM\in\R^{n\times n}$ be a non-negative, irreducible square matrix. Then 
\begin{enumerate}
    \item $\rho(\bbM)>0$ is a simple eigenvalue of $\bbM$ (called the Perron eigenvalue). 
    \item The left and right eigenvectors associated with $\rho(\bbM)$ are the only strictly positive eigenvectors. 
    \item The maximum modulus eigenvalues of $\bbM$ are $e^{2\pi\bbi p/k}\rho(\bbM)$ for $p\in [k]$, and each has an algebraic multiplicity $1$. 
    \item if $\bbM\geq \bbN$ are real, non-negative matrices and $\bbM$ is irreducible, then $\rho(\bbM)\geq \rho(\bbN)$, with equality iff $\bbM=\bbN$. 
\end{enumerate}
\end{theorem}

In \cite[remark IV.5]{elishco2024longterm} it was shown that the spectral radius $\rho(\bbM^{(k)})$ of the $k$-substitution matrix of an average $\tau$-mutation law is equal to $\tau+k-1$. 
Additionally, every Jordan block corresponding to a maximum modulus eigenvalue is of size one. 
Meaning, $\mu_a(\lambda_0)=\mu_g(\lambda_0)$, for $\lambda_0 = \rho(\bbM^{(k)}) = \tau+k-1$, where $\mu_a(\lambda),\mu_g(\lambda)$ denote the \textbf{algebraic multiplicity} and the \textbf{geometric multiplicity} of $\lambda$, respectively. 
Furthermore, the vector of all ones, $\bbl=\1$ is a left eigenvector associated with the eigenvalue $\lambda_0=\tau+k-1$ which is the maximal eigenvalue.
Due to $\bbM^{(k)}$ being irreducible, Theorem \ref{th:irreducible_mat_properties} implies that the right eigenvector $\bbr$ associated with $\lambda_0$ is strictly positive and that $\lambda_0$ is unique. 

One of the main results in \cite{elishco2024longterm} is concerned with the limit in probability of the frequency of $k$-tuples. 
For completeness, it is presented below. 
\begin{theorem}[see Theorem IV.11 in \cite{elishco2024longterm}] 
\label{th:main3}
Let $S$ be an average $\tau$-mutation system over the alphabet $\cA=\mathset{a_0,\dots,a_{d-1}}$. 
Fix $k$ and let $w$ be a starting word of length $|w|\geq k$. 
Assume $\bbM^{(k)}$ has a unique real eigenvalue $\lambda_0=\tau+k-1$ and let $\bbl$ and $\bbr$ represent the left and right eigenvectors of $\bbM^{(k)}$ associated with $\lambda_0$, respectively, normalized such that $\bbr$ forms a probability vector, and $\bbl \cdot \bbr = 1$. 
If for every eigenvalue $\lambda$ of $\bbM^{(k)}$ with $\mu_a(\lambda)>\mu_g(\lambda)$, the real part of $\lambda$ is bounded above by $k$, $\Re(\lambda)<k$, then $\bfr_n\stackrel{P}{\to} \bbr$. 
\end{theorem}

In order to obtain a slightly stronger version of Theorem \ref{th:main3}, we use the Lebesgue–Vitali theorem 
(see, for example, \cite[Theorem 6.5.4]{ash2000probability}). 
While Lebesgue-Vitali's theorem is general and requires uniform integrability, we slightly modify it here to our needs and write a weaker version of the theorem. 
\begin{theorem}[Lebesgue-Vitali's Theorem] 
\label{th:Leb_Vit}
Let $\mathset{f_n}_n$ be non-negative, Borel measurable, uniformly bounded functions on a finite measure space. 
Assume $f_n\stackrel{P}{\to} f$ in probability. Then $f_n\stackrel{L^p}{\to} f$ for any $p\geq 1$. 
\end{theorem} 

Since $\bfr^{(k)}$ is bounded in every entry, as an immediate corollary, we obtain the $L^1$ convergence of $\bfr^{(k)}$.  
\begin{corollary}
\label{cor:fr_conv_L1} 
With the same notation and assumptions as in Theorem \ref{th:main3}, we have $\bfr^{(k)}_n\stackrel{L^1}{\to} \bbr$.
\end{corollary}

\begin{example}
\label{ex:run_ex3}
We continue with our running example. 
Recall that 
\[\vt(0)=\begin{cases} 1& \text{w.p. } \frac{1}{2}\\ 00&\text{w.p. } \frac{1}{2}\end{cases}, \;  \vt(1)=\begin{cases} 0& \text{w.p. } \frac{1}{2}\\ 11&\text{w.p. } \frac{1}{2}\end{cases},\] 
and 
\begin{align*}
\bbM^{(2)}&= \bmat{\frac{3}{2}&1&\frac{1}{2}&0\\ \frac{1}{2}&1&0&\frac{1}{2}\\ \frac{1}{2}&0&1&\frac{1}{2}\\0&\frac{1}{2}&1&\frac{3}{2}}
\\
\bbM^{(3)}&= \bmat{2&\frac{3}{2}&\frac{1}{2}&0&\frac{1}{2}&0&0&0\\ \frac{1}{2}&1&\frac{1}{2}&1&0&\frac{1}{2}&0&0\\ \frac{1}{2}&0&1&\frac{1}{2}&0&0&\frac{1}{2}&0\\0&\frac{1}{2}&1&\frac{3}{2}&0&0&0&\frac{1}{2}\\
\frac{1}{2}&0&0&0&\frac{3}{2}&1&\frac{1}{2}&0\\
0&\frac{1}{2}&0&0&\frac{1}{2}&1&0&\frac{1}{2}\\
0&0&\frac{1}{2}&0&1&\frac{1}{2}&1&\frac{1}{2}\\
0&0&0&\frac{1}{2}&0&\frac{1}{2}&\frac{3}{2}&2\\}
\end{align*}
Notice that both $\bbM^{(2)}$ and $\bbM^{(3)}$ are irreducible. 
The right leading eigenvectors are 
\begin{align*}
    \bbr^{(2)} = \frac{1}{10}\bmat{3\\ 2\\ 2\\ 3}, ~
    \bbr^{(3)} = \frac{1}{60}\bmat{11\\ 7\\ 5\\ 7\\ 7\\ 5\\ 7\\ 11}.
\end{align*}
From Corollary \ref{cor:fr_conv_L1}, the asymptotic frequency of pairs is given by $\bfr^{(2)}=\bbr^{(2)}$ and of triples is given by $\bfr^{(3)}=\bbr^{(3)}$. 
Notice that the limiting frequency is shift invariant (for example, $\fr(01)=\fr(10)$), and since the mutation is symmetric, the frequency of a $k$-tuple depends on its runs lengths and not on the symbols. 
\end{example}
We are now ready to present the main results of our paper. 

\section{Main Results, Problem Setup, and Proof outline}
\label{sec:main_res}
The main theorem of this paper is a Central Limit Theorem for the count vector in mutation systems. We present the theorem, and then outline the proof strategy by constructing a key decomposition using the eigenvector structure of the substitution matrix.

\begin{theorem}
\label{Main Main Theorem}
Let $S$ be a $\tau$-average mutation system, fix $k\in \N$, and let $\bbM^{(k)}=\bbM$ be the $k$-substitution matrix of $S$. 
Assume that $\bbM$ is irreducible and, for simplicity, suppose that it is diagonalizable. 

Denote by $\bct_n$ the count vector of $k$-tuples in $S(n)$. From Theorem IV.11 of \cite{elishco2024longterm}, we know that
\[\frac{1}{n}\bct_n\xrightarrow{\Pr}(\tau-1)\bbr,\]
where $\bbr$ is the leading right eigenvector of $\bbM$ normalized such that $\|\bbr\|_1 = 1$.

Suppose that for every non-leading eigenvalue $\lambda_i$ of $\bbM$, the real part satisfies 
\[\Re(\lambda_i)<k + \frac{1}{2}(\tau-1).\]
Furthermore, assume that the frequency vector $\bfr_n^{(2k-1)}$ of $(2k-1)$-tuples converges to some deterministic vector. 
Then the count vector satisfies a Central Limit Theorem:
\[\frac{1}{\sqrt{n}}\left(\bct_n - n(\tau-1)\bbr\right) \xrightarrow{d} \cN(\boldsymbol{0},\Sigma),\]
for some covariance matrix $\Sigma$.

Overall, the count vector exhibits the following asymptotic normal behavior:
\[\bct_n = n(\tau-1)\bbr + \sqrt{n}\boldsymbol{\xi} + o(\sqrt{n}),\]
where $\boldsymbol{\xi}$ is a (zero-mean) normal random vector.
\end{theorem}

\begin{remark}
The case where $\bbM$ is not diagonalizable will be addressed at the end of the paper. 
\end{remark}

Instead of proving Theorem \ref{Main Main Theorem} directly, we prove a slightly different theorem as explained below. 
The core idea relies on decomposing the centered count vector $\bct_n - n(\tau-1)\bbr$ along the left eigenvectors of $\bbM$. We first establish notation for the eigenstructure of $\bbM$.

\subsubsection{Eigenvalue and Eigenvector Notation}

Since $\bbM$ has dimension $d^k\times d^k$, there are $d^k$ eigenvalues. As $\bbM$ is real and diagonalizable, complex eigenvalues come in conjugate pairs. 

Let $\lambda_0$ denote the leading (Perron) eigenvalue of $\bbM$. Suppose that among the remaining $d^k - 1$ eigenvalues, $2m$ form complex conjugate pairs and $d^k - 2m - 1$ are purely real. We index the eigenvalues as follows:
\begin{itemize}
    \item $\lambda_0$ denotes the leading eigenvalue,
    \item $\lambda_1, \ldots, \lambda_{2m}$ denote the complex eigenvalues, where $\lambda_{i+m} = \overline{\lambda_i}$ for $i = 1, \ldots, m$, and
    \item $\lambda_{2m+1}, \ldots, \lambda_{d^k-1}$ denote the purely real, non-leading eigenvalues.
\end{itemize}

For each eigenvalue $\lambda_i$, we denote by $\bbu_i$ its corresponding left eigenvector of $\bbM$. We express each eigenvector and eigenvalue in terms of their real and imaginary parts:
\[\bbu_i = \bbu_i^r + \bbi\bbu_i^c \quad \text{and} \quad \lambda_i = \lambda_i^r + \bbi\lambda_i^c.\]

\subsubsection{Construction of the Transformation Matrix}

We now construct a matrix $\bbU$ whose rows are the real and imaginary parts of the left eigenvectors.

\begin{definition}
\label{def:U_matrix}
Define the matrix $\bbU \in \mathbb{R}^{d^k \times d^k}$ by
\[\bbU := \begin{bmatrix}
\bbu_0^r \\
\bbu_1^r \\
\bbu_1^c \\
\bbu_2^r \\
\bbu_2^c \\
\vdots \\
\bbu_m^r \\
\bbu_m^c \\
\bbu_{2m+1}^r \\
\vdots \\
\bbu_{d^k-1}^r
\end{bmatrix}.\]
\end{definition}

Note that we include only the real and imaginary parts of the first $m$ eigenvectors $\bbu_1,\dots,\bbu_m$. 
Since $\lambda_{i+m} = \overline{\lambda_i}$, the corresponding eigenvectors are conjugate pairs as well $\bbu_{i+m} = \overline{\bbu_i}$. This maintains linear independence and keeps $\bbU$ at dimension $d^k \times d^k$ as shown next.

\begin{lemma}
\label{lem:U_full_rank}
The matrix $\bbU$ is full rank and invertible.
\end{lemma}
\begin{IEEEproof}
We verify that all rows of $\bbU$ are linearly independent.

Since $\bbM$ is diagonalizable, its left eigenvectors $\bbu_0, \bbu_1, \ldots, \bbu_{d^k-1}$ form a basis for $\mathbb{C}^{d^k}$ and are therefore linearly independent. For the non-diagonalizable case we use the general eigen decomposition, the general left eigenvectors form a basis for $\mathbb{C}^{d^k}$ and are therefore linearly independent.

We need to show that replacing each complex conjugate pair $\{\bbu_i, \bbu_{i+m}\}$ with their real and imaginary parts $\{\bbu_i^r, \bbu_i^c\}$ preserves linear independence.

For a complex conjugate pair where $\bbu_{i+m} = \overline{\bbu_i}$, we have
\[\bbu_i = \bbu_i^r + \bbi\bbu_i^c \quad \text{and} \quad \bbu_{i+m} = \bbu_i^r - \bbi\bbu_i^c.\]

Note that this transformation is invertible:
\[\bbu_i^r = \frac{1}{2}\left(\bbu_i + \bbu_{i+m}\right) \quad \text{and} \quad \bbu_i^c = \frac{1}{2\bbi}\left(\bbu_i - \bbu_{i+m}\right).\]

Therefore, $\bbu_i^r$ and $\bbu_i^c$ lie in the span of $\bbu_i$ and $\bbu_{i+m}$, and conversely, $\bbu_i$ and $\bbu_{i+m}$ lie in the span of $\bbu_i^r$ and $\bbu_i^c$. This shows that
\[\text{span}\{\bbu_i, \bbu_{i+m}\} = \text{span}\{\bbu_i^r, \bbu_i^c\}.\]

Since replacing conjugate pairs with their real and imaginary parts preserves the span at each step, and the original eigenvectors are linearly independent, the rows of $\bbU$ remain linearly independent. Thus, $\bbU$ is full rank and due to being square matrix, is invertible.
\end{IEEEproof} 

\begin{remark}
\label{rem:U_full_rank} 
Note that Lemma \ref{lem:U_full_rank} holds independently of the diagonalization assumption. It requires only that the matrix $\bbU$ be real-valued, a condition that is satisfied in our setting.
\end{remark}

\subsubsection{Projection onto Eigenvector Components}

We apply the transformation $\bbU$ to the centered count vector to obtain projections onto the eigenvector directions.

\begin{definition}
\label{def:projection_vector}
For $n \in \N$, define the projection vector $\bbX_n \in \mathbb{R}^{d^k}$ by
\[\bbX_n = \bbU\left(\bct_n - n(\tau-1)\bbr\right).\]
The components of $\bbX_n$ are given by
\[\bbX_n := \begin{bmatrix}
X_n^{r(0)} \; X_n^{r(1)} \; X_n^{c(1)} \; X_n^{r(2)} \; X_n^{c(2)} \; \cdots \; X_n^{r(m)} \; X_n^{c(m)} \; X_n^{r(2m+1)} \; \cdots \; X_n^{r(d^k-1)} \end{bmatrix}^T,\]
where for each eigenvalue index $i$, we define
\[X_n^{(i)} = \bbu_i \cdot \left(\bct_n - n(\tau-1)\bbr\right) = X_n^{r(i)} + \bbi X_n^{c(i)},\]
with
\begin{align*}
X_n^{r(i)} &= \bbu_i^r \cdot \left(\bct_n - n(\tau-1)\bbr\right), \\
X_n^{c(i)} &= \bbu_i^c \cdot \left(\bct_n - n(\tau-1)\bbr\right).
\end{align*}
\end{definition}

\begin{remark}
Recall that the leading left eigenvector is a vector of all ones, $\bbu_0=\1$. since the leading right eigenvector $\bbr$ is normalized to be a probability vector, we have $\bbu_0 \cdot \bbr = 1$. Additionally, by orthogonality of the left non-leading eigenvectors with $\bbr$, we have $\bbu_i \cdot \bbr = 0$ for all $i > 0$. 
Moreover, $\bbu_0^r \cdot \bct_n$ sums the number of $k$-tuples overall, meaning the number of symbols as well. 
Consequently, the projection components simplify to
\begin{align*}
X_n^{r(0)} &= Y_n - n(\tau-1), \\
X_n^{r(i)} &= \bbu_i^r \cdot \bct_n \quad \text{for } i > 0, \\
X_n^{c(i)} &= \bbu_i^c \cdot \bct_n \quad \text{for } i > 0.
\end{align*}
\end{remark}

Thus, to prove Theorem \ref{Main Main Theorem}, it is enough to show that 
\[\frac{1}{\sqrt{n}}\bbX_n \xrightarrow{d} \cN(\boldsymbol{0}, \Sigma')\]
for some covariance matrix $\Sigma'$. 
This is because $\bbU$ is invertible and therefore 
\[\frac{1}{\sqrt{n}}\left(\bct_n - n(\tau-1)\bbr\right) = \bbU^{-1} \cdot \frac{1}{\sqrt{n}}\bbX_n \xrightarrow{d} \cN(\boldsymbol{0}, \Sigma),\]
where $\Sigma = \bbU^{-1}\Sigma'(\bbU^{-1})^T$. 

Thus, our main goal is to prove the following theorem.
\begin{theorem}
\label{main theorem} 
Let $S$ be an average $\tau$-mutation system with irreducible substitution matrix $\bbM^{(k)}$, and assume that the frequency of $(2k-1)$-tuples converges in probability to some (non-negative) vector $\fr_n^{(2k-1)}\xrightarrow{\Pr}\bbr^{(2k-1)}\geq 0$. 
If for every non-leading eigenvalue $\lambda=\lambda^r+\bbi \lambda^c$, we have $\lambda^r < k + \frac{1}{2}(\tau-1)$, then $n^{-\frac{1}{2}}\bbX_n$ has an asymptotic zero mean joint normal distribution. 
\end{theorem}
Recall that a vector has a joint normal distribution if and only if every linear sum of its entries has a normal distribution. 
Thus, to prove \ref{main theorem}, we show that for every set of real coefficients $\alpha_i^r,\alpha_i^c\in \R$, assuming at least one coefficient is non-zero, the sum 
\begin{align}
\label{eq:lin_sum}
\frac{1}{\sqrt{n}}\alpha_0^rX^{r(0)}_n + \frac{1}{\sqrt{n}}\sum_{i=1}^m \parenv{\alpha_i^r X^{r(i)}_n + \alpha_i^c X^{c(i)}_n} + \frac{1}{\sqrt{n}} \sum_{i=2 m+1}^{d^k-1} \alpha_i^r X^{r(i)}_n
\end{align}
converges in distribution to a random variable $Z$, with a characteristic function $\E[e^{(-\frac{1}{2}\sigma^2 t^2)}]$, where $\sigma^2$ is a non-negative constant. In case $\sigma^2 = 0$, the sum converges to $0$ in probability. 
In the next section, we show how to find $\sigma^2$ and find a condition to ensure $\sigma^2 > 0$.

\begin{example}
\label{ex:run_ex4} 
Consider our running example as in Example \ref{ex:run_ex3} and fix $k=2$. 
Recall that the mutation law is a $3/2$-average mutation law. 
Since $\bbM^{(3)}$ is irreducible, Theorem \ref{th:main3} implies that $\bfr_n^{(2k-1)}$ converges in probability to a positive vector $\bbr^{(2k-1)}=\bbr^{(3)}>0$. 
Moreover, $\bbM^{(2)}$ is also irreducible, and the eigenvalues of $\bbM^{(2)}$ are 
\[\lambda_0 = \frac{5}{2},~ \lambda_1 = \frac{3}{2},~ \lambda_2 = 1,~ \lambda_3 = 0.\] 
Notice that $\lambda_0=\tau+k-1=3/2+2-1$ and that every non-leading eigenvalue satisfies the inequality $\lambda^r < k + \frac{1}{2}(\tau-1)$. 
The left eigenvectors of $\bbM^{(2)}$ are 
\[\bbu_0 = \bmat{1&1&1&1},~ \bbu_1 = \bmat{-1&-1&1&1},~ \bbu_2 = \bmat{0&-1&1&0},~ \bbu_3 = \bmat{1&-\frac{3}{2}&-\frac{3}{2}&1}.\]

Thus $\bbU$ is equal to:
\[\bbU = \begin{bmatrix}
    1 &1 &1 &1\\
    -1 &-1 &1 &1\\
    0 &-1 &1 &0\\
    1& -\frac32& -\frac32 &1
\end{bmatrix}.\]
From Theorem \ref{main theorem}, the projection vector $\bbX_n$ of $\bct^{(2)}_n - n(\tau-1)\bbr^{(2)}$ on $\bbU$ satisfies 
\[ \lim_{n\to\infty}n^{-\frac{1}{2}}\bbX_n \sim \cN(0, \Sigma').\]

Later on we show a way to calculate $\sigma^2$ for the sum  
\[\lim_{n\to\infty} \frac{1}{\sqrt{n}}\sum_{i=0}^3 \alpha_i^r X_n^{r(i)} \sim \cN(0, \sigma^2), \]
for every set of $\alpha_i^r$. By setting different values of $\alpha_i^r$, we would find all the elements of $\Sigma'$ and then find $\Sigma$. From $\Sigma$ we find the asymptotic distribution of $k$-tuples 
\[\frac{1}{\sqrt{n}}(\ct^{(2)}_n-\frac12 n \bmat{\frac{3}{10}, \frac{2}{10}, \frac{2}{10}, \frac{3}{10}}^T) \xrightarrow{d} \cN\parenv{0, \Sigma}\]
\end{example}

\subsection*{Standing Assumptions}
Throughout the paper, we will make the following assumptions. 
\begin{enumerate}
\item $S$ is an average $\tau$-mutation with $\tau>1$, meaning the length of the sequence steadily grows.
\item The matrix $\bbM^{(k)}$ is irreducible and has a regular eigendecomposition ($\mu_a(\lambda)=\mu_g(\lambda)$ for every eigenvalue $\lambda$). This assumption may be removed as explained at the end of the paper. 
\item We assume that $\bfr_n^{(2k-1)}$ converges in probability to some (non-negative) vector $\bbr^{(2k-1)}$. 
\item For every non-leading eigenvalue $\lambda=\lambda^r+\bbi\lambda^c$, the inequality $\lambda^r < k + \frac12(\tau-1)$ holds.
\end{enumerate}

\subsection*{Outline of the Proof}
We now outline the high-level strategy for proving Theorem~\ref{Main Main Theorem}.

\textbf{Step 1: Eigendecomposition and projection (Section~\ref{sec:pre}).}
We construct the transformation matrix $\bbU$ (Definition~\ref{def:U_matrix}) whose rows are the real and imaginary parts of the left eigenvectors of $\bbM^{(k)}$. Applying $\bbU$ to the centered count vector yields the projection vector $\bbX_n = \bbU(\bct_n - n(\tau-1)\bbr)$. Since $\bbU$ is invertible (Lemma~\ref{lem:U_full_rank}), establishing a CLT for $\frac{1}{\sqrt{n}}\bbX_n$ is equivalent to proving the CLT for $\frac{1}{\sqrt{n}}(\bct_n - n(\tau-1)\bbr)$. Thus, it suffices to show that every linear combination of the entries of $\frac{1}{\sqrt{n}}\bbX_n$ is asymptotically normal.

\textbf{Step 2: Martingale difference approximation.}
For each non-leading eigenvalue $\lambda_i$, the projection $X^{(i)}_n = \bbu_i\cdot(\bct_n - n(\tau-1)\bbr)$ satisfies a linear recurrence driven by a martingale difference $M_n^{(i)}$. Because the sequence length $Y_n$ is random, we first approximate $M_n^{(i)}$ by a simpler process $\overline{M}_n^{(i)}$, where the random length $Y_n$ is replaced by its deterministic mean $(n-1)(\tau-1)+Y_0$. The approximation error is shown to be asymptotically negligible in $L^1$ (Claim~\ref{cl:asym.equiv.}). This allows us to express each projection as an explicit linear combination of the terms $\overline{M}_j^{(i)}$, up to a small error term (Claim~\ref{cl:beta_val}). However, the resulting linear combination is no longer martingale.

\textbf{Step 3: Verifying the martingale CLT conditions.}
To leverage the martingale CLT, we take the coefficients derived for the linear combination of $\overline{M}_j^{(i)}$ and apply them instead to the original martingale differences $M_n^{(i)}$. 
We show that these two linear combinations are asymptotically equivalent in $L^1$. 
This substitution yields a martingale sequence. 
We can then apply the classical martingale CLT (Theorem~\ref{th:martingaleCLT}) to the array $\cX_{n,j}$, defined in~\eqref{eq:xi}, using $M_n^{(i)}$.
This requires verifying the following two conditions.
\begin{itemize}
\item \emph{Conditional Lindeberg condition}: The individual terms $|\cX_{n,j}|$ become uniformly negligible relative to $\sqrt{n}$, which follows from the spectral gap assumption $\lambda_i^r < k + \frac{1}{2}(\tau-1)$ (Section~\ref{sec:MCLT}).
\item \emph{Conditional variance condition}: The conditional sum $\frac{1}{n}\sum_{j=1}^n\E[\cX_{n,j}^2\mid\cF_{j-1}]$ converges in probability to a deterministic limit $\sigma^2 \geq 0$ (Claim~\ref{cl:CVC_deg}). This relies on the $L^1$ convergence of the $(2k-1)$-tuple frequency vector (Corollary~\ref{cor:fr_conv_L1}) and on explicitly computing the limits of the Cesàro averages $\frac{1}{n}\sum_{j}\beta_{j,n}^{(i)}\beta_{j,n}^{(l)}$ via Riemann-sum-to-integral approximations (Lemmas~\ref{lem:int_limit},~\ref{lem:finite_well_def}). Establishing this condition requires a preliminary calculation of the asymptotic differences covariance matrix.
\end{itemize}
Completing this step establishes Theorem~\ref{main theorem}.

\textbf{Step 4: Strict positivity (non-degenerate case).} 
Under the supplementary assumption that the distribution of the count increment $\nabla\bct_t^{(k)} - (\tau-1)\bbr$ is non-degenerate, the limiting covariance matrix $\Theta$ of the increments is strictly positive definite (Lemma~\ref{lem:theta_mat_pd}). 
This implies $\sigma^2>0$ for every non-zero choice of coefficients (Claim~\ref{cl:CVC_ndeg}).

\textbf{Step 5: Explicit covariance matrix.} 
By evaluating $\sigma^2$ across all canonical choices of coefficient vectors, we determine the full covariance matrix $\Sigma'$ for the limiting distribution of $\frac{1}{\sqrt{n}}\bbX_n$ (Corollary~\ref{cor:final_cov}). 
We then recover the original covariance matrix via $\Sigma = \bbU^{-1}\Sigma'(\bbU^{-1})^T$. 
This concludes the proof of Theorem~\ref{Main Main Theorem} for the diagonalizable case.

\textbf{Step 6: Non-diagonalizable case.} 
When $\bbM^{(k)}$ is not diagonalizable, we rely on its Jordan decomposition. 
The primary difference is that the coefficients $\beta_{j,n}^{(i,h)}$ acquire additional logarithmic factors of the form $O(n^{(\lambda_i^r-k)/(\tau-1)}\log^{H_i-1}n)$, where $H_i$ denotes the size of the corresponding Jordan block. 
Under our spectral gap assumptions, these factors remain $o(\sqrt{n})$, ensuring that the overall proof strategy remains valid (Section~\ref{sec:non_diag}).

\section{Approximation by Martingale Difference}\label{sec:approx}
In this section, we lay the basis for the proof of Theorem \ref{main theorem}. Since in each mutation step, the length of the sequence increases with random length, the length of the sequence after $n$ mutation steps is random. We denote its length as $Y_n$, i.e., $Y_n=|S(n)|$.
We start by constructing a sequence that is “close” to a martingale difference sequence and captures the behavior of the system. The main reason to do so is that those sequences will be easier to analyze. We follow a similar approach to that of \cite{elishco2024longterm}, using analogous notation and relying on several results established therein. 
Throughout this section, we will drop the $i$ index from expressions like $\bbu_i^r$ and $X^{(i)}_n$, as the following claims are true for all non-leading eigenvalues $1\leq i\leq d^k-1$. Because the analysis for the leading eigenvalue is different, we will keep the index notation in that case.
To ensure the paper is self-contained, we restate the relevant definitions and results from \cite{elishco2024longterm}. 

\begin{definition}
\label{def: nablas}
Fix $k\in \N$. Since the composition of $k$-tuples in the word changes with every mutation step, to capture the change in the composition of  $k$-tuples in step $n$, we define
\begin{align*}
    \nabla Y_n &:= Y_n - Y_{n-1}. \\
    \nabla \bct^{(k)}_n &:= \bct^{(k)}_n - \bct^{(k)}_{n-1}.
\end{align*}
Similarly, we define 
\[\nabla X^{r(0)}_n := X^{r(0)}_n - X^{r(0)}_{n-1} = \bbu_0^r \cdot(\nabla \bct_n^{(k)} - (\tau-1)\bbr) = \nabla Y_n - (\tau - 1)\]
and for non-leading eigenvalues
\begin{align*}
\nabla X^r_n &:= X^r_n - X^r_{n-1} = \bbu^r \cdot(\nabla \bct_n^{(k)} - (\tau-1)\bbr) = \bbu^r \cdot \nabla \bct_n^{(k)} \\
\nabla X^c_n &:= X^c_n - X^c_{n-1} = \bbu^c\cdot (\nabla \bct_n^{(k)} - (\tau-1)\bbr) = \bbu^c \cdot \nabla \bct_n^{(k)} \\
M^r_n &:= \nabla X^r_n - \frac{1}{Y_{n-1}}\parenv{(\lambda^r-k)X^r_{n-1} - \lambda^c X^c_{n-1}}\\
M^c_n &:= \nabla X^c_n - \frac{1}{Y_{n-1}}\parenv{(\lambda^r-k)X^c_{n-1} + \lambda^c X^r_{n-1}}.
\end{align*}
Since the distribution of $\vt$ is finitely supported, the length of the added random words is finite, i.e., $\nabla \bct_n$ is bounded, which implies that $\nabla X^r_n$ and $\nabla X^c_n$ are also bounded. We recall that since the $i$ index is dropped, the definitions above are actually a family of definitions and are defined for every $i=1,2,\dots, d^k-1$.
\end{definition}

The following lemma is a strengthening of a result from \cite{elishco2024longterm}. 
\begin{lemma}
\label{nabla ct}
Define $\cF_n$ as the $\sigma$-algebra generated by $\bct_j,\; j\leq n$, then:
\[\E\sparenv{\nabla\bct_n^{(k)} ~|~ \cF_{n-1}} \xrightarrow{L^1} \parenv{\tau-1}\bbr^{(k)}\]
\end{lemma}

\begin{IEEEproof}
From \cite[Lemma III.9 and Remark III.10]{elishco2024longterm} we have 
\[ \E\sparenv{\nabla\bct_n^{(k)} ~|~ \cF_{n-1}} = \parenv{\bbM^{(k)}-k\bbI}\frac{\bct_{n-1}^{(k)}}{Y_{n-1}} = \parenv{\bbM^{(k)}-k\bbI}\bfr^{(k)}_{n-1}.\]
From Corollary \ref{cor:fr_conv_L1} we obtain 
\[\parenv{\bbM^{(k)}-k\bbI}\bfr^{(k)}_{n-1} \xrightarrow{L^1} \parenv{\bbM^{(k)}-k\bbI}\bbr^{(k)} = \parenv{\tau+k-1 -k}\bbr^{(k)} = \parenv{\tau-1}\bbr^{(k)}\]
\end{IEEEproof}

If $\cF_n$ denotes the $\sigma$-algebra generated by $\bct_j^{(k)},\; j\leq n$, the following claim are proved in \cite[Claim IV.1]{elishco2024longterm}. 
\begin{claim} [See Claim IV.1 in \cite{elishco2024longterm}]
\label{cl:mart_dif}
$M^r_n$ and $M^c_n$ are Martingale differences, with respect to $\cF_n$, with zero mean $E[M^r_n] = E[M^c_n] = 0$.
\end{claim}

As an immediate result, we obtain the following claim. 
\begin{claim}
\label{cl:mart_dif_x0}
    $\nabla X^{r(0)}_n$ is Martingale differences, with respect to $\cF_n$, with zero mean as well.
\end{claim}
\begin{IEEEproof}
As we know from the definition of average $\tau$-Mutation, in every mutation step, every chosen symbol is replaced with a sub-word of average length $\tau$. Meaning:
\[E\sparenv{\nabla X^{r(0)}_n ~|~ \cF_{n-1}} = E\sparenv{\nabla Y_n - (\tau-1) ~|~ \cF_{n-1}} = E\sparenv{\nabla Y_n ~|~ \cF_{n-1}}  - (\tau-1) = \tau-1 - (\tau-1) = 0. \]
\end{IEEEproof}

\begin{definition} \label{def:M hat}
While the processes $M_n^r,M_n^c$ are martingale, the random length $Y_n$ makes them difficult to study. 
Similarly to \cite{elishco2024longterm}, we approximate those processes by replacing $Y_n$ with its mean. Clearly, the obtained processes are no longer martingale. 
Define the approximations
\begin{align*}
    &\overline{M}^r_n := \nabla X^r_n - \frac{1}{(n-1)(\tau-1) + Y_0}\parenv{(\lambda^r-k)X^r_{n-1} - \lambda^c X^c_{n-1}}
    \\
    &\overline{M}^c_n := \nabla X^c_n - \frac{1}{(n-1)(\tau-1) + Y_0}\parenv{(\lambda^r-k)X^c_{n-1} + \lambda^c X^r_{n-1}}
\end{align*}  
with $Y_0$ being the length of the starting word $w$ which is deterministic.
\end{definition}

To claim that we may analyze $\overline{M}_n^r,\overline{M}_n^c$ instead of $M_n^r,M_n^c$, we need to show asymptotic equivalence. 
\begin{claim}
\label{cl:asym.equiv.}
$\overline{M}^r_n$ and $\overline{M}^c_n$ are asymptotically $L^1$ equivalent to $M^r_n$ and $M^c_n$, respectively, i.e., 
\begin{align*}
\lim_{n\to\infty}\E\sparenv{\abs{\overline{M}_n^r-M_n^r}}&=0, \\ 
\lim_{n\to\infty}\E\sparenv{\abs{\overline{M}_n^c-M_n^c}}&=0. 
\end{align*}
\end{claim}
We denote this asymptotic equivalence by  $\overline{M}^r_n \stackrel{L^1}{\sim} M^r_n$ and $\overline{M}^c_n \stackrel{L^1}{\sim} M^c_n$.

\begin{IEEEproof} 
To prove the claim we write 
\[\overline{M}^r_n-M^r_n=\parenv{(\lambda^r-k)X^r_{n-1} - \lambda^c X^c_{n-1}}\parenv{\frac{1}{Y_{n-1}}-\frac{1}{(n-1)(\tau-1) + Y_0}}.\] 
Notice that 
\begin{align*}
&\parenv{(\lambda^r-k)X^r_{n-1} - \lambda^c X^c_{n-1}}\parenv{\frac{1}{Y_{n-1}}-\frac{1}{(n-1)(\tau-1) + Y_0}} \\ 
&= \parenv{(\lambda^r-k)\frac{X^r_{n-1}}{Y_{n-1}} - \lambda^c \frac{X^c_{n-1}}{Y_{n-1}}} \frac{Y_{n-1} - ((n-1)(\tau-1) + Y_0)}{(n-1)(\tau-1) + Y_0}\\ 
&= \parenv{(\lambda^r-k) \bbu_r\frac{\bct_{n-1}}{Y_{n-1}} - \lambda^c \bbu_c \frac{\bct_{n-1}}{Y_{n-1}}} \frac{Y_{n-1} - ((n-1)(\tau-1) + Y_0)}{(n-1)(\tau-1) + Y_0}. 
\end{align*}
Overall, we get 
\begin{align*}
&\E\sparenv{\abs{\overline{M}^r_n-M^r_n}} \\ 
&=\E\sparenv{\abs{ \parenv{(\lambda^r-k) \bbu_r\frac{\bct_{n-1}}{Y_{n-1}} - \lambda^c \bbu_c \frac{\bct_{n-1}}{Y_{n-1}}} \frac{Y_{n-1} - ((n-1)(\tau-1) + Y_0)}{(n-1)(\tau-1) + Y_0} }}.
\end{align*}
From Cauchy-Schwartz inequality, we get 
\begin{align*}
&\E\sparenv{\abs{ \parenv{(\lambda^r-k) \bbu_r\frac{\bct_{n-1}}{Y_{n-1}} - \lambda^c \bbu_c \frac{\bct_{n-1}}{Y_{n-1}}} \frac{Y_{n-1} - ((n-1)(\tau-1) + Y_0)}{(n-1)(\tau-1) + Y_0} }} \\
&\leq \parenv{\E \sparenv{\parenv{(\lambda^r-k) \bbu_r\frac{\bct_{n-1}}{Y_{n-1}} - \lambda^c \bbu_c \frac{\bct_{n-1}}{Y_{n-1}} }^2}  \E \sparenv{\parenv{ \frac{Y_{n-1} - ((n-1)(\tau-1) + Y_0)}{(n-1)(\tau-1) + Y_0}}^2 }}^{1/2}.
\end{align*}
In \cite[Corollary IV.12]{elishco2024longterm} it was shown that 
\begin{align}
\label{eq:fracbound}
\E \sparenv{\parenv{ \frac{Y_{n-1} - ((n-1)(\tau-1) + Y_0)}{(n-1)(\tau-1) + Y_0}}^2 }\leq \frac{C_0}{n-1} 
\end{align}
for some constant $C_0$. 
Plugging this in, we get 
\begin{align*}
&\parenv{\E \sparenv{\parenv{(\lambda^r-k) \bbu_r\frac{\bct_{n-1}}{Y_{n-1}} - \lambda^c \bbu_c \frac{\bct_{n-1}}{Y_{n-1}} }^2}  \E \sparenv{\parenv{ \frac{Y_{n-1} - ((n-1)(\tau-1) + Y_0)}{(n-1)(\tau-1) + Y_0}}^2 }}^{1/2} \\ 
&\leq \parenv{\E \sparenv{\parenv{(\lambda^r-k) \bbu_r\frac{\bct_{n-1}}{Y_{n-1}} - \lambda^c \bbu_c \frac{\bct_{n-1}}{Y_{n-1}} }^2}  \frac{C_0}{n-1}}^{1/2}.
\end{align*}
Since $\frac{\bct_n}{Y_n}=\bfr_n$ is bounded, and since $\bbu$ is fixed, we obtain that 
$\parenv{(\lambda^r-k) \bbu_r\frac{\bct_{n-1}}{Y_{n-1}} - \lambda^c \bbu_c \frac{\bct_{n-1}}{Y_{n-1}} }^2$ is bounded as well. 
Overall, we obtain 
\[\lim_{n\to\infty} \E\sparenv{\abs{\overline{M}_n^r-M_n^r}}\leq \lim_{n\to\infty} \frac{C}{\sqrt{n-1}}=0,\] 
for some constant $C$. 
Thus, 
\[\overline{M}^r_n \stackrel{L^1}{\sim} M^r_n.\]
Similarly, we obtain 
\[\overline{M}^c_n ~\stackrel{L^1}{\sim}~ M^c_n.\]
\end{IEEEproof}

Following similar lines as  \cite{elishco2024longterm}, for fixed $\alpha^r,\alpha^c\in\R$, we seek constants $\beta^r_{j,n},\beta^c_{j,n}$ for $j=1,\dots,n$ such that 
\begin{align}
\label{eq:betas}
\sum^n_{j=1}\parenv{\beta^r_{j,n} \overline{M}^r_j + \beta^c_{j,n}\overline{M}^c_j}= \alpha^r X^r_n + \alpha^c X^c_n + \varepsilon_n
\end{align}
where $\varepsilon_n=O\parenv{n^{\frac{\lambda^r-k}{\tau-1}}}$ is a small error term. 
Note that the leading eigenvalue is not a part of these sums. 
We also define the following sums. 
\begin{align}
\label{eq:v}
v_n &:= \sum^n_{j=1}\parenv{\beta^r_{j,n} M^r_j + \beta^c_{j,n}M^c_j},\\ 
\label{eq:vbar}
\overline{v}_n &:= \sum^n_{j=1} \parenv{\beta^r_{j,n}\overline{M}^r_j + \beta^c_{j,n}\overline{M}^c_j}= \alpha^r X^r_n + \alpha^c X^c_n + \varepsilon_n.
\end{align}

We utilize the set of estimators $\mathset{\overline{M}_j}$ to derive the coefficients $\beta_j$ for the sum $\bar{v}_n = \alpha^r X^r_n + \alpha^c X^c_n + \varepsilon_n$. These $\beta_j$ coefficients are subsequently applied to the martingale process $v_n$ to construct approximations for $X^c_n$ and $X^r_n$. Finally, by invoking the central limit theorem for martingale difference sequences, we establish that $\bbX_n$ converges to an asymptotic normal distribution. 
To that end, we first present the values of $\beta_{j,n}^c,\beta_{j,n}^r$, and then show and asymptotic $L^1$ equivalence between $v_n$ and $\overline{v}_n$. 

\begin{remark}
\label{re:real_lambda_alpha_c}
if $\lambda$ is real, meaning $\lambda^c=0$, its eigenvector is real as well. 
Then $\bbu_c=\0$ and $X^c_n=0$ as well. In that case, for convenience, we set $\alpha^c:=0$. 
\end{remark}

The first step is to find the values of $\beta_{j,n}^r$ and $\beta_{j,n}^c$, together with a simple bound. 
\begin{claim}
\label{cl:beta_val} 
The values of $\beta^r_{j,n}$ and $\beta^c_{j,n}$ are given below. 
\begin{align}
\label{eq:real betaequal}
\beta^r_{j,n} &= \parenv{\frac{n}{j+1}}^{\frac{\lambda^r-k}{\tau-1}} \parenv{\alpha^r\cos\parenv{\frac{\lambda^c}{\tau-1}\log\parenv{\frac{n}{j+1}}} +\alpha^c\sin\parenv{\frac{\lambda^c}{\tau-1}\log\parenv{\frac{n}{j+1}}}}+O\parenv{n^{\frac{\lambda^r-k}{\tau-1}-1}}
\\
\label{eq:complex betaequal}
\beta^c_{j,n} &= \parenv{\frac{n}{j+1}}^{\frac{\lambda^r-k}{\tau-1}} \parenv{\alpha^c\cos\parenv{\frac{\lambda^c}{\tau-1}\log\parenv{\frac{n}{j+1}}} -\alpha^r\sin\parenv{\frac{\lambda^c}{\tau-1}\log\parenv{\frac{n}{j+1}}}}+O\parenv{n^{\frac{\lambda^r-k}{\tau-1}-1}}.
\end{align}
The absolute values of $\beta_{j,n}^c$ and $\beta_{j,n}^r$ are bounded above by 
\begin{align}
\label{eq:real betabound}
\abs{\beta^r_{j,n}} &\le \parenv{|\alpha^r|+|\alpha^c|} \parenv{\frac{n}{j+1}}^{\frac{\lambda^r-k}{\tau-1}} +O\parenv{n^{\frac{\lambda^r-k}{\tau-1}-1}}
\\
\label{eq:complex betabound}
\abs{\beta^c_{j,n}} &\le \parenv{|\alpha^r|+|\alpha^c|} \parenv{\frac{n}{j+1}}^{\frac{\lambda^r-k}{\tau-1}} +O\parenv{n^{\frac{\lambda^r-k}{\tau-1}-1}}
\end{align}

For $\varepsilon_n$, we obtain 
\begin{align*}
    \varepsilon_n &= -X^r_0\parenv{ \parenv{1+\frac{\lambda^r-k}{m}}\beta^r_{1,n}+\frac{\lambda^c}{m}\beta^c_{1,n}}-X^c_0\parenv{ \parenv{1+\frac{\lambda^r-k}{m}}\beta^c_{1,n}-\frac{\lambda^c}{m}\beta^r_{1,n} }\\ 
    &= O\parenv{n^{\frac{\lambda^r-k}{\tau-1}}}.
\end{align*} 
\end{claim}

As the proof follows a methodology similar to that in \cite{elishco2024longterm}, wherein the values for $\beta^r_{j,n}$ and $\beta^c_{j,n}$ were derived for the specific case $\alpha^r=1$ and $\alpha^c=\bbi$, we defer the full derivation to the Appendix.

We now work our way to show an $L^1$ asymptotic equivalence between $v_n$ and $\overline{v}_n$. 
Recall that for each non-leading eigenvalue, we attach such $v_n$ and $\overline{v}_n$ (together with $M_n,\overline{M}_n$). 

We begin with the following useful lemma (recall that convergence in $L^2$ implies convergence in $L^1$). 
\begin{lemma}
\label{x/y L2 convergence} 
Fix a non-leading eigenvalue $\lambda$ and its corresponding $X_n$. We have that 
$\frac{X_n}{Y_n}$ converges to $0$ in $L^2$, i.e., 
\[\lim_{n\to\infty}\E\sparenv{\left|\frac{X_n}{Y_n}\right|^2}=0.\]
\end{lemma}

\begin{IEEEproof} 
We have shown in Theorem \ref{th:main3} that $\fr_n \xrightarrow{\Pr} \bbr$. 
Since the non-leading eigenvector $\bbu$ is orthogonal to $\bbr$, we get 
\[\frac{X_n}{Y_n} = \frac{\bbu \cdot\bct_n}{Y_n} = \bbu \cdot\bfr_n \xrightarrow{\Pr} \bbu \cdot \bbr = 0. \]

Since $\bfr_n$ is a frequency (probability) vector, its entries are bounded in $[0,1]$, and hence $\frac{X_n}{Y_n} = \bbu\cdot\bfr_n$ is uniformly bounded. By Theorem \ref{th:Leb_Vit}, convergence in probability together with uniform boundedness implies $\frac{X_n}{Y_n} \xrightarrow{L^2} 0$.
\end{IEEEproof}

We now show that for every non-leading eigenvalue, $\overline{v}_n$ and $v_n$ are asymptotically similar. 
\begin{claim}
    \label{v to v hat covergence} 
    Let $\lambda=\lambda^r+\bbi\lambda^c$ be a non-leading eigenvalue with $\lambda^r<k+\frac12(\tau-1)$. 
    Let $v_n$ and $\overline{v}_n$ be the sums as in \eqref{eq:v} and \eqref{eq:vbar}, respectively, corresponding to $\lambda$. 
    Then 
    \[\frac{\overline{v}_n}{\sqrt{n}} \stackrel{L^1}{\sim} \frac{v_n}{\sqrt{n}}.\]
\end{claim}

\begin{IEEEproof}
We begin by writing the expected value 
\begin{align}
\E\sparenv{\frac{1}{\sqrt{n}}\abs{\overline{v}_n-v_n}}&= 
\frac{1}{\sqrt{n}}\E\sparenv{\abs{\sum_{j=1}^n \parenv{\beta_{j,n}^r(\overline{M}_j^r-M_j^r)+\beta_{j,n}^c(\overline{M}_j^c-M_j^c)}}}.
\end{align}
From the triangle inequality and the linearity of expectation, we get 
\begin{align}
&\frac{1}{\sqrt{n}}\E\sparenv{\abs{\sum_{j=1}^n \parenv{\beta_{j,n}^r(\overline{M}_j^r-M_j^r)+\beta_{j,n}^c(\overline{M}_j^c-M_j^c)}}}\\
&\leq 
\frac{1}{\sqrt{n}}\sum_{j=1}^n \parenv{
\abs{\beta_{j,n}^r}\E\sparenv{\abs{\overline{M}_j^r-M_j^r}}+\abs{\beta_{j,n}^c}\E\sparenv{\abs{\overline{M}_j^c-M_j^c}}} \label{eq:sumb}.
\end{align}

Similar to \ref{cl:asym.equiv.}, we have 
\begin{align}
\label{eq:realM}
\E\sparenv{\abs{\overline{M}_j^r-M_j^r}}&= 
\E\sparenv{\abs{\parenv{(\lambda^r-k)X^r_{j-1} - \lambda^c X^c_{j-1}}\parenv{\frac{1}{Y_{j-1}} - \frac{1}{(j-1)(\tau-1) + Y_0}}}}, \\
\label{eq:imM}
\E\sparenv{\abs{\overline{M}_j^c-M_j^c}}&=
\E\sparenv{\abs{\parenv{(\lambda^r-k)X^c_{j-1} + \lambda^c X^r_{j-1}}\parenv{\frac{1}{Y_{j-1}} - \frac{1}{(j-1)(\tau-1) + Y_0}}}}. 
\end{align}
Focusing on \eqref{eq:realM}, we get 
\begin{align}
&\E\sparenv{\abs{\parenv{(\lambda^r-k)X^r_{j-1} - \lambda^c X^c_{j-1}}\parenv{\frac{1}{Y_{j-1}} - \frac{1}{(j-1)(\tau-1) + Y_0}}}} \nonumber \\ 
&= \E\sparenv{\abs{\parenv{(\lambda^r-k)\frac{X^r_{j-1}}{Y_{j-1}} - \lambda^c \frac{X^c_{j-1}}{Y_{j-1}}}\parenv{\frac{(j-1)(\tau-1) + Y_0 - Y_{j-1}}{(j-1)(\tau-1) + Y_0}}}}\nonumber \\
&\stackrel{c.s}{\leq} 
\abs{\lambda^r-k}\parenv{\E\sparenv{\abs{\frac{X_{j-1}^r}{Y_{j-1}}}^2}\E\sparenv{\abs{\frac{(j-1)(\tau-1)+Y_0-Y_{j-1}}{(j-1)(\tau-1)+Y_0}}^2}}^{1/2} \nonumber \\ 
&+ \abs{\lambda^c}\parenv{\E\sparenv{\abs{\frac{X_{j-1}^c}{Y_{j-1}}}^2}\E\sparenv{\abs{\frac{(j-1)(\tau-1)+Y_0-Y_{j-1}}{(j-1)(\tau-1)+Y_0}}^2}}^{1/2}.\label{eq:from8}
\end{align}
Similarly, for \eqref{eq:imM} we get 
\begin{align}
&\E\sparenv{\abs{\parenv{(\lambda^r-k)X^c_{j-1} - \lambda^c X^r_{j-1}}\parenv{\frac{1}{Y_{j-1}} - \frac{1}{(j-1)(\tau-1) + Y_0}}}} \nonumber \\
&\stackrel{c.s}{\leq} 
\abs{\lambda^r-k}\parenv{\E\sparenv{\abs{\frac{X_{j-1}^c}{Y_{j-1}}}^2}\E\sparenv{\abs{\frac{(j-1)(\tau-1)+Y_0-Y_{j-1}}{(j-1)(\tau-1)+Y_0}}^2}}^{1/2} \nonumber \\ 
&+ \abs{\lambda^c}\parenv{\E\sparenv{\abs{\frac{X_{j-1}^r}{Y_{j-1}}}^2}\E\sparenv{\abs{\frac{(j-1)(\tau-1)+Y_0-Y_{j-1}}{(j-1)(\tau-1)+Y_0}}^2}}^{1/2}. \label{eq:from9}
\end{align}
From Lemma \ref{x/y L2 convergence}, for every $\epsilon>0$, there exists $N\in\N$ such that for every $j\geq N$, $\E\sparenv{\abs{\frac{X_{j-1}^r}{Y_{j-1}}}^2}<\epsilon$ and $\E\sparenv{\abs{\frac{X_{j-1}^c}{Y_{j-1}}}^2}<\epsilon$. 
In addition, from \eqref{eq:fracbound} we have 
\[\E\sparenv{\abs{\frac{(j-1)(\tau-1)+Y_0-Y_{j-1}}{(j-1)(\tau-1)+Y_0}}^2}\leq \frac{C_0}{j-1}\]
for some constant $C_0$. 
Thus, taking $N$ large enough, for $n,j\geq N$, \eqref{eq:from8} and \eqref{eq:from9}, together with bounding $\abs{\beta_{j,n}^r}$ and $\abs{\beta_{j,n}^c}$ with \eqref{eq:real betabound} and \eqref{eq:complex betabound} we get:
\begin{align}
&\abs{\beta_{j,n}^r}\E\sparenv{\abs{\overline{M}_j^r-M_j^r}}+ \abs{\beta_{j,n}^c}\E\sparenv{\abs{\overline{M}_j^c-M_j^c}} \nonumber \\ 
&\leq 2\parenv{(|\alpha^r|+|\alpha^c|)\parenv{\frac{n}{j+1}}^{\frac{\lambda^r-k}{\tau-1}}+ O\parenv{n^{\frac{\lambda^r-k}{\tau-1}-1}}} \parenv{\abs{\lambda^r-k}+\abs{\lambda^c}} \sqrt{\frac{C_0}{j-1}} \sqrt{\epsilon}.
\end{align}
Notice that for the sum of the first $N$ arguments, $\E\sparenv{\abs{\overline{M}_j^r-M_j^r}}$ and $\E\sparenv{\abs{\overline{M}_j^c-M_j^c}}$ are bounded and $\abs{\beta_{j,n}^r}, \abs{\beta_{j,n}^r} = O\parenv{n^{\frac{\lambda^r-k}{\tau-1}}}$ thus we get:

\[\frac{1}{\sqrt{n}}\sum_{j=1}^N \parenv{
\abs{\beta_{j,n}^r}\E\sparenv{\abs{\overline{M}_j^r-M_j^r}}+\abs{\beta_{j,n}^c}\E\sparenv{\abs{\overline{M}_j^c-M_j^c}}} = \frac{1}{\sqrt{n}} O\parenv{n^{\frac{\lambda^r-k}{\tau-1}}} = O\parenv{n^{\frac{\lambda^r-k}{\tau-1}-\frac12}}\]

Since we take $n\to\infty$, this sum decays to $0$. 
Using the fact that $|\alpha^r|+|\alpha^c|$ is bounded, putting everything together we get: 
\begin{align}
&\lim_{n\to\infty}\frac{1}{\sqrt{n}} \E\sparenv{\abs{\overline{v}_n-v_n}} = \lim_{n\to\infty}\frac{1}{\sqrt{n}}\sum_{j=1}^n \parenv{
\abs{\beta_{j,n}^r}\E\sparenv{\abs{\overline{M}_j^r-M_j^r}}+\abs{\beta_{j,n}^c}\E\sparenv{\abs{\overline{M}_j^c-M_j^c}}}\nonumber\\ 
&\leq \lim_{n\to\infty}\frac{C}{\sqrt{n}} \sqrt{\epsilon}\sum_{j=N}^n 
\frac{1}{\sqrt{j-1}}\parenv{\parenv{\frac{n}{j+1}}^{\frac{\lambda^r-k}{\tau-1}}+ O\parenv{n^{\frac{\lambda^r-k}{\tau-1}-1}}} + O\parenv{n^{\frac{\lambda^r-k}{\tau-1}-\frac12}}
\end{align}
for some constant $C$. 
We may increase the sum and write 
\begin{align}
&\lim_{n\to\infty}\frac{1}{\sqrt{n}} \E\sparenv{\abs{\overline{v}_n-v_n}}\nonumber\\ 
&\leq \lim_{n\to\infty}\frac{C}{\sqrt{n}} \sqrt{\epsilon}\sum_{j=2}^n 
\frac{1}{\sqrt{j-1}}\parenv{\parenv{\frac{n}{j+1}}^{\frac{\lambda^r-k}{\tau-1}}+ O\parenv{n^{\frac{\lambda^r-k}{\tau-1}-1}}} + O\parenv{n^{\frac{\lambda^r-k}{\tau-1}-\frac12}}.
\end{align}
Since $\lambda^r<k+\frac12(\tau-1)$,  
we have $\frac{1}{\sqrt{n}}O\parenv{n^{\frac{\lambda^r-k}{\tau-1}-1}}=o(1/n)$, and thus 
\[\frac{1}{\sqrt{n}}\sum_{j=2}^n\frac{1}{\sqrt{j-1}}O\parenv{n^{\frac{\lambda^r-k}{\tau-1}-1}}\to 0\] 
as $n\to\infty$ (one may use integration to bound the sum and obtain that $\sum_{j=1}^n\frac{1}{\sqrt{j}}=O(\sqrt{n})$).  
Moreover, we have $O\parenv{n^{\frac{\lambda^r-k}{\tau-1}-\frac12}}=o(1)$. 
Thus, 
\begin{align}
\lim_{n\to\infty}\frac{1}{\sqrt{n}} \E\sparenv{\abs{\overline{v}_n-v_n}} 
&\leq 
\lim_{n\to\infty}C n^{\frac{\lambda^r-k}{\tau-1}-\frac{1}{2}} \sqrt{\epsilon}\sum_{j=2}^n 
\frac{1}{\sqrt{j-1}}\parenv{\frac{1}{j+1}}^{\frac{\lambda^r-k}{\tau-1}}\nonumber \\ 
&= \lim_{n\to\infty}C n^{\frac{\lambda^r-k}{\tau-1}-\frac{1}{2}} \sqrt{\epsilon}\sum_{j=2}^n 
\parenv{\frac{1}{j-1}}^{\frac{\lambda^r-k}{\tau-1}+\frac{1}{2}}\parenv{\frac{j-1}{j+1}}^{\frac{\lambda^r-k}{\tau-1}},
\end{align} 
where the last equality follows by multiplying by $\parenv{\frac{j-1}{j-1}}^{\frac{\lambda^r-k}{\tau-1}}$. 
Since $\parenv{\frac{j-1}{j+1}}\to 1$ as $j$ increases, we may disregard it, and get 
\begin{align}
&\lim_{n\to\infty}\frac{1}{\sqrt{n}} \E\sparenv{\abs{\overline{v}_n-v_n}}\nonumber \\ 
&\leq 
\lim_{n\to\infty}C n^{\frac{\lambda^r-k}{\tau-1}-\frac{1}{2}} \sqrt{\epsilon}\sum_{j=2}^n 
\parenv{\frac{1}{j-1}}^{\frac{\lambda^r-k}{\tau-1}+\frac{1}{2}}\nonumber \\ 
&= \lim_{n\to\infty}C n^{\frac{\lambda^r-k}{\tau-1}-\frac{1}{2}} \sqrt{\epsilon}\sum_{j=1}^n 
\parenv{\frac{1}{j}}^{\frac{\lambda^r-k}{\tau-1}+\frac{1}{2}}. \label{eq:sumend}
\end{align}
For $0<t<1$, we have 
\begin{align}
\label{eq:thesum}
\lim_{n\to\infty}\frac{\sum_{j=1}^n\parenv{\frac{1}{j}}^t}{n^{1-t}}&=\frac{1}{1-t}. 
\end{align}
Since $\frac{\lambda^r-k}{\tau-1}<\frac12$ we may write \eqref{eq:sumend} as 
\begin{align}
\lim_{n\to\infty}\frac{1}{\sqrt{n}} \E\sparenv{\abs{\overline{v}_n-v_n}}&\leq 
\lim_{n\to\infty}C n^{\frac{\lambda^r-k}{\tau-1}-\frac{1}{2}} \sqrt{\epsilon}\frac{n^{\frac{1}{2}-\frac{\lambda^r-k}{\tau-1}}}{\frac{1}{2}-\frac{\lambda^r-k}{\tau-1}}\nonumber \\ 
&= \tilde{C}\sqrt{\epsilon} 
\end{align} 
for some constant $\tilde{C}$. 
Since this is true for every $\epsilon>0$, we obtain the wanted result.
\end{IEEEproof}

As an immediate corollary, we obtain the following. 
\begin{corollary}
    \label{t1 to v hat convergence} 
    If for every non-leading eigenvalue $\lambda=\lambda^r+\bbi\lambda^c$ of $\bbM^{(k)}$, $\lambda^r < k+\frac12(\tau-1)$, then
    \begin{align*}
        &\frac{1}{\sqrt{n}}\parenv{\alpha_0^r X^{r(0)}_n + \sum_{i=1}^m \parenv{\alpha_i^r X^{r(i)}_n + \alpha_i^c X^{c(i)}_n} + \sum_{i=2 m+1}^{d^k-1} \alpha_i^r X^{r(i)}_n} \stackrel{L^1}{\sim}
        \\
        &\sum_{j=1}^n\frac{1}{\sqrt{n}}\parenv{\alpha_0^r \nabla X^{r(0)}_j +\sum_{i=1}^m\parenv{\beta_{j,n}^{r(i)}M_j^{r(i)}+\beta_{j,n}^{c(i)}M_j^{c(i)}}+\sum_{i=2m+1}^{d^k-1}\beta_{j,n}^{r(i)}M_j^{r(i)}}
    \end{align*}
\end{corollary}

\begin{IEEEproof} 
The proof follows immediately from Claim \ref{v to v hat covergence}, together with recalling that $\varepsilon_n=O\parenv{n^\frac{\lambda^r-k}{\tau-1}}$, and the assumption that $\frac{\lambda^r-k}{\tau-1}<1/2$. First notice that 
\begin{align*}
\frac{1}{\sqrt{n}} \alpha_0^r X^{r(0)}_n &= \frac{1}{\sqrt{n}}\alpha_0^r \parenv{X^{r(0)}_0 + \sum_{j=1}^n \nabla X^{r(0)}_j} \\ 
&= \frac{1}{\sqrt{n}}\alpha_0^r X^{r(0)}_0 + \sum_{j=1}^n\frac{1}{\sqrt{n}}\alpha_0^r \nabla X^{r(0)}_j \\ 
&\stackrel{L^1}{\sim} \sum_{j=1}^n\frac{1}{\sqrt{n}}\alpha_0^r \nabla X^{r(0)}_j. 
\end{align*}
In addition we notice for every $i$ we have:
\begin{align*}
    \frac{1}{\sqrt{n}}\parenv{\alpha^r X^{r}_n + \alpha^c X^{c}_n} &= \frac{1}{\sqrt{n}}\sum^n_{j=1} \parenv{\beta^r_{j,n}\overline{M}^r_j + \beta^c_{j,n}\overline{M}^c_j} - \frac{1}{\sqrt{n}}\varepsilon_n \\
    &= \frac{1}{\sqrt{n}}\sum^n_{j=1} \parenv{\beta^r_{j,n}\overline{M}^r_j + \beta^c_{j,n}\overline{M}^c_j} - O\parenv{n^{\frac{\lambda^r-k}{\tau-1} - \frac{1}{2}}} \\
    &\stackrel{L^1}{\sim} \frac{1}{\sqrt{n}}\sum^n_{j=1} \parenv{\beta^r_{j,n}\overline{M}^r_j + \beta^c_{j,n}\overline{M}^c_j} \\ 
    &\stackrel{L^1}{\sim} \frac{1}{\sqrt{n}}\sum^n_{j=1} \parenv{\beta^r_{j,n}M^r_j + \beta^c_{j,n}M^c_j}.
\end{align*}
Putting everything together proves the claim.
\end{IEEEproof}

In the next section, we use the martingale central limit theorem to prove Theorem \ref{main theorem}. 

\section{Martingale Central Limit Theorem}\label{sec:MCLT}
In this section, we use the sums $\overline{v}_n$ and $v_n$, together with the properties shown above, to prove the CLT for the projection vector $\frac{1}{\sqrt{n}}\bbX_n$ which implies asymptotic joint normal distribution of $\frac{1}{\sqrt{n}}\bbX_n$. 
We use the following CLT for martingales \cite[Theorem 3.2 and Corollary 3.1]{hall2014martingale}, stated here using our notations and adjusted to our needs. 
\begin{theorem}[Th. 3.2 and Cor. 3.1 in \cite{hall2014martingale}] 
\label{th:martingaleCLT} 
Let $\mathset{S_{n,j},\cF_{n,j}, 1\leq j\leq n, n\geq 1}$ be a zero-mean, square integrable martingale array and let $\mathset{\cX_{n,j}}_{n,j}$ be its corresponding difference sequence. 
Let $\sigma^2$ be an a.s. finite random variable.
Suppose that the following conditions hold: 
\begin{enumerate}
\item The conditional Lindeberg condition, i.e., for every $\epsilon>0$, 
\[\sum_{j=1}^n\E\sparenv{\cX_{n,j}^2 \;\bbI_{\sparenv{|\cX_{n,j}|>\epsilon}} \mid \cF_{n,j-1}}\xrightarrow{\Pr} 0,\] 
as $n\to\infty$, where $\bbI_{\sparenv{\text{condition}}}$ is an indicator function for the condition listed in the square brackets.
\item The conditional variance condition, i.e.,  
\[\sum_{j=1}^n \E\sparenv{\cX_{n,j}^2 \mid \cF_{n,j-1}}\xrightarrow{\Pr}\sigma^2,\] 
as $n\to\infty$. 
\end{enumerate} 
Then $\sum_{j=1}^n \cX_{n,j}\xrightarrow{\Pr} Z$, where $Z$ is a random variable with characteristic function $\E\sparenv{\exp(-\frac{1}{2}\sigma^2 t^2)}$.
\end{theorem}

We use Theorem \ref{th:martingaleCLT} with $\cF_{n,j}$ being the $\sigma$-algebra generated by $\bct_l$ for $l\leq j$, and 
\begin{align}
\label{eq:xi}
\cX_{n,j}=\frac{1}{\sqrt{n}}\parenv{\alpha_0^r \nabla X^{r(0)}_j +\sum_{i=1}^m\parenv{\beta_{j,n}^{r(i)}M_j^{r(i)}+\beta_{j,n}^{c(i)}M_j^{c(i)}}+\sum_{i=2m+1}^{d^k-1}\beta_{j,n}^{r(i)}M_j^{r(i)}}.
\end{align}

Clearly, since $M_j^{r(i)},M_j^{c(i)}$ and $\nabla X_j^{(0)}$ are martingale differences for every $i$, their sum over $i$ is also a martingale difference.

First, notice that proving that $\cX_{n,j}$ satisfies the conditional Lindeberg condition and the conditional variance condition implies that 
\[\frac{1}{\sqrt{n}}\alpha_0X^{r(0)}_j + \frac{1}{\sqrt{n}}\sum_{i=1}^m \parenv{\alpha_i^r X^{r(i)}_n + \alpha_i^c X^{c(i)}_n} + \frac{1}{\sqrt{n}} \sum_{i=2 m+1}^{d^k-1} \alpha_i^r X^{r(i)}_n\] 
converges in distribution to a normally distributed random variable $Z$. 
Indeed, from Corollary \ref{t1 to v hat convergence}, we obtain 
\begin{align}
&\frac{1}{\sqrt{n}}\alpha_0X^{r(0)}_j + \frac{1}{\sqrt{n}}\sum_{i=1}^m \parenv{\alpha_i^r X^{r(i)}_n + \alpha_i^c X^{c(i)}_n} + \frac{1}{\sqrt{n}} \sum_{i=2 m+1}^{d^k-1} \alpha_i^r X^{r(i)}_n \nonumber \\ 
&\stackrel{L^1}{\sim} \frac{1}{\sqrt{n}} \sum_{j=1}^n\parenv{\alpha_0^r \nabla X^{r(0)}_j + \sum_{i=1}^m\parenv{\beta_{j,n}^{r(i)}M_j^{r(i)}+\beta_{j,n}^{c(i)}M_j^{c(i)}}+\sum_{i=2m+1}^{d^k-1}\beta_{j,n}^{r(i)}M_j^{r(i)}} \nonumber \\ 
&\xrightarrow{d} Z. \label{eq:small_main}
\end{align} 

We now show the conditions, starting with the conditional Lindeberg condition. 
\subsection{Conditional Lindeberg condition}
In this section we show the conditional Lindeberg condition. 
Before diving into the proof, we provide a bird's-eye view of the proof. 
At first, we bound above $|M_j^{r(i)}|$ and $|M_j^{c(i)}|$ for every $j$ and $i$. 
We then use this bound to claim that the expected value of a sum of $\beta$s and $M$s goes to $0$. 
Using this, we show that the nonconditional Lindeberg condition follows, and that in our case, the unconditional Lindeberg condition is equivalent to the conditional Lindeberg condition. 

We begin with the upper bound on $|M_j^{r(i)}|$ and $|M_j^{c(i)}|$. 
\begin{lemma}
\label{lem:M_is_bounded}
Let $\lambda_i$ be a non-leading eigenvalue. Then there exists a constant $C_1$ such that for every $j$, $|M_j^{r(i)}|<C_1$ and $|M_j^{c(i)}|<C_1$. 
\end{lemma}

\begin{IEEEproof}
Fix $i$, by definition we have that $\nabla X^{r(i)}_j$, $\nabla X^{c(i)}_n$ and $\bfr_j$ are bounded. For simplicity, since $i$ is fixed we remove it from the notation and write $X^{r}_j,X^{c}_j,M_j^r,M_j^c$. 
Thus:
 \begin{align*}
  \left|M^r_j\right| &= \abs{\nabla X^r_j - \frac{(\lambda^r-k)X^r_{j-1} - \lambda^c X^c_{j-1}}{Y_{j-1}}} \\ 
  &= \abs{\nabla X^r_j - \frac{(\lambda^r-k)\bbu_r \bct_{j-1} - \lambda^c \bbu_c \bct_{j-1}}{Y_{j-1}}}\\ 
  &=\abs{\nabla X^r_j - \parenv{(\lambda^r-k)\bbu_r - \lambda^c \bbu_c }\bfr_{j-1}} \\ 
  &\le C
 \end{align*}
 for some constant $C$. Similar arguments also imply that $|M^c_n|<C'$ for some constant $C'$. We now take $C_1^{(i)}=\max\mathset{C,C'}$. 
 Since this is true for every $i$, we define $C_1$ to be the maximum of $C_1^{(i)}$ over all $i$. 
\end{IEEEproof}

Next, we show that the sum decays to zero. 
\begin{claim}
\label{cl:square_to_0} 
Assume that for every non-leading eigenvalue $\lambda_i$, the real part satisfies $\lambda_i^r < k + \frac{1}{2}(\tau-1)$. 
Then for $n \to \infty$
\[ \E\sparenv{\max_{1\leq j\leq n}\abs{\frac{1}{\sqrt{n}} \parenv{\alpha_0^r \nabla X^{r(0)}_j + \sum_{i=1}^m 
 \parenv{\beta^{r(i)}_{j,n} M^{r(i)}_j + \beta^{c(i)}_{j,n}M^{c(i)}_j} + \sum_{i=2 m+1}^{d^k-1} \beta^{r(i)}_{j,n} M^{r(i)}_j}}} \to 0.\]
\end{claim}

\begin{IEEEproof}
Fix $j$. We first notice that we can take $\frac{1}{\sqrt{n}}$ out of the expectation. 
From the triangle inequality, we get 
\begin{align*}
&\abs{\alpha_0^r \nabla X^{r(0)}_j + \sum_{i=1}^m 
 \parenv{\beta^{r(i)}_{j,n} M^{r(i)}_j + \beta^{c(i)}_{j,n}M^{c(i)}_j} + \sum_{i=2 m+1}^{d^k-1} \beta^{r(i)}_{j,n} M^{r(i)}_j}\\
 &\leq \abs{\alpha_0}\abs{\nabla X^{r(0)}_j} + \sum_{i=1}^m 
 \parenv{\abs{\beta^{r(i)}_{j,n}} \abs{M^{r(i)}_j} + \abs{\beta^{c(i)}_{j,n}} \abs{M^{c(i)}_j}} + \sum_{i=2 m+1}^{d^k-1} \abs{\beta^{r(i)}_{j,n}} \abs{M^{r(i)}_j}.
\end{align*}
We know $\abs{\nabla X^{r(0)}_j}$ is bounded, thus $\abs{\alpha_0}\abs{\nabla X^{r(0)}_j} \le \bar{C}$ for some $\bar{C}$. 
We now use Claim \ref{lem:M_is_bounded} and obtain 
\begin{align*}
&\sum_{i=1}^m 
 \parenv{\abs{\beta^{r(i)}_{j,n}} \abs{M^{r(i)}_j} + \abs{\beta^{c(i)}_{j,n}} \abs{M^{c(i)}_j}} + \sum_{i=2 m+1}^{d^k-1} \abs{\beta^{r(i)}_{j,n}} \abs{M^{r(i)}_j} \\ 
 &\leq 
 C_1\parenv{\sum_{i=1}^m 
 \parenv{\abs{\beta^{r(i)}_{j,n}} + \abs{\beta^{c(i)}_{j,n}}} + \sum_{i=2 m+1}^{d^k-1} \abs{\beta^{r(i)}_{j,n}}}.
\end{align*}
Using the bounds \eqref{eq:real betabound} and \eqref{eq:complex betabound} on $\beta_{j,n}$, we get 
\begin{align}
 &C_1\parenv{\sum_{i=1}^m 
 \parenv{\abs{\beta^{r(i)}_{j,n}} + \abs{\beta^{c(i)}_{j,n}}} + \sum_{i=2 m+1}^{d^k-1} \abs{\beta^{r(i)}_{j,n}}} \nonumber \\ 
 &\leq 
 C_1\parenv{\sum_{i=1}^m 2\parenv{|\alpha_i^r| +|\alpha_i^c|}\parenv{\frac{n}{j+1}}^{\frac{\lambda_i^r-k}{\tau-1}} + \sum_{i=2 m+1}^{d^k-1} |\alpha_i^r|\parenv{\frac{n}{j+1}}^{\frac{\lambda_i^r-k}{\tau-1}}+\sum_{i=1}^{d^k-1}O\parenv{n^{\frac{\lambda_i^r-k}{\tau-1}-1}}}. \label{eq:up_to_here}
\end{align}
Since for every $i$, we have $\lambda_i^r < k + \frac{1}{2}(\tau-1)$, there exists some $\epsilon>0$ such that for all $i$, 
\[\parenv{\frac{n}{j+1}}^{\frac{\lambda_i^r-k}{\tau-1}}=O\parenv{n^{\frac{1}{2}-\epsilon}}.\] 
Thus, we get 
\begin{align*}
 &\abs{\alpha_0}\abs{\nabla X^{r(0)}_j} +C_1\parenv{\sum_{i=1}^m 2\parenv{|\alpha_i^r| +|\alpha_i^c|}\parenv{\frac{n}{j+1}}^{\frac{\lambda_i^r-k}{\tau-1}} + \sum_{i=2 m+1}^{d^k-1} |\alpha_i^r|\parenv{\frac{n}{j+1}}^{\frac{\lambda_i^r-k}{\tau-1}}+\sum_{i=1}^{d^k-1} O\parenv{n^{\frac{\lambda_i^r-k}{\tau-1}-1}}} \\ 
 &\leq \bar{C} + C_1 \parenv{O\parenv{n^{\frac{1}{2}-\epsilon}}}\\ 
 &= O\parenv{n^{\frac{1}{2}-\epsilon}}.
\end{align*}
Putting everything together, we get 
\begin{align*}
&\E\sparenv{\max_{1\leq j\leq n}\abs{\frac{1}{\sqrt{n}} \parenv{\alpha_0^r \nabla X^{r(0)}_j + \sum_{i=1}^m 
 \parenv{\beta^{r(i)}_{j,n} M^{r(i)}_j + \beta^{c(i)}_{j,n}M^{c(i)}_j} + \sum_{i=2 m+1}^{d^k-1} \beta^{r(i)}_{j,n} M^{r(i)}_j}}}\\ 
 &\leq 
 \frac{1}{\sqrt{n}}\E\sparenv{\max_{1\leq j\leq n} O\parenv{n^{\frac12-\epsilon}}}\\ 
 &\leq O\parenv{n^{-\epsilon}}\\ 
 &\to 0,
 \end{align*}
and the claim follows. 
\end{IEEEproof}

We now show that Claim \ref{cl:square_to_0} implies the conditional Lindeberg condition. 
Let us denote 
\[\cX_{n,j}=\frac{1}{\sqrt{n}}\parenv{\alpha_0^r \nabla X^{r(0)}_j + \sum_{i=1}^m\parenv{\beta_{j,n}^{r(i)}M_j^{r(i)}+\beta_{j,n}^{c(i)}M_j^{c(i)}}+\sum_{i=2m+1}^{d^k-1}\beta_{j,n}^{r(i)}M_j^{r(i)}}.\] 
First, notice that from Markov's inequality 
\begin{align*}
\Pr\parenv{\max_{1\leq j\leq n} |\cX_{n,j}|>\epsilon}\leq \frac{\E\sparenv{\max_{1\leq j\leq n} |\cX_{n,j}|}}{\epsilon}.
\end{align*}
Thus, for every $\epsilon>0$, as $n\to\infty$, 
\begin{align}
\label{eq:probto0}
\Pr\parenv{\max_{1\leq j\leq n} |\cX_{n,j}|>\epsilon}\to 0. 
\end{align}

For every $n$, denote also $T_n=\sum_{j=1}^n\cX_{n,j}^2$ and we claim that $T_n$ is uniformly bounded under $L^1$, meaning $\sup_n \E\sparenv{|T_n|}<\infty$, which implies that $\mathset{T_n}_n$ is uniformly integrable. This, will be used to show the unconditional Lindeberg condition and then the conditional Lindeberg condition.
Following the same calculation as in the proof of Claim \ref{cl:square_to_0} until \eqref{eq:up_to_here}, we get 
\begin{align*}
&\sup_n \E\sparenv{\abs{\sum_{j=1}^n\cX_{n,j}^2}} \le \sup_n \sum_{j=1}^n\E\sparenv{\abs{\cX_{n,j}^2}} \\ 
&\leq \sup_n \frac{1}{n}\sum_{j=1}^n\abs{\parenv{\bar{C} + C \sum_{i=1}^{d^k-1}\parenv{\frac{n}{j+1}}^{\frac{\lambda_i^r-k}{\tau-1}}+O\parenv{n^{\frac{\lambda_i^r-k}{\tau-1}-1}}}^2},
\end{align*} 
for some constant $C$. 
We may further increase the sum by selecting the non-leading eigenvalue with the maximal real part. 
Denote $\lambda^r_{max}=\max_{i=1,2,\dots,(d^k-1)} \lambda_i^r$, we have 
\begin{align*}
&\sup_n \frac{1}{n}\sum_{j=1}^n\abs{\parenv{\bar{C} + C \sum_{i=1}^{d^k-1}\parenv{\frac{n}{j+1}}^{\frac{\lambda_i^r-k}{\tau-1}}+O\parenv{n^{\frac{\lambda_i^r-k}{\tau-1}-1}}}^2}
\\ 
&\leq \sup_n \frac{1}{n}\sum_{j=1}^n\abs{\parenv{\bar{C} + C (d^k-1)\parenv{\frac{n}{j+1}}^{\frac{\lambda_{max}^r-k}{\tau-1}}+O\parenv{n^{\frac{\lambda_{max}^r-k}{\tau-1}-1}}}^2}
\\
&\leq \sup_n \parenv{C (d^k-1)}^2\frac{1}{n}\sum_{j=1}^n\abs{\parenv{\frac{\bar{C}}{C (d^k-1)} + \parenv{\frac{n}{j+1}}^{\frac{\lambda_{max}^r-k}{\tau-1}}+O\parenv{n^{\frac{\lambda_{max}^r-k}{\tau-1}-1}}}^2}
\\
&= \sup_n C_2 \frac{1}{n}\sum_{j=1}^n \parenv{\bar{C}_1 + \parenv{\frac{n}{j+1}}^{\frac{\lambda_{max}^r-k}{\tau-1}}+O\parenv{n^{\frac{\lambda_{max}^r-k}{\tau-1}-1}}}^2, 
\end{align*}
where $C_2=\parenv{C (d^k-1)}^2$ and $\bar{C}_1=\frac{\bar{C}}{C(d^k-1)}$ are constants. 
We fix $n$ and obtain (ignoring the constant $C_2$ at the beginning) 
\begin{align*}
    &\frac{1}{n}\sum_{j=1}^n \parenv{\bar{C}_1 + \parenv{\frac{n}{j+1}}^{\frac{\lambda_{max}^r-k}{\tau-1}}+O\parenv{n^{\frac{\lambda_{max}^r-k}{\tau-1}-1}}}^2 =
    \\ 
    &\bar{C}_1^2 + \frac{1}{n}\sum_{j=1}^n \parenv{\frac{n}{j+1}}^{2\frac{\lambda_{max}^r-k}{\tau-1}} + \frac{2}{n}\bar{C}_1\sum_{j=1}^n \parenv{\frac{n}{j+1}}^{\frac{\lambda_{max}^r-k}{\tau-1}} + \frac{1}{n}\sum_{j=1}^n O\parenv{n^{2\frac{\lambda_{max}^r-k}{\tau-1}-1}}. 
\end{align*}
We recall that $\frac{\lambda^r-k}{\tau-1}<\frac{1}{2}$, thus for some small $\varepsilon>0$ and constant $C_3$, we use \eqref{eq:thesum} to bound the expression above  by 
\begin{align*}
    & \bar{C}_1^2 + \frac{1}{n}\sum_{j=1}^n \parenv{\frac{n}{j+1}}^{1-2\varepsilon} + \frac{2}{n}\bar{C}_1\sum_{j=1}^n \parenv{\frac{n}{j+1}}^{\frac{1}{2}-\varepsilon} + O(n^{-\varepsilon})\\ 
    &\le \bar{C}_1^2 + \frac{1}{2\varepsilon} + \frac{2}{\varepsilon + \frac{1}{2}}\bar{C}_1 + C_3 \\ 
    &= O(1).
\end{align*}
Plugging this into the original expression, we get 
\begin{align}
   &\sup_n C_2 \frac{1}{n}\sum_{j=1}^n \parenv{\bar{C} + \parenv{\frac{n}{j+1}}^{\frac{\lambda_{max}^r-k}{\tau-1}}+O\parenv{n^{\frac{\lambda_{max}^r-k}{\tau-1}-1}}}^2 \nonumber \\
   &\leq C_2\parenv{\bar{C}_1^2 + \frac{1}{2\varepsilon} + \frac{2}{\varepsilon + \frac{1}{2}}\bar{C}_1 + C_3}\nonumber \\ 
   &= O(1).\label{eq:uni_int}
\end{align}

Thus, $\mathset{T_n}_n$ is uniformly bounded in $L^1$, which implies that it is uniformly integrable. 
By the definition of uniform integrability, 
\[\lim_{L\to \infty} \parenv{\sup_n\E\sparenv{|T_n| \bbI_{[|T_n|>L]}}}=0.\] 
For any $L$ and $\epsilon$, we notice that  
\begin{align*}
&\sum_{j=1}^n\E\sparenv{\cX_{n,j}^2\bbI_{[|\cX_{n,j}|>\epsilon]}}\\ 
&\leq \E\sparenv{\sum_{j=1}^n\cX_{n,j}^2\bbI_{[\sum_{j=1}^n\cX_{n,j}^2>L]}}+L\Pr\parenv{\max_{1\leq j\leq n} |\cX_{n,j}|>\epsilon}.
\end{align*}
Indeed, if $\sum_{j=1}^n\cX_{n,j}^2\leq L$, then for all $j$, $\cX_{n,j}^2\leq L$, and 
\begin{align*}
&\sum_{j=1}^n\E\sparenv{\cX_{n,j}^2\bbI_{[|\cX_{n,j}|>\epsilon]}}\\ 
&\leq \E\sparenv{\sum_{j=1}^n \cX_{n,j}^2 \bbI_{[\max_j |\cX_{n,j}|>\epsilon]}}\\ 
&\leq L\E\sparenv{\bbI_{[\max_j |\cX_{n,j}|>\epsilon]}}\\
&= L\Pr\parenv{\max_{1\leq j\leq n} |\cX_{n,j}|>\epsilon}.
\end{align*} 
Else, we have $\sum_{j=1}^n\cX_{n,j}^2> L$ and we have 
\begin{align*}
&\sum_{j=1}^n\E\sparenv{\cX_{n,j}^2\bbI_{[|\cX_{n,j}|>\epsilon]}}\\
&\leq \E\sparenv{\sum_{j=1}^n\cX_{n,j}^2}\\
&\leq \E\sparenv{\sum_{j=1}^n\cX_{n,j}^2\bbI_{[\sum_{j=1}^n\cX_{n,j}^2>L]}}\\ 
&\leq \E\sparenv{\sum_{j=1}^n\cX_{n,j}^2\bbI_{[\sum_{j=1}^n\cX_{n,j}^2>L]}}+L\Pr\parenv{\max_{1\leq j\leq n} |\cX_{n,j}|>\epsilon}.
\end{align*}

Let $\delta>0$. From uniform integrability, select $L$ large enough so that 
\[\E\sparenv{\sum_{j=1}^n\cX_{n,j}^2\bbI_{[\sum_{j=1}^n\cX_{n,j}^2>L]}}\leq \delta\] 
and then use \eqref{eq:probto0} to obtain that for large $n$,  
\[L\cdot \Pr\parenv{\max_{1\leq j\leq n} |\cX_{n,j}|>\epsilon}<\delta. \] 
We get 
\[\sum_{j=1}^n\E\sparenv{\cX_{n,j}^2\bbI_{[|\cX_{n,j}|>\epsilon]}}<2\delta,\] 
so for every $\epsilon>0$, 
\[\sum_{j=1}^n\E\sparenv{\cX_{n,j}^2\bbI_{[|\cX_{n,j}|>\epsilon]}}\to 0\]
which is the \underline{unconditional} Lindeberg condition. 

To show that the conditional Lindeberg condition holds, we use \cite[Theorem 2.23]{hall2014martingale}, which claims that if $T'_n=\sum_{j=1}^n\E\sparenv{\cX_{n,j}^2 ~|~ \cF_{j-1}}$ is uniformly integrable, then the conditional Lindeberg condition is equivalent to the unconditional Lindeberg condition. To show uniform integrability, we show (again) that $\sup_n \E\sparenv{\abs{T'_n}}<\infty$. We have 
\begin{align*}
&\sup_n \E\sparenv{\abs{T'_n}} \\ 
&= \sup_n \E\sparenv{\abs{\sum_{j=1}^n\E\sparenv{\cX_{n,j}^2 ~|~ \cF_{j-1}}}}
\end{align*}
From the triangle inequality and Jensen's inequality, we get 
\begin{align*}
&\le \sup_n \E\sparenv{\sum_{j=1}^n\abs{\E\sparenv{\cX_{n,j}^2 ~|~ \cF_{j-1}}}}\\ 
&\le \sup_n \E\sparenv{\sum_{j=1}^n\E\sparenv{\abs{\cX_{n,j}^2} ~|~ \cF_{j-1}}}\\
&=\sup_n \sum_{j=1}^n\E\sparenv{\E\sparenv{\abs{\cX_{n,j}^2 }~|~ \cF_{j-1}}}. 
\end{align*}
From the law of total expectation and \eqref{eq:uni_int}, we obtain
\begin{align*}
&= \sup_n \sum_{j=1}^n\E\sparenv{\abs{\cX_{n,j}^2}} \\ 
&\le C_2 \parenv{\bar{C}_1^2 + \frac{1}{2\varepsilon} + \frac{2}{\varepsilon + \frac{1}{2}}\bar{C}_1 + C_3}\\ 
&= O\parenv{1}.
\end{align*}
Thus $T'_n$ is uniformly bounded and uniformly integrable, meaning the conditional Lindeberg condition holds.

Next, we work our way to show the conditional variance condition. 
To that end, we require the notion of the asymptotic differences covariance matrix, explained in the next subsection. 

\subsection{The Asymptotic Differences Covariance Matrix}
In this subsection, we derive several auxiliary results needed to verify the conditional variance condition in the following section.
The primary objective is to determine the asymptotic covariance matrix of the increments of the counting vector, $\nabla \bct_{n}$.
We begin by computing, for $k$-tuples $m_a,m_b\in\cA^k$, the conditional second moment
\[\E\sparenv{\nabla \ct^{(k)}_{n}(m_a)\nabla \ct^{(k)}_{n}(m_b)~\mid~ \cF_{n-1}},\] 
where $\cF_{n-1}$ is the $\sigma$-algebra generated by $\bct^{(k)}_j$ for $j\leq n-1$. 

We begin with a lemma regarding the limit of the above expression. 
\begin{lemma}
\label{lem:theta_lim_val}
Let $S=(\cA,w,(\vt,\bbP))$ be an average $\tau$-mutation system over an alphabet $\cA$ of size $d$. 
Fix $k\in \N$ and let $m_a,m_b\in \cA^k$. Let $\bbr^{(2k-1)}$ denote the asymptotic frequency (probability) vector of $(2k-1)$-tuples, which is the right Perron eigenvector of $\bbM^{(2k-1)}$ (from Corollary \ref{cor:fr_conv_L1}). Then 
\begin{align*}
\E\sparenv{\nabla \ct^{(k)}_{n}(m_a)\nabla \ct^{(k)}_{n}(m_b)~\mid~ \cF_{n-1}}\xrightarrow{L^1} \theta_{m_a,m_b}, 
\end{align*} 
where 
\begin{align*}
\theta_{m_a,m_b}&:=\sum_{v\in\cA^{2k-1}}\sum_{\eta\in \supp(\vt(v_{k-1}))}
\bbr^{(2k-1)}(v) \Pr(\vt(v_{k-1})=\eta) \\ 
&\parenv{\bct^{(k)}_{v_0^{k-2} \eta v_k^{2k-2}}(m_a)-\bct^{(k)}_{v}(m_a)}\parenv{\bct^{(k)}_{v_0^{k-2} \eta v_k^{2k-2}}(m_b)-\bct^{(k)}_{v}(m_b)}.
\end{align*}
\end{lemma}
Notice that for every $m_a,m_b\in\cA^k$, $\theta_{m_a,m_b}$ is bounded above. 

To prove Lemma \ref{lem:theta_lim_val}, we first show the following. 
\begin{lemma}
\label{lem:for_theta_mat} 
Fix $k$ and let $\bbM^{(2k-1)}$ denote the substitution matrix for $(2k-1)$-tuples. 
Assume $\bbM^{(2k-1)}$ is irreducible and for every eigenvalue $\lambda$ of $\bbM^{(2k-1)}$ with $\mu_a(\lambda)>\mu_g(\lambda)$, the real part of $\lambda$ satisfies $\lambda^r<2k-1$ (this property allows us to use Corollary \ref{cor:fr_conv_L1}). 
Let $\bbr^{(2k-1)}$ denote the right Perron eigenvector of the matrix $\bbM^{(2k-1)}$, and let $\cF_n$ denote the $\sigma$-algebra generated by $\bct^{(k)}_j,\; j\leq n$ (generated by the count vectors of $k$-tuples). 
Then 
\[\E\sparenv{\bfr^{(2k-1)}_n ~\mid~ \cF_n}\xrightarrow{L^1}\bbr^{(2k-1)}.\]
We remark that for the lemma to hold, it suffices to assume that $\bbr^{(2k-1)}$ is the asymptotic frequency vector of $(2k-1)$-tuples,
i.e.,
\[\bfr^{(2k-1)}_n \xrightarrow{L^1} \bbr^{(2k-1)}.\]

\end{lemma}

\begin{IEEEproof}
First, since the real part of $\lambda$ satisfies $\lambda^r<2k-1$ we may use Corollary \ref{cor:fr_conv_L1}).
From the property on the non-leading eigenvalues of $\bbM^{(2k-1)}$, we can use Corollary \ref{cor:fr_conv_L1}, obtaining 
\[\bfr^{(2k-1)}_n \xrightarrow{L^1} \bbr^{(2k-1)}\] 
or equivalently,
\[\lim_{n\to\infty} \E\sparenv{\abs{\bfr^{(2k-1)}_n - \bbr^{(2k-1)}}} = 0.\] 
We then bound the quantity of interest as follows:
\begin{align*}
&\E\sparenv{\abs{\E\sparenv{\bfr^{(2k-1)}_n ~\mid~ \cF_{n}} - \bbr^{(2k-1)}}}\\ 
&= \E\sparenv{\abs{\E\sparenv{\bfr^{(2k-1)}_n - \bbr^{(2k-1)} ~\mid~ \cF_{n}} }}\\ 
&\stackrel{(a)}{\leq} \E\sparenv{\E\sparenv{\abs{\bfr^{(2k-1)}_n- \bbr^{(2k-1)}} ~\mid~ \cF_{n} }}\\
&\stackrel{(b)}{=} \E\sparenv{\abs{\bfr^{(2k-1)}_n- \bbr^{(2k-1)}}}
\end{align*}
where $(a)$ follows from Jensen inequality and $(b)$ follows from the law of total expectation. 
Taking the limit and using Corollary \ref{cor:fr_conv_L1} we have 
\[ \lim_{n\to\infty} \E\sparenv{\abs{\bfr^{(2k-1)}_n- \bbr^{(2k-1)}}}=0\] 
which implies the wanted result 
\[\E\sparenv{\bfr^{(2k-1)}_n ~\mid~ \cF_n}\xrightarrow{L^1}\bbr^{(2k-1)}.\]
\end{IEEEproof}

We can now prove Lemma \ref{lem:theta_lim_val}. 
\begin{IEEEproof}[Proof of Lemma \ref{lem:theta_lim_val}]
The proof follows from a straightforward calculation. 
Let $m_a,m_b\in\cA^k$, let $\cF_n$ be the $\sigma$-algebra generated by $\bct^{(k)}_j,\; j\leq n$ and let $\cF_n^{(2k-1)}$ be the $\sigma$-algebra generated by $\bct^{(2k-1)}_j,\; j\leq n$. 
Using the law of total expectation we get 
\[\E\sparenv{\nabla\bct^{(k)}_n(m_a) \nabla\bct^{(k)}_n(m_b) ~\mid~ \cF_{n-1}}= 
\E\sparenv{\E\sparenv{\nabla\bct^{(k)}_n(m_a) \nabla\bct^{(k)}_n(m_b)~\mid~ \cF^{(2k-1)}_{n-1} } ~\mid~ \cF_{n-1}}.\]
Explicitly writing the inner expectation, we have 
\begin{align*}
&\E\sparenv{\E\sparenv{\nabla\bct^{(k)}_n(m_a) \nabla\bct^{(k)}_n(m_b)~\mid~ \cF^{(2k-1)}_{n-1}} ~\mid~ \cF_{n-1}}\\ 
&= \E\left[\sum_{v\in\cA^{2k-1}}\sum_{\eta\in \supp(\vt(v_{k-1}))}
\fr^{(2k-1)}_{n-1}(v)\Pr(\vt(v_{k-1})=\eta) \right. \\ 
& \quad \left.\vphantom{\sum_{v\in\cA^{2k-1}}}\parenv{\bct^{(k)}_{v_0^{k-2} \eta v_k^{2k-2}}(m_a)-\bct^{(k)}_{v}(m_a)}\parenv{\bct^{(k)}_{v_0^{k-2} \eta v_k^{2k-2}}(m_b)-\bct^{(k)}_{v}(m_b)}~|~ \cF_{n-1}\right].  
\end{align*}
Using linearity of expectation we obtain 
\begin{align*}
    &\sum_{v\in\cA^{2k-1}}\sum_{\eta\in \supp(\vt(v_{k-1}))}\E\sparenv{\fr^{(2k-1)}_{n-1}(v)~\mid~ \cF_{n-1}}\Pr(\vt(v_{k-1})=\eta) \\ 
    &\phantom{\sum_{v\in\cA^{2k-1}}}\parenv{\bct^{(k)}_{v_0^{k-2} \eta v_k^{2k-2}}(m_a)-\bct^{(k)}_{v}(m_a)}\parenv{\bct^{(k)}_{v_0^{k-2} \eta v_k^{2k-2}}(m_b)-\bct^{(k)}_{v}(m_b)},
\end{align*}
which, together with Lemma \ref{lem:for_theta_mat}, concludes the proof.
\end{IEEEproof}

\begin{definition}
\label{def:Theta_matrix}
We now use $\theta_{m_a,m_b}$ to define the \textbf{asymptotic differences correlation matrix} $\overline{\Theta}$ 
\begin{align*}
\E\sparenv{\nabla \bct^{(k)}_n \parenv{\nabla \bct^{(k)}_n}^T ~|~ \cF_{n-1}} \xrightarrow{L^1} \overline{\Theta} &:= \bmat{
    \theta_{m_0,m_0}& \theta_{m_0,m_1}& \theta_{m_0,m_2}& \dots& \theta_{m_0,m_{d^k-1}}\\
    \theta_{m_1,m_0}& \theta_{m_1,m_1}& \theta_{m_1,m_2}& \dots& \theta_{m_1,m_{d^k-1}}\\
    \vdots& \vdots& \vdots& \ddots& \vdots\\
    \theta_{m_{d^k-1},m_0}& \theta_{m_{d^k-1},m_1}& \theta_{m_{d^k-1},m_2}& \dots& \theta_{m_{d^k-1},m_{d^k-1}}},
\end{align*}
where $m_0,m_1,\dots,m_{d^k-1}$ is a lexicographic enumeration of all the $k$-tuples $\cA^k$. 
We can use the correlation matrix to define the \textbf{asymptotic differences covariance matrix} $\Theta$
\[\Theta := \lim_{n\to\infty} \E\sparenv{\parenv{\nabla \bct^{(k)}_n-(\tau-1)\bbr} \parenv{\nabla \bct^{(k)}_n-(\tau-1)\bbr}^T ~|~ \cF_{n-1}}\]
and use the correlation matrix to calculate its value
\begin{align*}
    &\E\sparenv{\parenv{\nabla \bct^{(k)}_n-(\tau-1)\bbr} \parenv{\nabla \bct^{(k)}_n-(\tau-1)\bbr}^T ~|~ \cF_{n-1}}
    \\
    &=\E\sparenv{\nabla \bct^{(k)}_n \parenv{\nabla \bct^{(k)}_n}^T ~|~ \cF_{n-1}} + (\tau-1)^2\bbr\bbr^T - \E\sparenv{\nabla \bct^{(k)}_n ~|~ \cF_{n-1}} (\tau-1)\bbr^T - (\tau-1)\bbr\E\sparenv{(\nabla \bct^{(k)}_n)^T ~|~ \cF_{n-1}}. 
\end{align*}
We use \cite[Claim I.V13]{elishco2024longterm} to calculate the limit of the conditional mean of $\nabla \bct^{(k)}_n$ 
\begin{align*}
\E\sparenv{\nabla \bct^{(k)}_n ~|~ \cF_{n-1}} &= \frac{1}{Y_{n-1}}\parenv{\bbM^{(k)} - k\bbI}\bct^{(k)}_n \\ 
&= \parenv{\bbM^{(k)} - k\bbI}\bfr^{(k)}_n \\ 
&\xrightarrow{L^1} \parenv{\bbM^{(k)} - k\bbI}\bbr \\ 
&= \parenv{\tau + k - 1- k}\bbr \\
&= (\tau-1)\bbr.
\end{align*}
Plugging it in our calculations, we get:
\begin{align*}
    &\E\sparenv{\nabla \bct^{(k)}_n \parenv{\nabla \bct^{(k)}_n}^T ~|~ \cF_{n-1}} + (\tau-1)^2\bbr\bbr^T - \E\sparenv{\nabla \bct^{(k)}_n ~|~ \cF_{n-1}} (\tau-1)\bbr^T - (\tau-1)\bbr\E\sparenv{(\nabla \bct^{(k)}_n)^T ~|~ \cF_{n-1}} \\ 
    &\xrightarrow{L^1} \overline{\Theta} + (\tau-1)^2\bbr\bbr^T - 2(\tau-1)^2\bbr\bbr^T = \overline{\Theta} - (\tau-1)^2\bbr\bbr^T \\
    &=\Theta \parenv{:= \overline{\Theta} - (\tau-1)^2\bbr\bbr^T}.
\end{align*}
\end{definition}

Using the definition of $\Theta$ together with Lemma \ref{lem:theta_lim_val}, we obtain that 
\begin{align}
&\E\sparenv{\nabla X^{r(i)}_n \parenv{\nabla X^{r(l)}_n}^T ~|~ \cF_{n-1}} \nonumber \\ 
&= \E\sparenv{\parenv{\bbu_i^r \parenv{\nabla \bct^{(k)}_n-(\tau-1)\bbr}}\parenv{\bbu_l^r \parenv{\nabla \bct^{(k)}_n-(\tau-1)\bbr}}^T ~|~ \cF_{n-1}} \nonumber \\ 
&= \bbu_i^r\E\sparenv{\parenv{\nabla \bct^{(k)}_n-(\tau-1)\bbr}\parenv{\nabla \bct^{(k)}_n-(\tau-1)\bbr}^T ~|~ \cF_{n-1}} \parenv{\bbu_l^r}^T \nonumber \\
&\xrightarrow{L^1} \bbu_i^r \Theta \parenv{\bbu_l^r}^T. \label{eq:L1_lim_XrXr}
\end{align}
Similarly, we obtain 
\begin{align}
\label{eq:L1_lim_XcXc}
\E\sparenv{\nabla X^{c(i)}_n \parenv{\nabla X^{c(l)}_n}^T ~|~ \cF_{n-1}} \xrightarrow{L^1} \bbu_i^c \Theta \parenv{\bbu_l^c}^T, 
\end{align}
and 
\begin{align}
\label{eq:L1_lim_XrXc}
\E\sparenv{\nabla X^{r(i)}_n \parenv{\nabla X^{c(l)}_n}^T ~|~ \cF_{n-1}}\xrightarrow{L^1} \bbu_i^r \Theta \parenv{\bbu_l^c}^T. 
\end{align}

This, together with the definition of $M_n^{r(i)},M_n^{c(i)}$, implies 
\begin{align}
&\E\sparenv{M^{r(i)}_n M^{r(l)}_n ~|~ \cF_{n-1}} \nonumber \\ 
&= \E\sparenv{\nabla X^{r(i)}_n \nabla X^{r(l)}_n ~|~ \cF_{n-1}} -\frac{(\lambda_i^r-k)X^{r(i)}_{n-1} + \lambda_i^c X^{c(i)}_{n-1}}{Y_{n-1}}
\frac{(\lambda_l^r-k)X^{r(l)}_{n-1} + \lambda_l^c X^{c(l)}_{n-1}}{Y_{n-1}}\nonumber \\ 
&\xrightarrow{L^1} \bbu_i^r \Theta \parenv{\bbu_l^r}^T. \label{eq:L1_lim_MrMr}
\end{align}
Similarly, we have 
\begin{align}
\label{eq:L1_lim_McMc}
\E\sparenv{M^{c(i)}_n M^{c(l)}_n ~|~ \cF_{n-1}}\xrightarrow{L^1} \bbu_i^c \Theta \parenv{\bbu_l^c}^T,
\\
\label{eq:L1_lim_MrMc}
\E\sparenv{M^{r(i)}_n M^{c(l)}_n ~|~ \cF_{n-1}}\xrightarrow{L^1} \bbu_i^r \Theta \parenv{\bbu_l^c}^T,
\end{align}
and 
\begin{align}
    &\E\sparenv{\nabla X^{r(0)}_n M^{r(l)}_n ~|~ \cF_{n-1}} \nonumber \\ 
    &= \E\sparenv{\nabla X^{r(0)}_n\parenv{\nabla X^{r(l)}_n - \frac{1}{Y_{n-1}}\parenv{(\lambda^r-k)X^{r(l)}_{n-1} - \lambda^c X^{c(l)}_{n-1}}} ~|~ \cF_{n-1}} \nonumber \\
    &=\E\sparenv{\nabla X^{r(0)}_n \nabla X^{r(l)}_n ~|~ \cF_{n-1}} - \frac{1}{Y_{n-1}}\parenv{(\lambda^r-k)X^{r(l)}_{n-1} - \lambda^c X^{c(l)}_{n-1}} \E\sparenv{\nabla X^{r(0)}_n~|~ \cF_{n-1}} \nonumber \\
    &\stackrel{(a)}{=}\E\sparenv{\nabla X^{r(0)}_n \nabla X^{r(l)}_n ~|~ \cF_{n-1}} \nonumber \\ 
    &\xrightarrow{L^1} \bbu_0 \Theta \parenv{\bbu_l^r}^T,
\end{align}
where $(a)$ follows since $\nabla X^{r(0)}_n$ is martingale difference, thus $\E\sparenv{\nabla X^{r(0)}_n~|~ \cF_{n-1}}=0$.
Similarly
\begin{align}
    \E\sparenv{\nabla X^{r(0)}_n M^{c(l)}_n ~|~ \cF_{n-1}} \xrightarrow{L^1} \bbu_0 \Theta \parenv{\bbu_l^c}^T.
\end{align}

\begin{example}
\label{ex:run_ex4b}
Continuing Example \ref{ex:run_ex4}, we now calculate the asymptotic differences covariance matrix $\Theta$ for our mutation system with $k=2$. We find each value of the matrix using the equation:

\begin{align*}
\theta_{m_a,m_b}&:=\sum_{v\in\cA^{2k-1}}\sum_{\eta\in \supp(\vt(v_{k-1}))}
\bbr^{(2k-1)}(v) \Pr(\vt(v_{k-1})=\eta) 
\parenv{\bct^{(k)}_{v_0^{k-2} \eta v_k^{2k-2}}(m_a)-\bct^{(k)}_{v}(m_a)}\parenv{\bct^{(k)}_{v_0^{k-2} \eta v_k^{2k-2}}(m_b)-\bct^{(k)}_{v}(m_b)}.
\end{align*}

We notice we already found $\bbr^{(2k-1)}=\bbr^{(3)}$. 
In each summand, we pick $v=v_0v_1v_2$ to be a 3-symbol long word $v\in \Sigma^3$ and $\eta$ to be a possible outcome of a mutation applied on the central symbol of $v_{k-1}=v_1$. 
For example, if $v=010$, $\eta$ would be a mutation on $1$ thus it could be either $0$ or $11$ according to \ref{ex:run_ex2}. 
Set $v=010, \eta=0, m_a=00$. As already found, $\bbr^{(3)}(010) = \frac{1}{12}, \Pr(\vt(1)=0) = \frac{1}{2}$. 
It is easy to see that $\bct^{(2)}_v(00) = 0$ (as $00$ does not appear in $010$) and that $\bct^{(2)}_{v_0 \eta v_2}(00) = 2$.

The information needed to calculate $\theta_{m_a,m_b}$ for every $m_a,m_b$ is presented in Table \ref{tabel:poss}.
\begin{table}[h!]
  \centering
  \caption{All the possible $v, \eta, \Pr(\vt(v_1)=\eta), v_0 \eta v_2$, and the changes in the counts for every pair, needed to calculate $\theta_{m_a,m_b}$.} 
  \label{tabel:poss}
\begin{tabular}{ |c|c|c|c|c c c c| } 
\hline
triple & $\eta$ & $\Pr(\vt(v_1)=\eta)$ & $v_0 \eta v_2$ & "00" & "01" & "10" & "11" \\
\hline
\multirow{2}{4em}{"000"} & "1" & 1/2 & "010" & -2 & 1 & 1 & 0 \\ 
& "00" & 1/2 & "0000" & 1 & 0 & 0 & 0 \\
\hline
\multirow{2}{4em}{"001"} & "1" & 1/2 & "011" & -1 & 0 & 0 & 1 \\ 
& "00" & 1/2 & "0001" & 1 & 0 & 0 & 0 \\
\hline
\multirow{2}{4em}{"010"} & "0" & 1/2 & "000" & 2 & -1 & -1 & 0 \\ 
& "11" & 1/2 & "0110" & 0 & 0 & 0 & 1 \\
\hline
\multirow{2}{4em}{"011"} & "0" & 1/2 & "001" & 1 & 0 & 0 & -1 \\ 
& "11" & 1/2 & "0111" & 0 & 0 & 0 & 1 \\
\hline
\multirow{2}{4em}{"100"} & "1" & 1/2 & "110" & -1 & 0 & 0 & 1 \\ 
& "00" & 1/2 & "1000" & 1 & 0 & 0 & 0 \\
\hline
\multirow{2}{4em}{"101"} & "1" & 1/2 & "111" & 0 & -1 & -1 & 2 \\ 
& "00" & 1/2 & "1001" & 1 & 0 & 0 & 0 \\
\hline
\multirow{2}{4em}{"110"} & "0" & 1/2 & "100" & 1 & 0 & 0 & -1 \\ 
& "11" & 1/2 & "1110" & 0 & 0 & 0 & 1 \\
\hline
\multirow{2}{4em}{"111"} & "0" & 1/2 & "101" & 0 & 1 & 1 & -2 \\ 
& "11" & 1/2 & "1111" & 0 & 0 & 0 & 1 \\
\hline
\end{tabular}
\end{table}
From the table, we can calculate the asymptotic corelation matrix and get:
\[\overline{\Theta} = \frac{1}{60}\begin{bmatrix}
    61 & -16 & -16 & -14\\
    -16 & 16 & 16 & -16\\
    -16 & 16 & 16 & -16\\
    -14 & -16 & -16 & 61\\
\end{bmatrix}\]
and from that, obtaining the asymptotic covariance matrix
\[\Theta = \overline{\Theta} - (\tau-1)^2\bbr\bbr^T = \frac{1}{1200}
\begin{bmatrix}
 1193  & -338  & -338  & -307 \\
 -338  &  308  &  308  & -338 \\
 -338  &  308  &  308  & -338 \\
 -307  & -338  & -338  &  1193
\end{bmatrix}\]
with $\bbr = \bbr^{(2)}$.
\end{example}

A few useful properties of the matrix $\Theta$ are obtained next.  
\begin{lemma}
\label{lem:theta_mat_psd}
The matrix $\Theta$ is (symmetric) positive semi-definite. 
\end{lemma}

\begin{IEEEproof}
We start by noticing that symmetry follows immediately from the definition
\begin{align*}
\Theta^T &:= \lim_{n \to \infty} \parenv{\E\sparenv{\parenv{\nabla \bct^{(k)}_n-(\tau-1)\bbr} \parenv{\nabla \bct^{(k)}_n-(\tau-1)\bbr}^T ~|~ \cF_{n-1}}}^T \\ 
&= \lim_{n \to \infty} \E\sparenv{\parenv{\parenv{\nabla \bct^{(k)}_n-(\tau-1)\bbr} \parenv{\nabla \bct^{(k)}_n-(\tau-1)\bbr}^T}^T ~|~ \cF_{n-1}}\\ 
&= \lim_{n \to \infty} \E\sparenv{\parenv{\nabla \bct^{(k)}_n-(\tau-1)\bbr} \parenv{\nabla \bct^{(k)}_n-(\tau-1)\bbr}^T ~|~ \cF_{n-1}}\\
&:= \Theta. 
\end{align*}

To show positive semi-definite, we need to show that for every real column vector $\bbu\in\R^{d^k}$, we have $\bbu^T \Theta \bbu\geq 0$. 
To that end, let $\bbv$ be a real column vector $\bbv\in \R^{d^k}$. Notice that 
\[\bbu^T(\bbv\bbv^T)\bbu = (\bbu^T\bbv)(\bbu^T\bbv)^T\] 
but since $\bbu^T\bbv\in \R$, we obtain 
\[\bbu^T(\bbv\bbv^T)\bbu= (\bbu^T\bbv)^2\geq 0.\]
Thus, for every column vector $\bbv$, the matrix $\bbv\bbv^T$ is positive semi-definite. 
Now notice that for every $n$, the conditional expectation 
\[\E\sparenv{\parenv{\nabla \bct^{(k)}_n-(\tau-1)\bbr} \parenv{\nabla \bct^{(k)}_n-(\tau-1)\bbr}^T ~|~ \cF_{n-1}}\]
is a function that, for every argument, returns 
a positive semi-definite matrix. 
Since it converges to a constant function and since the eigenvalues are continuous with respect to the matrix entries, the limit matrix $\Theta$ is also positive semi-definite. 
\end{IEEEproof}

In some cases, it is possible to show that $\Theta$ satisfies a slightly stronger property, the property of positive definiteness. 
For that, we require the following definition. 
\begin{definition}
Let $\mu$ be a distribution over $\R^n$. We say that $\mu$ is \textbf{degenerate} if its support $\supp(\mu)$ lies in a subspace of dimension less than $n$. In other words, if $\mu$ is degenerate and $\bbx$ is chosen according to $\mu$, then there exists a (deterministic) vector $\bbu\neq \0$ such that $\Pr(\bbu\bbx^T=0)=1$. 
We say that $\mu$ is non-degenerate if its support is $n$-dimensional. 
\end{definition}

\begin{lemma}
\label{lem:theta_mat_pd}
Let $S$ be an average $\tau$-mutation system and fix $k$. 
Assume that the frequency vector of $(2k-1)$-tuples converge to a \textbf{strictly positive} vector $\bbr^{(2k-1)}>0$, and suppose that the distribution of $\nabla \bct^{(k)}_t - (\tau-1)\bbr$ is non-degenerate for some $t\in\N$, 
Then the matrix $\Theta$ is (symmetric) positive definite. 

\textbf{Note:} the first assumption can be concluded, assuming the $(2k-1)$-substitution matrix $\bbM^{(2k-1)}$ is irreducible, utilizing Theorem \ref{th:main3}.
\end{lemma}

\begin{IEEEproof} 
Throughout the proof, we denote $\bct^{(k)}_n$ by $\bct_n$. 
Our goal is to show that for every vector $\bbu$, we have $\bbu\Theta \bbu^T>0$. 
Suppose we have a word $w=w_0\dots w_{n-1}\in\cA^n$ for some $n$. 
Suppose also that we select the $i$th symbol for mutation, i.e., that we obtain $\vt(w)$ by selecting $w_i$ and then 
$\vt(w)=w_0,\dots, w_{i-1},\vt(w_i),w_{i+1},\dots,w_{n-1}$. 
It is clear that $\bct_{\vt(w)}-\bct_w$ depends only on $w_{i-k+1+[2k-1]}=w_{i-k+1}\dots w_{i+k-1}$ and on $\vt(w_i)$. 
In other words, the vector $\bct_{\vt(w)}-\bct_w$ depends only on the $(2k-1)$-length sub-block of $w$, in which the central position was chosen for mutation (and in the symbol itself). 

Assume now a mutation system with starting word $\eta$. 
From the assumption that $\nabla\bct_t - (\tau-1)\bbr$ is non-degenerate, we obtain that for every non-zero row vector $\bbu$, 
\[\Pr\parenv{\bbu\cdot (\nabla\bct_t - (\tau-1)\bbr) = 0} < 1 \Rightarrow \E\sparenv{(\bbu\cdot (\nabla\bct_n - (\tau-1)\bbr))^2} > 0.\] 
By writing the expected value explicitly, we have 
\[\E\sparenv{(\bbu\cdot (\nabla\bct_n - (\tau-1)\bbr))^2}=\sum_{\bbv'\in\supp(\nabla\bct_t)} \Pr\parenv{\nabla\bct_t=\bbv'}(\bbu\cdot(\bbv' - (\tau-1)\bbr))^2>0.\] 
Thus, there exists $\bbv\in\supp(\nabla\bct_t)$ for which $\Pr\parenv{\nabla\bct_t=\bbv}(\bbu\cdot(\bbv - (\tau-1)\bbr))^2>0$. 
This implies that there exists a $(2k-1)$-length sub-word of $\vt^t(\eta)$, denoted as $w=w_{-k+1}\dots w_0\dots w_{k-1}$, such that
\[\Pr\parenv{\bct_{w_{k+1}\dots w_{-1}\vt(w_0)w_{1}\dots w_{k-1}}-\bct_{w_{-k+1}\dots w_{k-1}}=\bbv}>0.\] 
Since we assume that the frequency of $(2k-1)$-tuples converges in probability to a positive vector, it means the $(2k-1)$-tuple $w$ appears with positive probability. 

Overall, we have 
\begin{align*} 
\bbu\Theta\bbu^T&= \lim_{n\to\infty} \E\sparenv{\bbu\cdot(\nabla\bct_n - (\tau-1)\bbr)\cdot(\nabla\bct_n - (\tau-1)\bbr)^T\cdot\bbu^T ~|~ \cF_{n-1}}
\end{align*}
where $\cF_n$ is the $\sigma$-algebra generated by $\bct_j$ for $j\leq n$. 
Let $\cG_n$ denote the $\sigma$-algebra generated by $\bct^{(2k-1)}_j$ for $j\leq n$, and notice that $\cF_n\subseteq \cG_n$. 
Since $\bbu\cdot(\nabla\bct_n - (\tau-1)\bbr)(\cdot\nabla\bct_n - (\tau-1)\bbr)^T\cdot\bbu^T=(\bbu\cdot(\nabla\bct_n - (\tau-1)\bbr))^2$, and from the law of total expectation, we get 
\begin{align*} 
&\lim_{n\to\infty} \E\sparenv{\bbu\cdot(\nabla\bct_n - (\tau-1)\bbr)(\cdot\nabla\bct_n - (\tau-1)\bbr)^T\cdot\bbu^T ~|~ \cF_{n-1}}\\ 
&=\lim_{n\to\infty} \E\sparenv{\E\sparenv{(\bbu\cdot(\nabla\bct_n - (\tau-1)\bbr))^2 ~|~\cG_{n-1}} ~|~ \cF_{n-1}}.
\end{align*}
Bounding the inner conditional expectation we have, 
\begin{align*} 
&\lim_{n\to\infty} \E\sparenv{\E\sparenv{(\bbu(\nabla\bct_n - (\tau-1)\bbr))^2 ~|~\cG_{n-1}} ~|~ \cF_{n-1}}\\ 
&\geq \lim_{n\to\infty} \E\sparenv{\E\sparenv{\bfr^{(2k-1)}_{n-1}(w) ~|~ \cG_{n-1}} ~|~ \cF_{n-1}} \Pr\parenv{\bct_{w_{-k+1}\dots\vt(w_0)\dots w_{k-1}}-\bct_{w}=\bbv}(\bbu\cdot(\bbv - (\tau-1)\bbr))^2\\ 
&= \lim_{n\to\infty} \E\sparenv{\bfr^{(2k-1)}_{n-1}(w) ~|~ \cF_{n-1}} \Pr\parenv{\bct_{w_{-k+1}\dots\vt(w_0)\dots w_{k-1}}-\bct_{w}=\bbv}(\bbu\cdot(\bbv - (\tau-1)\bbr))^2
\end{align*} 
where $w=w_{-k+1}\dots w_0\dots w_{k-1}$ is the same $(2k-1)$-length word mentioned above. 
 
We denote by $\bbr^{(2k-1)}(w)>0$ the limiting value $\lim_{n\to\infty}\bfr^{(2k-1)}_n(w)$, which is strictly positive according to our assumption. 
For large enough $n$, we bound the conditional expectation and obtain that  
\[\E\sparenv{\bfr^{(2k-1)}_{n-1}(w) ~|~ \cF_{n-1}}>\frac{1}{2}\bbr^{(2k-1)}(w)>0,\] 
which, in turn, implies that 
\begin{align*}
    \bbu\Theta\bbu^T = &\lim_{n\to\infty} \E\sparenv{\bbu\cdot(\nabla\bct_n - (\tau-1)\bbr)\cdot(\nabla\bct_n - (\tau-1)\bbr)^T\cdot\bbu^T ~|~ \cF_{n-1}}  \\
    &\geq \lim_{n\to\infty} \E\sparenv{\bfr^{(2k-1)}_{n-1}(w) ~|~ \cF_{n-1}} \Pr\parenv{\bct_{w_{-k+1}\dots\vt(w_0)\dots w_{k-1}}-\bct_{w}=\bbv}(\bbu\cdot(\bbv - (\tau-1)\bbr))^2 \\ 
    &> 0.
\end{align*}
Thus $\Theta$ is positive definite.
\end{IEEEproof}

\begin{example}
\label{ex:run_ex5}
    Continuing from \ref{ex:run_ex4b} we will show the distribution of $\nabla\bct^{(2)} - (\tau-1)\bbr$ is degenerate as its support is not full rank. 
    From Table \ref{tabel:poss} we have 
\begin{align*}
\supp(\nabla \bct^{(2)}) &= \supp\parenv{\bmat{ -2\\ 1\\ 1\\ 0},\bmat{ 1\\ 0\\ 0\\ 0},\bmat{ -1\\ 0\\ 0\\ 1},\bmat{ 2\\ -1\\ -1\\ 0},\bmat{ 0\\ 0\\ 0\\ 1},\bmat{ 1\\ 0\\ 0\\ -1},\bmat{ 0\\ -1\\ -1\\ 2},\bmat{ 0\\ 1\\ 1\\ -2}}\\ 
&=\supp\parenv{\bmat{ 1\\ 0\\ 0\\ 0},\bmat{ 0\\ 1\\ 1\\ 0},\bmat{ 0\\ 0\\ 0\\ 1}}. 
\end{align*}
We recall for our example $\tau = \frac32$ and
\[\bbr = \frac{1}{10}\bmat{3\\ 2\\ 2\\ 3}.\]
So for $\bbu = \bmat{ 0& 1& -1& 0}$ we get $\Pr\parenv{\bbu (\nabla \bct^{(2)} - \frac12 \bbr) = 0} = 1$, so $\nabla \bct^{(2)} - (\tau-1)\bbr$ is degenerate and $\Theta$ is indeed symmetric positive \textbf{semi}-definite.
\end{example} 

\subsection{Conditional Variance Condition}
We are now ready to prove the conditional variance condition from Theorem \ref{main theorem}. 
Later, we show special properties in case $\nabla\bct_n^{(k)} - (\tau-1)\bbr$ is non-degenerate, but for now we will make no such assumption. 
Throughout, we maintain the same assumptions that for every non-leading eigenvalue $\lambda$ of $\bbM^{(k)}$, $\lambda^r < k + \frac12(\tau-1)$ and that if $\mu_a(\lambda)>\mu_g(\lambda)$ then $\lambda^r<k$ (for the general non-diagonalizable case). In addition, assume that the frequency vector of $(2k-1)$-tuples converges to some deterministic vector $\bbr^{(2k-1)}$.

Our goal is to show that there exists an a.s. finite random variable $\sigma^2$ such that $\sum_{j=1}^n\E\sparenv{\chi_{n,j}^2~|~\cF_{n,j-1}}\xrightarrow{p} \sigma^2$, where $\chi_{n,j}$ is defined in \eqref{eq:xi}. 
In other words, we prove the following theorem. 

\begin{claim}
\label{cl:CVC_deg}
Let $S$ be an average $\tau$-mutation system and let $k$ be a fixed integer. 
Let $\bbM^{(k)}$ denote the $k$-substitution matrix, and assume that for every eigenvalue $\lambda$ of $\bbM^{(k)}$ with $\mu_a(\lambda)>\mu_g(\lambda)$, it holds that $\lambda^r<k$. Assume also that for every non-leading eigenvalue $\lambda$, we have $\lambda^r< k + \frac12(\tau-1)$. Then there exists a non-negative deterministic value $\sigma^2$ such that 
\begin{align}
\label{eq:CVC_deg}
\frac{1}{n} \sum^n_{j=1} \E\sparenv{ \parenv{\alpha_0^r \nabla X^{r(0)}_j +\sum_{i=1}^m 
 \parenv{\beta^{r(i)}_{j,n} M^{r(i)}_j + \beta^{c(i)}_{j,n}M^{c(i)}_j} + \sum_{i=2 m+1}^{d^k-1} \beta^{r(i)}_{j,n} M^{r(i)}_j}^2 ~\Bigg|~ \cF_{j-1}} \xrightarrow{\Pr} \sigma^2. 
\end{align}
\end{claim}
In fact, we prove that the convergence in \eqref{eq:CVC_deg} is in $L^1$, which implies convergence in probability, but to prove Theorem \ref{main theorem}, only convergence in probability is required. 
First, several lemmas are needed. 
\begin{lemma}
\label{lem:L1_conv_betot}
For every non-leading eigenvalue $\lambda_i,\lambda_l$ of $\bbM^{(k)}$, we have 
\begin{align}
\frac{1}{n}\sum_{j=1}^n \beta_{j,n}^{r(i)}\beta_{j,n}^{r(l)} \E\sparenv{M_j^{r(i)} M_j^{r(l)} ~|~ \cF_{j-1}}&\xrightarrow{L^1} 
\frac{1}{n}\sum_{j=1}^n \beta_{j,n}^{r(i)}\beta_{j,n}^{r(l)} \bbu_i^r\Theta\parenv{\bbu_l^r}^T, \label{eq:first_lim} \\ 
\frac{1}{n}\sum_{j=1}^n \beta_{j,n}^{c(i)}\beta_{j,n}^{c(l)} \E\sparenv{M_j^{c(i)} M_j^{c(l)} ~|~ \cF_{j-1}}&\xrightarrow{L^1} 
\frac{1}{n}\sum_{j=1}^n \beta_{j,n}^{c(i)}\beta_{j,n}^{c(l)} \bbu_i^c\Theta\parenv{\bbu_l^c}^T, \label{eq:second_lim} \\
\frac{1}{n}\sum_{j=1}^n \beta_{j,n}^{c(i)}\beta_{j,n}^{r(l)} \E\sparenv{M_j^{c(i)} M_j^{r(l)} ~|~ \cF_{j-1}}&\xrightarrow{L^1} 
\frac{1}{n}\sum_{j=1}^n \beta_{j,n}^{c(i)}\beta_{j,n}^{r(l)} \bbu_i^c\Theta\parenv{\bbu_l^r}^T,
\label{eq:third_lim} \\
\frac{1}{n}\sum_{j=1}^n \alpha_0^r \beta_{j,n}^{r(l)} \E\sparenv{\nabla X_j^{(0)} M_j^{r(l)} ~|~ \cF_{j-1}}&\xrightarrow{L^1} 
\frac{1}{n}\sum_{j=1}^n \alpha_0^r \beta_{j,n}^{r(l)} \bbu_0\Theta\parenv{\bbu_l^r}^T,
\label{eq:forth_lim} \\
\frac{1}{n}\sum_{j=1}^n \alpha_0^r \beta_{j,n}^{c(l)} \E\sparenv{\nabla X_j^{(0)} M_j^{c(l)} ~|~ \cF_{j-1}}&\xrightarrow{L^1} 
\frac{1}{n}\sum_{j=1}^n \alpha_0^r \beta_{j,n}^{c(l)} \bbu_0\Theta\parenv{\bbu_l^c}^T,
\label{eq:fifth_lim} \\
\frac{1}{n}\sum_{j=1}^n \parenv{\alpha_0}^2 \E\sparenv{\parenv{\nabla X_j^{(0)}}^2 ~|~ \cF_{j-1}}&\xrightarrow{L^1} 
\frac{1}{n}\sum_{j=1}^n \parenv{\alpha_0}^2 \bbu_0\Theta\bbu_0^T = \parenv{\alpha_0}^2 \bbu_0\Theta\bbu_0^T.
\label{eq:sixth_lim} 
\end{align}
\end{lemma}

\begin{IEEEproof}
We show here only the first limit \eqref{eq:first_lim}, the proof of the rest is similar. 
This is shown by a straightforward calculation. 
Fix $i,l$ where $\lambda_i,\lambda_l$ are non-leading eigenvalues of $\bbM^{(2k-1)}$. 
Using the triangle inequality, for every $n$ we have 
\begin{align}
&\E\sparenv{\abs{\frac{1}{n}\sum_{j=1}^n \beta_{j,n}^{r(i)}\beta_{j,n}^{r(l)} \E\sparenv{M_j^{r(i)}M_j^{r(l)} ~|~ \cF_{j-1}}- \frac{1}{n}\sum_{j=1}^n \beta_{j,n}^{r(i)}\beta_{j,n}^{r(l)} \bbu_i^r\Theta\parenv{\bbu_l^r}^T }} \nonumber \\ \label{eq:ff1}
&\leq \frac{1}{n}\sum_{j=1}^n |\beta_{j,n}^{r(i)}| |\beta_{j,n}^{r(l)}| \E\sparenv{\abs{ \E\sparenv{M_j^{r(i)}M_j^{r(l)} ~|~ \cF_{j-1}}-  \bbu_i^r\Theta\parenv{\bbu_l^r}^T }}.
\end{align}
From \eqref{eq:L1_lim_MrMr} we have that 
$\E[ M^{r(i)}_j M^{r(l)}_j ~|~ \cF_{j-1}]$ converges in $L^1$ to $\bbu_i^r \Theta \parenv{\bbu_l^r}^T$. Thus, for every $\epsilon > 0$ we have $N$ such that for every $j \geq N$, 
\[\E \sparenv{\abs{  \E\sparenv{ M^{r(i)}_j M^{r(l)}_j ~|~ \cF_{j-1}} - \bbu_i^r \Theta \parenv{\bbu_l^r}^T }} < \epsilon.\] 
Thus, for large enough $N$, we may bound above  \eqref{eq:ff1} by  
\begin{align}
\label{eq:ff2}
\frac{1}{n} \sum^{N-1}_{j=1} |\beta^{r(i)}_{j,n}| |\beta^{r(l)}_{j,n}| \E\sparenv{\abs{ \E\sparenv{ M^{r(i)}_j M^{r(l)}_j ~|~ \cF_{j-1}} - \bbu_i^r \Theta \parenv{\bbu_l^r}^T }} + \frac{1}{n} \sum^n_{j=N}  |\beta^{r(i)}_{j,n}| |\beta^{r(l)}_{j,n}| \epsilon.
\end{align}
From our assumption that $\lambda_i^r< k + \frac12(\tau-1)$ and from the upper bound on $|\beta^{r(i)}_{j,n}|$ in \eqref{eq:real betabound}, we obtain that for every $i$, 
\begin{align*}
|\beta_{j,n}^{r(i)}|&\leq 
\parenv{|\alpha_i^r|+|\alpha_i^c|} \parenv{\frac{n}{j+1}}^{\frac{\lambda_i^r-k}{\tau-1}} +O\parenv{n^{\frac{\lambda_i^r-k}{\tau-1}-1}}\\ 
&\leq o\parenv{n^{1/2}}.
\end{align*} 
Since 
\[\sum^{N-1}_{j=1} \E\sparenv{\abs{ \E\sparenv{ M^{r(i)}_j M^{r(l)}_j ~|~ \cF_{j-1}} - \bbu_i^r \Theta \parenv{\bbu_l^r}^T }}\] 
is bounded, we get 
\begin{align}
\label{eq:ff3}
&\frac{1}{n}\sum^{N-1}_{j=1} |\beta^{r(i)}_{j,n}| |\beta^{r(l)}_{j,n}| \E\sparenv{\abs{ \E\sparenv{ M^{r(i)}_j M^{r(l)}_j ~|~ \cF_{j-1}} - \bbu_i^r \Theta \parenv{\bbu_l^r}^T }}\\ \nonumber 
&\leq \frac{1}{n}o(n).
\end{align}
On the other hand, we have
\begin{align*}
&\frac{\epsilon}{n} \sum^n_{j=N}  |\beta^{r(i)}_{j,n}| |\beta^{r(l)}_{j,n}| \\ 
&\leq \frac{\epsilon}{n} \sum^n_{j=N}  \abs{\parenv{|\alpha_i^r|+|\alpha_i^c|} \parenv{\frac{n}{j+1}}^{\frac{\lambda_i^r-k}{\tau-1}} + O\parenv{n^{\frac{\lambda_i^r-k}{\tau-1}-1}}}  \abs{\parenv{|\alpha_l^r|+|\alpha_l^c|} \parenv{\frac{n}{j+1}}^{\frac{\lambda_l^r-k}{\tau-1}} +O\parenv{n^{\frac{\lambda_l^r-k}{\tau-1}-1}}}\\ 
&\leq \frac{1}{n}o(n)+ \frac{\epsilon}{n} \parenv{|\alpha_i^r|+|\alpha_i^c|} \parenv{|\alpha_l^r|+|\alpha_l^c|} \sum^n_{j=N} \parenv{\frac{n}{j+1}}^{\frac{\lambda_i^r-k}{\tau-1}+\frac{\lambda_l^r-k}{\tau-1}}\\ 
&\stackrel{(a)}{\leq} \frac{1}{n}o(n)+ \epsilon \parenv{|\alpha_i^r|+|\alpha_i^c|} \parenv{|\alpha_l^r|+|\alpha_l^c|} \frac{1}{1-\frac{\lambda_i^r-k}{\tau-1}-\frac{\lambda_l^r-k}{\tau-1}}\\ 
&=O(\epsilon), 
\end{align*}
where $(a)$ follows since $\parenv{\frac{\lambda_i^r-k}{\tau-1}+ \frac{\lambda_l^r-k}{\tau-1}}<1$ and since $\sum_{j=1}^n \frac{1}{j^a}=\frac{n^{1-a}}{1-a}$ for $0<a<1$. 
Since this is true for every $\epsilon>0$, we obtain that 
\begin{align*}
\lim_{n\to\infty}\E\sparenv{\abs{\frac{1}{n}\sum_{j=1}^n \beta_{j,n}^{r(i)}\beta_{j,n}^{r(l)} \E\sparenv{M_j^{r(i)}M_j^{r(l)} ~|~ \cF_{j-1}}- \frac{1}{n}\sum_{j=1}^n \beta_{j,n}^{r(i)}\beta_{j,n}^{r(l)} \bbu_i^r\Theta\parenv{\bbu_l^r}^T }}= 0
\end{align*} 
which finishes the proof. 
\end{IEEEproof}

The following three lemmas are technical lemmas we require to prove Claim \ref{cl:CVC_deg}. 
As such, we defer their proofs to the appendix. 
\begin{lemma}
\label{lem:int_limit}
Let $r<1$ be a real number, let $a,b\in \R$. Then we have 
\begin{align}
&\label{eq:int1} \int_0^1 t^{-r}\cos\parenv{a\log(t)}\cos\parenv{b\log(t)}dt= 
\frac{1-r}{2}\parenv{\frac{1}{(1-r)^2+(a-b)^2}+\frac{1}{(1-r)^2+(a+b)^2}}, \\ 
&\label{eq:int2} \int_0^1 t^{-r}\cos\parenv{a\log(t)}\sin\parenv{b\log(t)}dt= 
\frac{1}{2}\parenv{\frac{a+b}{(1-r)^2+(a+b)^2}-\frac{a-b}{(1-r)^2+(a-b)^2}}, \\
&\label{eq:int3} \int_0^1 t^{-r}\sin\parenv{a\log(t)}\sin\parenv{b\log(t)}dt= 
\frac{1-r}{2}\parenv{\frac{1}{(1-r)^2+(a-b)^2}-\frac{1}{(1-r)^2+(a+b)^2}}. 
\end{align}
\end{lemma} 

\begin{lemma}
\label{lem:finite_well_def}
The following limits exist for all non-leading eigenvalues $\lambda_i,\lambda_l$ of $\bbM^{(k)}$.
We denote $a=\parenv{1-\frac{\lambda_i^r+\lambda_l^r-2k}{\tau-1}}$, $b=\frac{\lambda_i^c-\lambda_l^c}{\tau-1}$ and $c=\frac{\lambda_i^c+\lambda_l^c}{\tau-1}$. 
\begin{align}
&\lim_{n \to \infty}\frac{1}{n} \sum^n_{j=1} \beta^{r(i)}_{j,n} \beta^{r(l)}_{j,n} \nonumber \\ 
&= \frac{a\parenv{\alpha_i^r \alpha_l^r-\alpha_i^c \alpha_l^c}- c\parenv{\alpha_i^r \alpha_l^c+\alpha_i^c \alpha_l^r}}{2(a^2+c^2)} 
+\frac{a\parenv{\alpha_i^r \alpha_l^r+\alpha_i^c \alpha_l^c}+ b\parenv{\alpha_i^r \alpha_l^c-\alpha_i^c \alpha_l^r}}{2(a^2+b^2)}. \label{eq:aa1}  \\
&\lim_{n \to \infty}\frac{1}{n} \sum^n_{j=1} \beta^{c(i)}_{j,n} \beta^{c(l)}_{j,n} \nonumber \\ 
&= \frac{a\parenv{\alpha_i^c \alpha_l^c-\alpha_i^r \alpha_l^r}+ c\parenv{\alpha_i^c \alpha_l^r+\alpha_i^r \alpha_l^c}}{2(a^2+c^2)} 
+\frac{a\parenv{\alpha_i^r \alpha_l^r+\alpha_i^c \alpha_l^c}+ b\parenv{\alpha_i^r \alpha_l^c-\alpha_i^c \alpha_l^r}}{2(a^2+b^2)}. \label{eq:aa2}  \\
&\lim_{n \to \infty}\frac{1}{n} \sum^n_{j=1} \beta^{r(i)}_{j,n} \beta^{c(l)}_{j,n} \nonumber \\ 
&= \frac{a\parenv{\alpha_i^r \alpha_l^c+\alpha_i^c \alpha_l^r}+ c\parenv{\alpha_i^r \alpha_l^r-\alpha_i^c \alpha_l^c}}{2(a^2+c^2)} 
+\frac{a\parenv{\alpha_i^r \alpha_l^c-\alpha_i^c \alpha_l^r}- b\parenv{\alpha_i^r \alpha_l^r+\alpha_i^c \alpha_l^c}}{2(a^2+b^2)}. \label{eq:aa3} 
\end{align}

In case $\lambda_i,\lambda_l$ are purely real, we have the equality 
\begin{align} 
\label{eq:aa_r}
\lim_{n \to \infty}\frac{1}{n} \sum^n_{j=1} \beta^{r(i)}_{j,n} \beta^{r(l)}_{j,n} = \frac{\alpha_i^r \alpha_l^r}{1 - \frac{\lambda_i^r+\lambda_l^r-2k}{\tau-1}}.
\end{align}
\end{lemma}

\begin{lemma}
\label{lem:finite_well_def_x0}
The following limits exists for every non-leading eigenvalue $\lambda_l$ of $\bbM^{(k)}$ assuming $\lambda^r < k +\frac{1}{2}(\tau-1)$
\begin{align}
    \lim_{n\to\infty}\frac{1}{n} \sum_{j=1}^{n}\alpha_0^r \beta_{j,n}^{r(l)} &= \alpha_0^r \frac{\parenv{\tau-1}\parenv{\alpha_l^r\parenv{\tau + k - 1 -\lambda_l^r} - \alpha_l^c \lambda_l^c}}{\parenv{\tau + k - 1 -\lambda_l^r}^2 + \parenv{\lambda_l^c}^2},
    \\
    \label{beta c conv}
    \lim_{n\to\infty}\frac{1}{n} \sum_{j=1}^{n}\alpha_0^r \beta_{j,n}^{c(l)} &= \alpha_0^r \frac{\parenv{\tau-1}\parenv{\alpha_l^c\parenv{\tau + k - 1 -\lambda_l^r} + \alpha_l^r \lambda_l^c}}{\parenv{\tau + k - 1 -\lambda_l^r}^2 + \parenv{\lambda_l^c}^2}.
\end{align}
\end{lemma}

To ease the notation, for $d^k>i,l>0$, we define the following limits. 
\begin{align}
    B^{r,r}_{i,l}&:=\lim_{n \to \infty}\frac{1}{n} \sum^n_{j=1} \beta^{r(i)}_{j,n} \beta^{r(l)}_{j,n}, \nonumber  \\
    B^{r,c}_{i,l}&:=\lim_{n \to \infty}\frac{1}{n} \sum^n_{j=1} \beta^{r(i)}_{j,n} \beta^{c(l)}_{j,n}, \nonumber  \\
    B^{c,c}_{i,l}&:=\lim_{n \to \infty}\frac{1}{n} \sum^n_{j=1} \beta^{c(i)}_{j,n} \beta^{c(l)}_{j,n}, \nonumber  \\
    B^{r,r}_{0,l}&:=\lim_{n \to \infty}\frac{1}{n} \sum^n_{j=1} \alpha_0^r \beta^{r(l)}_{j,n}, \nonumber  \\
    B^{r,c}_{0,l}&:=\lim_{n \to \infty}\frac{1}{n} \sum^n_{j=1} \alpha_0^r \beta^{c(l)}_{j,n},\nonumber  \\ 
    B^{r,r}_{0,0}&:=\lim_{n \to \infty}\frac{1}{n} \sum^n_{j=1} \parenv{\alpha_0^r}^2 = \parenv{\alpha_0^r}^2.\label{def:sum_limits}
\end{align}
Notice that the above limits were already given in closed forms in Lemma \ref{lem:finite_well_def} and Lemma \ref{lem:finite_well_def_x0}, but are dependent on the values of $\alpha^r$ and $\alpha^c$.

We are now ready to prove Claim \ref{cl:CVC_deg}. 
\begin{IEEEproof}[Proof of Claim \ref{cl:CVC_deg}] 
First, since the expected value is taken over non-negative values, if the $L^1$ limit exists, it is clear that it is non-negative. 
We begin by writing the inner argument in \eqref{eq:CVC_deg}. 
\begin{align*}
& \parenv{\alpha_0^r \nabla X_j^{(0)} + \sum_{i=1}^m \parenv{\beta^{r(i)}_{j,n} M^{r(i)}_j + \beta^{c(i)}_{j,n}M^{c(i)}_j} + \sum_{i=2 m+1}^{d^k-1} \beta^{r(i)}_{j,n} M^{r(i)}_j}^2
\\ 
&= \parenv{\alpha_0}^2 \parenv{\nabla X_j^{(0)}}^2  + 2\sum_{l=1}^m \alpha_0^r \nabla X_j^{(0)} \parenv{\beta^{r(l)}_{j,n} M^{r(l)}_j + \beta^{c(l)}_{j,n}M^{c(l)}_j} + 2\sum_{l=2 m+1}^{d^k-1} \alpha_0^r \nabla X_j^{(0)} \beta^{r(l)}_{j,n} M^{r(l)}_j
\\
&\quad + \sum_{i=1}^m \sum_{l=1}^m \parenv{\beta^{r(i)}_{j,n} M^{r(i)}_j + \beta^{c(i)}_{j,n}M^{c(i)}_j} \parenv{\beta^{r(l)}_{j,n} M^{r(l)}_j + \beta^{c(l)}_{j,n}M^{c(l)}_j}
\\ 
&\quad + \sum_{i=2 m+1}^{d^k-1} \sum_{l=2 m+1}^{d^k-1} \beta^{r(i)}_{j,n} M^{r(i)}_j \beta^{r(l)}_{j,n} M^{r(l)}_j
\\ 
& \quad + 2 \sum_{i=2 m+1}^{d^k-1}  \sum_{l=1}^m \beta^{r(i)}_{j,n} M^{r(i)}_j \parenv{\beta^{r(l)}_{j,n} M^{r(l)}_j + 
\beta^{c(l)}_{j,n}M^{c(l)}_j}. 
\end{align*}

Thus, we have

\begin{align}
&\frac{1}{n} \sum^n_{j=1} \E\sparenv{ \parenv{\alpha_0^r \nabla X_j^{(0)} + \sum_{i=1}^m 
 \parenv{\beta^{r(i)}_{j,n} M^{r(i)}_j + \beta^{c(i)}_{j,n}M^{c(i)}_j} + \sum_{i=2 m+1}^{d^k-1} \beta^{r(i)}_{j,n} M^{r(i)}_j}^2 ~\Bigg|~ \cF_{j-1}}
 \nonumber \\ 
&= \frac{1}{n} \sum^n_{j=1} \parenv{\alpha_0}^2 \E\sparenv{ \parenv{\nabla X_j^{(0)}}^2 ~\Big|~ \cF_{j-1}}  + \sum_{l=1}^m \frac{2}{n} \sum^n_{j=1} \alpha_0^r \parenv{\beta^{r(l)}_{j,n} \E\sparenv{ \nabla X_j^{(0)}  M^{r(l)}_j  ~\Big|~ \cF_{j-1}} + \beta^{c(l)}_{j,n} \E\sparenv{ \nabla X_j^{(0)} M^{c(l)}_j  ~\Big|~ \cF_{j-1}}}
\nonumber \\
&\quad +\sum_{l=2 m+1}^{d^k-1} \frac{2}{n} \sum^n_{j=1} \alpha_0^r \beta^{r(l)}_{j,n} \E\sparenv{ \nabla X_j^{(0)}  M^{r(l)}_j ~\Big|~ \cF_{j-1}}
\nonumber \\
&\quad +\sum_{i=1}^m \sum_{l=1}^m \frac{1}{n} \sum^n_{j=1} \parenv{\beta^{r(i)}_{j,n}\beta^{r(l)}_{j,n} 
\E\sparenv{ M^{r(i)}_j M^{r(l)}_j ~\Big|~ \cF_{j-1}} + \beta^{c(i)}_{j,n}\beta^{c(l)}_{j,n} 
\E\sparenv{ M^{c(i)}_j M^{c(l)}_j ~\Big|~ \cF_{j-1}} + 2\beta^{r(i)}_{j,n}\beta^{c(l)}_{j,n} 
\E\sparenv{M^{r(i)}_j M^{c(l)}_j ~\Big|~ \cF_{j-1}}}
\nonumber \\
&\quad +\sum_{i=2 m+1}^{d^k-1} \sum_{l=2 m+1}^{d^k-1} \frac{1}{n} \sum^n_{j=1} \beta^{r(i)}_{j,n}\beta^{r(l)}_{j,n} 
\E\sparenv{ M^{r(i)}_j M^{r(l)}_j ~\Big|~ \cF_{j-1}}
\nonumber \\ 
&\quad + \sum_{i=2 m+1}^{d^k-1} \sum_{l=1}^m\frac{2}{n} \sum^n_{j=1} \parenv{\beta^{r(i)}_{j,n}\beta^{r(l)}_{j,n}
\E\sparenv{ M^{r(i)}_j M^{r(l)}_j ~\Big|~ \cF_{j-1}} + \beta^{r(i)}_{j,n}\beta^{c(l)}_{j,n} 
\E\sparenv{ M^{r(i)}_j M^{c(l)}_j ~\Big|~ \cF_{j-1}}}. \label{eq:CC1}
\end{align}
Using Lemma \ref{lem:L1_conv_betot}, we replace the sums with their $L^1$ limits to obtain 
\begin{align}
&\frac{1}{n} \sum^n_{j=1} \E\sparenv{ \parenv{\alpha_0^r \nabla X_j^{(0)} + \sum_{i=1}^m \parenv{\beta^{r(i)}_{j,n} M^{r(i)}_j + \beta^{c(i)}_{j,n}M^{c(i)}_j} + \sum_{i=2 m+1}^{d^k-1} \beta^{r(i)}_{j,n} M^{r(i)}_j}^2 ~\Bigg|~ \cF_{j-1}}\stackrel{L^1}{\sim} \\ 
&\parenv{\alpha_0^r}^2 \bbu_0\Theta\bbu_0^T  + \sum_{l=1}^m \frac{2}{n} \sum^n_{j=1} \alpha_0^r \parenv{\beta^{r(l)}_{j,n} \bbu_0\Theta\parenv{\bbu_l^r}^T + \beta^{c(l)}_{j,n} \bbu_0\Theta\parenv{\bbu_l^c}^T}
\nonumber \\
&\quad +\sum_{l=2 m+1}^{d^k-1} \frac{2}{n} \sum^n_{j=1} \alpha_0^r \beta^{r(l)}_{j,n} \bbu_0\Theta\parenv{\bbu_l^r}^T
\nonumber \\
&\quad +\sum_{i=1}^m \sum_{l=1}^m \frac{1}{n} \sum^n_{j=1} \parenv{ \beta^{r(i)}_{j,n}\beta^{r(l)}_{j,n} \bbu_i^r \Theta \parenv{\bbu_l^r}^T + \beta^{c(i)}_{j,n}\beta^{c(l)}_{j,n}\bbu_i^c \Theta \parenv{\bbu_l^c}^T + 
2\beta^{r(i)}_{j,n}\beta^{c(l)}_{j,n}\bbu_i^r \Theta \parenv{\bbu_l^c}^T } \nonumber\\
&\quad + \sum_{i=2 m+1}^{d^k-1} \sum_{l=2 m+1}^{d^k-1} \frac{1}{n}\sum^n_{j=1} \beta^{r(i)}_{j,n}\beta^{r(l)}_{j,n} 
\bbu_i^r \Theta \parenv{\bbu_l^r}^T \nonumber\\ 
&\quad + \sum_{i=2 m+1}^{d^k-1} \sum_{l=1}^m \frac{2}{n} \sum^n_{j=1} \parenv{\beta^{r(i)}_{j,n}\beta^{r(l)}_{j,n}\bbu_i^r \Theta \parenv{\bbu_l^r}^T + \beta^{r(i)}_{j,n}\beta^{c(l)}_{j,n}\bbu_i^r \Theta \parenv{\bbu_l^c}^T}. \nonumber
\end{align}
Since $\bbu_i^r \Theta \parenv{\bbu_l^r}^T$ is finite (and similarly, for $\bbu_i^c,\bbu_l^c$), from Lemma \ref{lem:finite_well_def} we obtain that the $L^1$ limit exists, which proves the theorem. Utilizing \eqref{def:sum_limits} we can present the limit as 

\begin{align}
&\lim_{n\to\infty}\frac{1}{n} \sum^n_{j=1} \E\sparenv{ \parenv{\alpha_0^r \nabla X_j^{(0)} + \sum_{i=1}^m \parenv{\beta^{r(i)}_{j,n} M^{r(i)}_j + \beta^{c(i)}_{j,n}M^{c(i)}_j} + \sum_{i=2 m+1}^{d^k-1} \beta^{r(i)}_{j,n} M^{r(i)}_j}^2 ~\Bigg|~ \cF_{j-1}}\xrightarrow{L^1} \\
&B^{r,r}_{0,0} \bbu_0\Theta\bbu_0^T  + \sum_{l=1}^m 2\parenv{B^{r,r}_{0,l} \bbu_0\Theta\parenv{\bbu_l^r}^T + B^{r,c}_{0,l} \bbu_0\Theta\parenv{\bbu_l^c}^T}
\nonumber \\
&\quad +\sum_{l=2 m+1}^{d^k-1} 2B^{r,r}_{0,l} \bbu_0\Theta\parenv{\bbu_l^r}^T
\nonumber \\
&\quad + \sum_{i=1}^m \sum_{l=1}^m  \parenv{ B^{r,r}_{i,l} \bbu_i^r \Theta \parenv{\bbu_l^r}^T +  B^{c,c}_{i,l} \bbu_i^c \Theta \parenv{\bbu_l^c}^T + 
2B^{r,c}_{i,l}\bbu_i^r \Theta \parenv{\bbu_l^c}^T } \nonumber\\
&\quad + \sum_{i=2 m+1}^{d^k-1} \sum_{l=2 m+1}^{d^k-1} B^{r,r}_{i,l} 
\bbu_i^r \Theta \parenv{\bbu_l^r}^T \nonumber\\ 
&\quad + \sum_{i=2 m+1}^{d^k-1} \sum_{l=1}^m 2 \parenv{B^{r,r}_{i,l}\bbu_i^r \Theta \parenv{\bbu_l^r}^T + B^{r,c}_{i,l}\bbu_i^r \Theta \parenv{\bbu_l^c}^T} = \sigma^2. \label{eq:sigma}
\end{align}
This concludes the proof.
\end{IEEEproof}

Thus, Claim \ref{cl:CVC_deg} is proved. Together with the conditional Lindeberg condition, we obtained \eqref{eq:small_main}. 
Thus, for every set of $\alpha_0, \alpha_1^r, \alpha_1^c...$, the linear sum
\[\frac{1}{\sqrt{n}}\alpha_0^r X_n^{r(0)} + \sum_{i=1}^m \frac{1}{\sqrt{n}}\parenv{\alpha_i^r X_n^{r(i)} + \alpha_i^c X_n^{c(i)}} + \sum_{l=2 m+1}^{d^k-1} \frac{1}{\sqrt{n}}\alpha_i^r X_n^{r(i)}\]
converges in distribution to a normally distributed, zero mean, random value with variance $\sigma^2$, where  
\begin{align*}
   \sigma^2 = &B^{r,r}_{0,0} \bbu_0\Theta\bbu_0^T  + \sum_{l=1}^m 2\parenv{B^{r,r}_{0,l} \bbu_0\Theta\parenv{\bbu_l^r}^T + B^{r,c}_{0,l} \bbu_0\Theta\parenv{\bbu_l^c}^T}\\
    &\quad +\sum_{l=2 m+1}^{d^k-1} 2B^{r,r}_{0,l} \bbu_0\Theta\parenv{\bbu_l^r}^T\\
    &\quad + \sum_{i=1}^m \sum_{l=1}^m  \parenv{ B^{r,r}_{i,l} \bbu_i^r \Theta \parenv{\bbu_l^r}^T + B^{c,c}_{(i,l)}\bbu_i^c \Theta \parenv{\bbu_l^c}^T + 
    2B^{r,c}_{i,l}\bbu_i^r \Theta \parenv{\bbu_l^c}^T }\\
    &\quad + \sum_{i=2 m+1}^{d^k-1} \sum_{l=2 m+1}^{d^k-1} B^{r,r}_{i,l} 
    \bbu_i^r \Theta \parenv{\bbu_l^r}^T\\ 
    &\quad + \sum_{i=2 m+1}^{d^k-1} \sum_{l=1}^m 2 \parenv{B^{r,r}_{i,l}\bbu_i^r \Theta \parenv{\bbu_l^r}^T + B^{r,c}_{i,l}\bbu_i^r \Theta \parenv{\bbu_l^c}^T}.
\end{align*}
This, in turn, implies that $\frac{1}{\sqrt{n}}\bbX_n$ converges to a vector with a joint normal distribution, i.e., proves Theorem \ref{main theorem}. 
To prove Theorem \ref{Main Main Theorem}, we need to find the covariance matrix of the distribution. Before doing so, we will show the unique property of $\sigma^2$ in case $\nabla \bct_n - (\tau-1)\bbr$ is non-degenerate.

\subsection{Variance in the non-degenerate case}
In the case that $\nabla\bct_t^{(k)} - (\tau-1)\bbr$ is non-degenerate for some $t>0$ (and positivity of $\bbr^{(2k-1)}$), we show that $\sigma^2>0$ for any set of constants $\alpha_i^r, \alpha_i^c$ (assuming at least one value of $\alpha_i^r$ or $\alpha_i^r$ is non-zero). This leads to the asymptotic projections covariance matrix $\Sigma'$ being strictly positive.

\begin{claim}
\label{cl:CVC_ndeg}
Given the same assumptions as before in Claim \ref{cl:CVC_deg}, and the assumptions of non-degeneracy and positivity of $\bbr^{(2k-1)}$ from Lemma \ref{lem:theta_mat_pd}, if at least one value of $\alpha_0^r$, $\alpha_i^r$ or $\alpha_i^c$ is non-zero, then 
\begin{align*}
    &\frac{1}{\sqrt{n}}\alpha_0^r X_n^{r(0)} + \sum_{i=1}^m \frac{1}{\sqrt{n}}\parenv{\alpha_i^r X_n^{r(i)} + \alpha_i^c X_n^{c(i)}} + \sum_{i=2 m+1}^{d^k-1} \frac{1}{\sqrt{n}}\alpha_i^r X_n^{r(i)} \xrightarrow{d} \cN\parenv{0,\sigma^2}, 
\end{align*} 
with $\sigma^2 > 0$ strictly positive. 
This implies that the asymptotic projections covariance matrix $\Sigma'$ is strictly positive. 
\end{claim}
To that end, we use the following known results about matrices (see, for example, \cite[Chap. 7]{horn2012matrix}). 
\begin{lemma}
\label{lem:diag_pd_prop}
Let $A$ be a real, square matrix, then the following holds. 
\begin{enumerate}
\item If $A$ is real valued and diagonalizable, $\lambda$ is a complex eigenvalue of $A$, and $u+\bbi v$ is the eigenvector corresponding to $\lambda$, then $u$ and $v$ are linearly independent. 
\item If $A$ is symmetric and positive definite, then there exists a matrix $A^{1/2}$ such that $(A^{1/2})^2=A$ and $A^{1/2}$ is symmetric, positive definite (and in particular, invertible). 
\end{enumerate}
\end{lemma} 

\begin{IEEEproof}[Proof of Claim \ref{cl:CVC_ndeg}] 
According to Lemma \ref{lem:theta_mat_pd}, the matrix $\Theta$ is positive definite, which implies, according to Lemma \ref{lem:diag_pd_prop} that $\Theta^{1/2}$ exists (and is symmetric, positive definite, and invertible). 
Since $\sigma^2\geq 0$ (Theorem \ref{th:martingaleCLT} holds for the degenerate and non-degenerate cases), we only need to show that $\sigma^2>0$. 
Recall that from the proof for Claim \ref{cl:CVC_deg}, $\sigma^2$ is the limit in probability of 
\begin{align}
\label{eq:ab1}
\frac{1}{n} \sum^n_{j=1} \E\sparenv{ \parenv{\alpha_0^r \nabla X_j^{(0)} + \sum_{i=1}^m \parenv{\beta^{r(i)}_{j,n} M^{r(i)}_j + \beta^{c(i)}_{j,n}M^{c(i)}_j} + \sum_{i=2 m+1}^{d^k-1} \beta^{r(i)}_{j,n} M^{r(i)}_j}^2 ~\Bigg|~ \cF_{j-1}}.
\end{align}
Recalling also \eqref{eq:sigma}, we have  
\begin{align}
\sigma^2 &= \lim_{n \to \infty} \parenv{\alpha_0^r}^2 \bbu_0\Theta\bbu_0^T + \sum_{l=1}^m \frac{2}{n} \sum^n_{j=1} \alpha_0^r \parenv{\beta^{r(l)}_{j,n} \bbu_0\Theta\parenv{\bbu_l^r}^T + \beta^{c(l)}_{j,n} \bbu_0\Theta\parenv{\bbu_l^c}^T} 
+\sum_{l=2 m+1}^{d^k-1} \frac{2}{n} \sum^n_{j=1} \alpha_0^r \beta^{r(l)}_{j,n} \bbu_0\Theta\parenv{\bbu_l^r}^T
\nonumber \\
& + \sum_{i=1}^m \sum_{l=1}^m \frac{1}{n} \sum^n_{j=1} \parenv{ \beta^{r(i)}_{j,n}\beta^{r(l)}_{j,n} \bbu_i^r 
\Theta \parenv{\bbu_l^r}^T +\beta^{c(i)}_{j,n}\beta^{c(l)}_{j,n}\bbu_i^c \Theta \parenv{\bbu_l^c}^T + 2\beta^{r(i)}_{j,n} \beta^{c(l)}_{j,n}\bbu_i^r \Theta \parenv{\bbu_l^c}^T }
\label{eq:ab2}\\
&+ \sum_{i=2 m+1}^{d^k-1} \sum_{l=2 m+1}^{d^k-1} \frac{1}{n} \sum^n_{j=1} \beta^{r(i)}_{j,n} \beta^{r(l)}_{j,n} 
\bbu_i^r \Theta \parenv{\bbu_l^r}^T 
+ \sum_{i=2 m+1}^{d^k-1} \sum_{l=1}^m \frac{2}{n} \sum^n_{j=1} \parenv{ \beta^{r(i)}_{j,n}\beta^{r(l)}_{j,n}\bbu_i^r \Theta \parenv{\bbu_l^r}^T + \beta^{r(i)}_{j,n}\beta^{c(l)}_{j,n}\bbu_i^r \Theta \parenv{\bbu_l^c}^T }
\nonumber \\ 
&= \lim_{n \to \infty} \frac{1}{n} \sum^n_{j=1} \Bigg\| \alpha_0^r \bbu_0 \Theta^\frac{1}{2} + \sum_{i=1}^m 
 \parenv{\beta^{r(i)}_{j,n} \bbu_i^r \Theta^\frac{1}{2} + \beta^{c(i)}_{j,n}\bbu_i^c \Theta^\frac{1}{2}} + \sum_{i=2 m+1}^{d^k-1} \beta^{r(i)}_{j,n} \bbu_i^r \Theta^\frac{1}{2}\Bigg\|^2 \label{eq:ab3} \\
 &\geq 0. \nonumber 
\end{align}
In order to show the above non-negativity, for large $n$ we can approximate \eqref{eq:ab2} by replacing 
$\beta_{j,n}^{r(i)},\beta_{j,n}^{c(i)}$ with  
\begin{align*}
\beta_i^r(t)&:=\parenv{\frac{1}{t}}^{\frac{\lambda_i^r-k}{\tau-1}}\parenv{\alpha_i^r \cos\parenv{\frac{\lambda_i^c}{\tau-1}\log(t)}- \alpha_i^c \sin\parenv{\frac{\lambda_i^c}{\tau-1}\log(t)}}, \\ 
\beta_i^c(t)&:=\parenv{\frac{1}{t}}^{\frac{\lambda_i^r-k}{\tau-1}}\parenv{\alpha_i^c \cos\parenv{\frac{\lambda_i^c}{\tau-1}\log(t)}+ \alpha_i^r \sin\parenv{\frac{\lambda_i^c}{\tau-1}\log(t)}}, 
\end{align*}
where $0\leq t\leq 1$. 
Inverting the argument $n/j$ to $j/n$ (and subsequently substituting $t$) within the logarithm induces a sign change in the sine function. 
For a more rigorous derivation of this approximation, see the proof of Lemma \ref{lem:finite_well_def}, specifically equation \eqref{eq:temp1}. 
Thus, we need to show that 
\begin{align}
\label{eq:fa1}
\int_0^1 \Bigg\|  \alpha_0^r \bbu_0^r \Theta^\frac{1}{2} + \sum_{i=1}^m \parenv{\beta_i^r(t) \bbu_i^r \Theta^\frac{1}{2} + \beta_i^c(t)\bbu_i^c \Theta^\frac{1}{2}} + \sum_{i=2 m+1}^{d^k-1} \beta_i^r(t) \bbu_i^r \Theta^\frac{1}{2}\Bigg\|^2 dt >0. 
\end{align}

Due to the smoothness of the functions, to prove \eqref{eq:fa1}, it is enough to show that there exists $t\in (0,1)$ for which the norm is not zero. 
To that end, recall that eigenvectors of distinct eigenvalues are linearly independent. From Lemma \ref{lem:diag_pd_prop}, if $\bbu_i^r, \bbu_i^c \neq 0$ then they are linearly independent, which means the set of all vectors $\bbu_0$, $\bbu_i^r$ and $\bbu_i^c$ is linearly independent. Since $\Theta^{1/2}$ is invertible, the set of all vectors $\bbu_0\Theta^{1/2}$, $\bbu_i^r\Theta^{1/2}$ and $\bbu_i^c\Theta^{1/2}$ is linearly independent as well. 
Thus, the sum 
\[\alpha_0^r \bbu_0 \Theta^\frac{1}{2} + \sum_{i=1}^m \parenv{\beta_i^r(t) \bbu_i^r \Theta^\frac{1}{2} + \beta_i^c(t)\bbu_i^c \Theta^\frac{1}{2}} + 
\sum_{i=2 m+1}^{d^k-1} \beta_i^r(t) \bbu_i^r \Theta^\frac{1}{2}\] 
is equal to the zero vector $\0$ iff $\alpha_0^r=\beta_i^r(t)=\beta_i^c(t)=0$ for all $i$. 
Thus, since $\alpha^r_0$ could be $0$, to show \eqref{eq:fa1} it is enough to show that there exists $i$ and $t\in (0,1)$ for which $\beta_i^r(t)\neq 0$. 
Assume $i$ is such that $\alpha_i^r,\alpha_i^c$ are not both zero. 
Since $t^{-\frac{\lambda_i^r-k}{\tau-1}}>0$ for $t>0$, we may ignore this part in $\beta_i^r(t)$. 
We remain with 
\begin{align*}
&\alpha_i^r \cos\parenv{\frac{\lambda_i^c}{\tau-1}\log(t)}- \alpha_i^c \sin\parenv{\frac{\lambda_i^c}{\tau-1}\log(t)} \\ 
&= \sqrt{\parenv{\alpha_i^r}^2+\parenv{\alpha_i^c}^2}\cos\parenv{\frac{\lambda_i^c}{\tau-1}\log(t)-\phi}.
\end{align*}
If $\lambda_i^c\neq 0$, then the cosine function is equal to zero only on a discrete set. 
If $\lambda_i^c=0$, the eigenvalue is real and is equal to $\sqrt{\parenv{\alpha_i^r}^2+\parenv{\alpha_i^c}^2}\neq 0$, and 
$\beta_i^r(t) = \parenv{\frac{1}{t}}^{\frac{\lambda_i^r-k}{\tau-1}}\alpha_i^r\neq 0$. 
This proves \eqref{eq:fa1} and thus the claim.
\end{IEEEproof}
In the next section, we show how to explicitly calculate the covariance matrix, and thus finish the proof Theorem \ref{Main Main Theorem}. 

\section{Calculating the Covariance Matrix}\label{sec:cov_matrix}
We have found in \eqref{eq:ab2} the asymptotic variance $\sigma^2$ for an arbitrary linear sum of the entries of $\frac{1}{\sqrt{n}}\bbX_n$, now we would like to find the asymptotic distribution of $\frac{1}{\sqrt{n}}\bbX_n$. We know it has a zero mean joint normal distribution, thus we only need to find the \textbf{asymptotic projection covariance matrix} $\Sigma'$, to do so we would find the asymptotic variances and covariances of all entries in $\frac{1}{\sqrt{n}}\bbX_n$ by setting different values for the weights $\alpha_i^r$, $\alpha_i^c$ in the linear sum and finding its variance $\sigma^2$ using previous proofs.

We begin with the asymptotic variance of $\frac{1}{\sqrt{n}}X^{r(i)}_n$. 
Setting $\alpha_i^r=1$ and the rest of the weights to $0$ in \eqref{eq:ab2}, we obtain 
\begin{align*}
    \lim_{n\to\infty}&\var(\frac{1}{\sqrt{n}}X^{r(i)}_n) \\
    &=B^{r,r}_{0,0} \bbu_0\Theta\bbu_0^T  + \sum_{l=1}^m 2\parenv{B^{r,r}_{0,l} \bbu_0\Theta\parenv{\bbu_l^r}^T + B^{r,c}_{0,l} \bbu_0\Theta\parenv{\bbu_l^c}^T}\\
    &\quad +\sum_{l=2 m+1}^{d^k-1} 2B^{r,r}_{0,l} \bbu_0\Theta\parenv{\bbu_l^r}^T\\
    &\quad + \sum_{i=1}^m \sum_{l=1}^m  \parenv{ B^{r,r}_{i,l} \bbu_i^r \Theta \parenv{\bbu_l^r}^T + B^{c,c}_{(i,l)}\bbu_i^c \Theta \parenv{\bbu_l^c}^T + 
    2B^{r,c}_{i,l}\bbu_i^r \Theta \parenv{\bbu_l^c}^T }\\
    &\quad + \sum_{i=2 m+1}^{d^k-1} \sum_{l=2 m+1}^{d^k-1} B^{r,r}_{i,l} 
    \bbu_i^r \Theta \parenv{\bbu_l^r}^T\\ 
    &\quad + \sum_{i=2 m+1}^{d^k-1} \sum_{l=1}^m 2 \parenv{B^{r,r}_{i,l}\bbu_i^r \Theta \parenv{\bbu_l^r}^T + B^{r,c}_{i,l}\bbu_i^r \Theta \parenv{\bbu_l^c}^T}.
\end{align*}
Plugging in $\alpha_i^r=1$ and the rest of the weights $0$ gives 
\[B^{r,r}_{i,i} \bbu^r_i \Theta \parenv{\bbu^r_i}^T + B^{c,c}_{i,i} \bbu^c_i \Theta \parenv{\bbu^c_i}^T + 2B^{r,c}_{i,i} \bbu^r_i \Theta \parenv{\bbu^c_i}^T.\] 
Replacing the values of $B_{i,i}^{r,r},B_{i,i}^{c,c},B_{i,i}^{r,c}$ as in \eqref{def:sum_limits}, we obtain 
\begin{align*}
    &\frac{\tau-1}{2}\parenv{\frac{\tau-1+2k-2\lambda_i^r}{\parenv{\tau-1+2k-2\lambda_i^r}^2+\parenv{2\lambda_i^c}^2}+\frac{1}{\tau-1+2k-2\lambda_i^r}} \bbu^r_i \Theta \parenv{\bbu^r_i}^T \\
    &+\frac{\tau-1}{2}\parenv{\frac{-\parenv{\tau-1+2k-2\lambda_i^r}}{\parenv{\tau-1+2k-2\lambda_i^r}^2+\parenv{2\lambda_i^c}^2}+\frac{1}{\tau-1+2k-2\lambda_i^r}} \bbu^c_i \Theta \parenv{\bbu^c_i}^T \\
    &+ 2\frac{\tau-1}{2}\frac{2\lambda_i^c}{\parenv{\tau-1+2k-2\lambda_i^r}^2+\parenv{2\lambda_i^c}^2} \bbu^r_i \Theta \parenv{\bbu^c_i}^T \\
    &=\frac{\tau-1}{2}\parenv{\frac{\parenv{\tau-1+2k-2\lambda_i^r}\parenv{\bbu^r_i \Theta \parenv{\bbu^r_i}^T - \bbu^c_i \Theta \parenv{\bbu^c_i}^T} + 4\lambda_i^c \bbu^r_i \Theta \parenv{\bbu^c_i}^T}{\parenv{\tau-1+2k-2\lambda_i^r}^2+\parenv{2\lambda_i^c}^2}+\frac{\bbu^r_i \Theta \parenv{\bbu^r_i}^T + \bbu^c_i \Theta \parenv{\bbu^c_i}^T}{\tau-1+2k-2\lambda_i^r}}. 
\end{align*}
Similarly, to calculate the asymptotic variance of $\frac{1}{\sqrt{n}}X^{c(i)}_n$, we set $\alpha_i^c=1$ and the values of the other $\alpha$s to $0$ and obtain 
\begin{align*}
    &\lim_{n\to\infty}\var(\frac{1}{\sqrt{n}}X^{c(i)}_n) \\
    &=\frac{\tau-1}{2}\parenv{\frac{\parenv{\tau-1+2k-2\lambda_i^r}\parenv{\bbu^c_i \Theta \parenv{\bbu^c_i}^T - \bbu^r_i \Theta \parenv{\bbu^r_i}^T} - 4\lambda_i^c \bbu^r_i \Theta \parenv{\bbu^c_i}^T}{\parenv{\tau-1+2k-2\lambda_i^r}^2+\parenv{2\lambda_i^c}^2}+\frac{\bbu^r_i \Theta \parenv{\bbu^r_i}^T + \bbu^c_i \Theta \parenv{\bbu^c_i}^T}{\tau-1+2k-2\lambda_i^r}}. 
\end{align*}
For the case $i=0$, we get 
\[\lim_{n\to\infty}\var(\frac{1}{\sqrt{n}}X^{r(0)}_n) = \bbu^r_0 \Theta \parenv{\bbu^r_0}^T.\]

In order to find the covariance of $\frac{1}{\sqrt{n}}X_n^{r(i)}$ and $\frac{1}{\sqrt{n}}X_n^{c(i)}$, we use the formula 
\begin{align*}
&\text{Cov}\parenv{\frac{1}{\sqrt{n}}X^{r(i)}_n, \frac{1}{\sqrt{n}}X^{c(i)}_n} = \\ 
&\frac{1}{2}\parenv{\var\parenv{\frac{1}{\sqrt{n}}X^{r(i)}_n + \frac{1}{\sqrt{n}}X^{c(i)}_n}-\var\parenv{\frac{1}{\sqrt{n}}X^{r(i)}_n}-\var\parenv{\frac{1}{\sqrt{n}}X^{c(i)}_n}}. 
\end{align*}
Thus, we only need to find the variance of $\frac{1}{\sqrt{n}}X^{r(i)}_n + \frac{1}{\sqrt{n}}X^{c(i)}_n$, which is obtained by setting 
$\alpha_i^r,\alpha_i^c=1$ and the rest of the weights to $0$. This yields 
\begin{align*}
    &\lim_{n\to\infty}\text{Cov}\parenv{\frac{1}{\sqrt{n}}X^{r(i)}_n, \frac{1}{\sqrt{n}}X^{c(i)}_n}= \\ 
    &=\parenv{\tau-1}\frac{\lambda_i^c\parenv{\bbu^c_i \Theta \parenv{\bbu^c_i}^T - \bbu^r_i \Theta \parenv{\bbu^r_i}^T} + \parenv{\tau-1+2k-2\lambda_i^r}\bbu^r_i \Theta \parenv{\bbu^c_i}^T}{\parenv{\tau-1+2k-2\lambda_i^r}^2+\parenv{2\lambda_i^c}^2}.
\end{align*}
In a similar fashion, we find the covariance between all the pairs. 
For simplicity, we denote 
\begin{align*}
    &a_{i,l} = \tau -1 + 2k -\lambda_i^r -\lambda_l^r,\\ 
    &b_{i,l} = \lambda_i^c -\lambda_l^c, \\
    &c_{i,l} = \lambda_i^c +\lambda_l^c, \\ 
    &a_{0,l} = \tau -1 + k -\lambda_l^r, \\ 
    &a_{0,0} = \tau -1.
\end{align*}
\begin{corollary}
\label{cor:final_cov}
The covariance is listed below. 
\begin{align*}
    \lim_{n\to\infty}&\text{Cov}\parenv{\frac{1}{\sqrt{n}}X^{r(i)}_n, \frac{1}{\sqrt{n}}X^{r(l)}_n} = \\
    &\frac{\tau-1}{2}\frac{a_{i,l}\parenv{\bbu^r_i \Theta \parenv{\bbu^r_l}^T - \bbu^c_i \Theta \parenv{\bbu^c_l}^T}+c_{i,l} \parenv{\bbu^r_i \Theta \parenv{\bbu^c_l}^T + \bbu^c_i \Theta \parenv{\bbu^r_l}^T}}{a_{i,l}^2 + c_{i,l}^2} \\
    &+\frac{\tau-1}{2}\frac{a_{i,l}\parenv{\bbu^r_i \Theta \parenv{\bbu^r_l}^T + \bbu^c_i \Theta \parenv{\bbu^c_l}^T} -b_{i,l}\parenv{\bbu^r_i \Theta \parenv{\bbu^c_l}^T-\bbu^c_i \Theta \parenv{\bbu^r_l}^T}}{a_{i,l}^2+b_{i,l}^2}, \\ 
    \lim_{n\to\infty}&\text{Cov}\parenv{\frac{1}{\sqrt{n}}X^{c(i)}_n, \frac{1}{\sqrt{n}}X^{c(l)}_n} = \\
    &\frac{\tau-1}{2}\frac{a_{i,l}\parenv{\bbu^c_i \Theta \parenv{\bbu^c_l}^T - \bbu^r_i \Theta \parenv{\bbu^r_l}^T}-c_{i,l} \parenv{\bbu^r_i \Theta \parenv{\bbu^c_l}^T + \bbu^c_i \Theta \parenv{\bbu^r_l}^T}}{a_{i,l}^2 + c_{i,l}^2} \\
    &+\frac{\tau-1}{2}\frac{a_{i,l}\parenv{\bbu^r_i \Theta \parenv{\bbu^r_l}^T + \bbu^c_i \Theta \parenv{\bbu^c_l}^T}-b_{i,l} \parenv{\bbu^r_i \Theta \parenv{\bbu^c_l}^T - \bbu^c_i \Theta \parenv{\bbu^r_l}^T}}{a_{i,l}^2+b_{i,l}^2}, \\
    \lim_{n\to\infty}&\text{Cov}\parenv{\frac{1}{\sqrt{n}}X^{r(i)}_n, \frac{1}{\sqrt{n}}X^{c(l)}_n} = \\
    &\frac{\tau-1}{2} \frac{c_{i,l}\parenv{\bbu^c_i \Theta \parenv{\bbu^c_l}^T - \bbu^r_i \Theta \parenv{\bbu^r_l}^T}+a_{i,l} \parenv{\bbu^r_i \Theta \parenv{\bbu^c_l}^T + \bbu^c_i \Theta \parenv{\bbu^r_l}^T}}{a_{i,l}^2 + c_{i,l}^2} \\
    &+\frac{\tau-1}{2} \frac{b_{i,l}\parenv{\bbu^r_i \Theta \parenv{\bbu^r_l}^T + \bbu^c_i \Theta \parenv{\bbu^c_l}^T} + a_{i,l}\parenv{\bbu^r_i \Theta \parenv{\bbu^c_l}^T - \bbu^c_i \Theta \parenv{\bbu^r_l}^T}}{a_{i,l}^2+b_{i,l}^2}.
\end{align*}
\end{corollary}
We may now construct the asymptotic covariance matrix of $\frac{1}{\sqrt{n}}\bbX_n$, denoted as $\Sigma'$, 
\[\Sigma' = \lim_{n\to\infty}\text{Cov}(\frac{1}{\sqrt{n}}\bbX_n).\]
The asymptotic distribution of $\bbX_n$ is given by 
\[\frac{1}{\sqrt{n}}\bbX_n = \frac{1}{\sqrt{n}}U(\bct_n - n(\tau -1)\bbr) \xrightarrow{d} \cN(0, \Sigma')\]
Since $\bbU$ is invertible by Lemma \ref{lem:U_full_rank}, this immediately implies
\[\frac{1}{\sqrt{n}}\left(\bct_n - n(\tau-1)\bbr\right) = \bbU^{-1} \cdot \frac{1}{\sqrt{n}}\bbX_n \xrightarrow{d} \cN(\boldsymbol{0}, \Sigma),\]
where $\Sigma = \bbU^{-1}\Sigma'(\bbU^{-1})^T$, thus proving Theorem \ref{Main Main Theorem}.

\begin{example}
\label{ex:run_ex6}
    We now present the entire process of calculating the Central Limit Theorem for $k$-tuples in an average $\tau$-mutation system. 
    Consider Example \ref{ex:run_ex5} of a mutation law with $\tau = \frac{3}{2}$: 
\[\vt(0)=\begin{cases} 1& \text{w.p. } \frac{1}{2}\\ 00&\text{w.p. } \frac{1}{2}\end{cases}, \;  \vt(1)=\begin{cases} 0& \text{w.p. } \frac{1}{2}\\ 11&\text{w.p. } \frac{1}{2}\end{cases}.\] 
and set $k=2$. 
We have already found that 
\begin{align*}
&\bbM^{(2)}= \bmat{\frac{3}{2}&1&\frac{1}{2}&0\\ \frac{1}{2}&1&0&\frac{1}{2}\\ \frac{1}{2}&0&1&\frac{1}{2}\\0&\frac{1}{2}&1&\frac{3}{2}},
\end{align*}
and that the frequency of triples converges in probability to a strictly positive vector. 
We also found the right leading eigenvector of $\bbM^{(2)}$ 
\begin{align*}
    \bbr^{(2)} = \frac{1}{10}\bmat{3\\ 2 \\2 \\3}, 
\end{align*}
and the left eigendecomposition of $\bbM^{(2)}$ 
\[\lambda_0 = \frac{5}{2},~ \lambda_1 = \frac{3}{2},~ \lambda_2 = 1,~ \lambda_3 = 0\]
with 
\[\bbU = \bmat{
    1 &1 &1 &1\\
    -1 &-1 &1 &1\\
    0 &-1 &1 &0\\
    1& -\frac32& -\frac32 &1
}.\]
We also calculated the $\Theta$ matrix (recall in this case it is positive semi-definite) 
\[\frac{1}{1200}
\bmat{
 1193  & -338  & -338  & -307 \\
 -338  &  308  &  308  & -338 \\
 -338  &  308  &  308  & -338 \\
 -307  & -338  & -338  &  1193
}. \] 
Notice that all the eigenvalues are real and the condition for all non-leading eigenvalues $\lambda^r < k + \frac12(\tau-1)$ holds, so we can find every entry in the asymptotic covariance matrix by calculating the of asymptotic covariance of the projections $\frac{1}{\sqrt{n}}X_n^{r(i)},\frac{1}{\sqrt{n}}X_n^{r(l)}$, for every pair $i,l$, using Corollary \ref{cor:final_cov}, obtaining 
\[ \Sigma' =
\bmat{
 \frac{1}{4}  &  0  &  0  &  \frac{1}{2} \\[6pt]
 0  &  \frac{5}{2}  &  0  &  0 \\[6pt]
 0  &  0  &  0  &  0 \\[6pt]
 \frac{1}{2}  &  0  &  0  &  \frac{43}{6}}.\] 
Calculating $\Sigma$ 
we get
\[\Sigma = \begin{bmatrix}
 0.274685  & -0.018852  & -0.018852  & -0.141981 \\
-0.018852  &  0.033852  &  0.033852  & -0.018852 \\
-0.018852  &  0.033852  &  0.033852  & -0.018852 \\
-0.141981  & -0.018852  & -0.018852  &  0.274685
\end{bmatrix}.\]
Therefore, 
\[\lim_{n\to\infty}\frac{1}{\sqrt{n}}\parenv{\bct_n - \frac12 n\bbr^{(2)}} \sim \cN(\boldsymbol{0}, \Sigma).\]
\end{example}

Next, we describe the case in which the substitution matrix is non-diagonalizable. 

\section{Non-Diagonalizable Substitution Matrices}\label{sec:non_diag}
In this section, we outline the proof for the case in which the substitution matrix $\bbM^{(k)}$ is non-diagonalizable.
The proof follows a similar structure to the regular case, replacing the standard eigendecomposition with the \textbf{Jordan decomposition}. 
The analysis in this case is much more involved, mainly due to the difference in the $\beta$s. 
Since the core ideas are the same as in the diagonalizable case, we only sketch how to obtain the value for $\beta$s, while the rest of the analysis follows the exact same steps as in the diagonalizable case. 
Notice that Lemma \ref{lem:diag_pd_prop}, as stated above, requires the matrix $\bbM$ to be diagonalizable. The lemma still holds even for the non-diagonalizable case, only stated with "general eigenvectors" instead of "eigenvectors". 

Suppose the Jordan decomposition of $\bbM^{(k)}$ yields a Jordan block of size $H_i$ for the eigenvalue $\lambda_i$.
We denote the corresponding chain of left generalized eigenvectors as:
\[
    \parenv{\bbu_{i,1},\, \bbu_{i,2},\, \bbu_{i,3},\, \dots,\, \bbu_{i,H_i}}
\]
where $\bbu_{i,1}M^{(k)} = \lambda_i\bbu_{i,1}$ and $\bbu_{i,h}M^{(k)} = \lambda_i\bbu_{i,h} + \bbu_{i,h-1}$ for $h>1$.
For convenience, we define $\bbu_{i,0} := \boldsymbol{0}$ for all $i$, so that the recurrence relation holds uniformly:
\[
    \bbu_{i,h}M^{(k)} = \lambda_i\bbu_{i,h} + \bbu_{i,h-1}
    \quad \text{for all } h \ge 1.
\]
We define the projections for the general case as follows:
\begin{align*}
    \bbu_{i,h} &= \bbu_{i,h}^r + \bbi\, \bbu_{i,h}^c,
    \\
    X_n^{r(i,h)} &= \bbu_{i,h}^r \cdot \left(\bct_n - n(\tau-1)\bbr\right),
    \\
    X_n^{c(i,h)} &= \bbu_{i,h}^c \cdot \left(\bct_n - n(\tau-1)\bbr\right).
\end{align*}
Our goal is to establish a central limit theorem for the following linear combination of projections:
\[
    \frac{1}{\sqrt{n}}\alpha_0^r X^{r(0)}_n
    + \frac{1}{\sqrt{n}}\sum_{i=1}^{p-1}\sum^{H_i}_{h=1}
    \parenv{\alpha_{i,h}^r X^{r(i,h)}_n + \alpha_{i,h}^c X^{c(i,h)}_n}
    \xrightarrow{d} \cN(0,\sigma^2)
\]
for every collection of real-valued constants $\alpha_0^r,\, \alpha_{i,h}^r,\, \alpha_{i,h}^c \in \mathbb{R}$.
We note that for real eigenvectors $X^{c(i,h)}_n = 0$ and $\alpha_{i,h}^c = 0$, while for complex conjugate eigenvalue pairs
we retain only one representative from each pair and decompose it into its real and imaginary parts, as in the regular case.
Here $p$ denotes the number of Jordan blocks in the decomposition, excluding blocks corresponding to the discarded eigenvalue
from each complex conjugate pair.

By Theorem~\ref{th:irreducible_mat_properties}, the Jordan block associated with the maximal eigenvalue $\lambda_0$ has size $1$,
hence $H_0 = 1$ and the analysis of $X^{r(0)}_n$ is identical to the regular case.
For Theorem~\ref{th:main3}, we additionally require that for any defective eigenvalue satisfying $\mu_a(\lambda) > \mu_g(\lambda)$,
the real part is strictly bounded above by $k$, i.e., $\lambda^r < k$.

To prove the CLT 
\[\frac{1}{\sqrt{n}}\bbX_n \xrightarrow{d} \cN(0,\Sigma'),\] 
we compute the asymptotic covariance matrix of the projection vector 
\begin{align*}
    \bbX_n &:= \begin{bmatrix}
        X_n^{r(0)} \;&
        X_n^{r(1,1)} \;&
        X_n^{c(1,1)} \;&
        X_n^{r(1,2)} \;&
        X_n^{c(1,2)} \;&
        \cdots \;&
        X_n^{r(p-1,\, H_{p-1})} \;&
        X_n^{c(p-1,\, H_{p-1})}
    \end{bmatrix}^T.
\end{align*}
To this end, we compute the conditional mean of the projection increments:
\begin{align*}
    \nabla X_n^{(i,h)}
        &:= X_n^{(i,h)} - X_{n-1}^{(i,h)},
    \\
    \nabla X_n^{(i,h)}
        &= \nabla X_n^{r(i,h)} + \bbi\, \nabla X_n^{c(i,h)},
    \\
    \mathbb{E}\!\sparenv{\nabla X_n^{(i,h)} \mid \cF_{n-1}}
        &= \bbu_{i,h}\, \frac{1}{Y_{n-1}}(M^{(k)}-kI)\bct_{n-1}
         = \frac{1}{Y_{n-1}}\!\parenv{(\lambda_i-k)X_{n-1}^{(i,h)} + X_{n-1}^{(i,h-1)}},
    \\
    \mathbb{E}\!\sparenv{\nabla X_n^{r(i,h)} \mid \cF_{n-1}}
        &= \frac{1}{Y_{n-1}}\!\parenv{(\lambda_i^r-k)X_{n-1}^{r(i,h)}
           - \lambda_i^c X_{n-1}^{c(i,h)} + X_{n-1}^{r(i,h-1)}},
    \\
    \mathbb{E}\!\sparenv{\nabla X_n^{c(i,h)} \mid \cF_{n-1}}
        &= \frac{1}{Y_{n-1}}\!\parenv{(\lambda_i^r-k)X_{n-1}^{c(i,h)}
           + \lambda_i^c X_{n-1}^{r(i,h)} + X_{n-1}^{c(i,h-1)}},
\end{align*}
and approximate the cumulative sum as a sum of martingale differences:
\begin{align*}
    M^{r(i,h)}_n
        &:= \nabla X^{r(i,h)}_n
            - \frac{1}{Y_{n-1}}\!\parenv{(\lambda_i^r-k)X_{n-1}^{r(i,h)}
              - \lambda_i^c X_{n-1}^{c(i,h)} + X_{n-1}^{r(i,h-1)}},
    \\
    M^{c(i,h)}_n
        &:= \nabla X^{c(i,h)}_n
            - \frac{1}{Y_{n-1}}\!\parenv{(\lambda_i^r-k)X_{n-1}^{c(i,h)}
              + \lambda_i^c X_{n-1}^{r(i,h)} + X_{n-1}^{c(i,h-1)}},
    \\
    \overline{M}^{r(i,h)}_n
        &:= \nabla X^{r(i,h)}_n
            - \frac{1}{(n-1)(\tau-1)+Y_0}\!\parenv{(\lambda_i^r-k)X_{n-1}^{r(i,h)}
              - \lambda_i^c X_{n-1}^{c(i,h)} + X_{n-1}^{r(i,h-1)}},
    \\
    \overline{M}^{c(i,h)}_n
        &:= \nabla X^{c(i,h)}_n
            - \frac{1}{(n-1)(\tau-1)+Y_0}\!\parenv{(\lambda_i^r-k)X_{n-1}^{c(i,h)}
              + \lambda_i^c X_{n-1}^{r(i,h)} + X_{n-1}^{c(i,h-1)}},
    \\
    \overline{v}^{(i)}_n
        &:= \sum^{H_i}_{h=1}\sum^n_{j=1}
            \parenv{\beta^{r(i,h)}_{j,n}\overline{M}^{r(i,h)}_j
                  + \beta^{c(i,h)}_{j,n}\overline{M}^{c(i,h)}_j}
         := \sum^{H_i}_{h=1}
            \parenv{\alpha_{i,h}^r X^{r(i,h)}_n + \alpha_{i,h}^c X^{c(i,h)}_n}
            + \varepsilon_n^{(i)},
    \\
    v^{(i)}_n
        &:= \sum^{H_i}_{h=1}\sum^n_{j=1}
            \parenv{\beta^{r(i,h)}_{j,n} M^{r(i,h)}_j
                  + \beta^{c(i,h)}_{j,n} M^{c(i,h)}_j}.
\end{align*}
The derivation of the coefficients $\beta^{r(i,h)}_{j,n}$ and $\beta^{c(i,h)}_{j,n}$ is done similarly to the diagonalizable case, only more involved, and therefore is deferred to the appendix.
We obtain the following. 
\begin{claim}
\label{cl:beta_val_gen} 
The values of $\beta^{r(i,h)}_{j,n}$ and $\beta^{c(i,h)}_{j,n}$ are given below. 
\begin{align*}
\beta^{r(i,h)}_{j-1,n}& = \parenv{\frac{n}{j}}^{\!\frac{\lambda_i^r-k}{\tau-1}}\Bigg\{
         \cos\!\parenv{\frac{\lambda_i^c}{\tau-1}\log(n/j)}
         \parenv{\alpha^r_{i,h} + \sum_{l=1}^{H_i-h}\parenv{e^{r,l}_{i,j}\alpha^r_{i,h+l}
        + e^{c,l}_{i,j}\alpha^c_{i,h+l}}} \notag\\ 
    &\quad + \sin\!\parenv{\frac{\lambda_i^c}{\tau-1}\log(n/j)}
         \parenv{\alpha^c_{i,h}  + \sum_{l=1}^{H_i-h}\parenv{e^{r,l}_{i,j}\alpha^c_{i,h+l} - e^{c,l}_{i,j}\alpha^r_{i,h+l}}}
       \Bigg\} + O\!\parenv{n^{\frac{\lambda_i^r-k}{\tau-1}-1}\ln^{H_i-h}n} \\ 
    &= O\!\parenv{n^{\frac{\lambda_i^r-k}{\tau-1}}\ln^{H_i-h}n}, \\
\beta^{c(i,h)}_{j-1,n} &= \parenv{\frac{n}{j}}^{\!\frac{\lambda_i^r-k}{\tau-1}} \Bigg\{
         \cos\!\parenv{\frac{\lambda_i^c}{\tau-1}\log(n/j)} \parenv{\alpha^c_{i,h}
           + \sum_{l=1}^{H_i-h}\parenv{e^{r,l}_{i,j}\alpha^c_{i,h+l}  - e^{c,l}_{i,j}\alpha^r_{i,h+l}}} \notag\\
    &\quad - \sin\!\parenv{\frac{\lambda_i^c}{\tau-1}\log(n/j)} \parenv{\alpha^r_{i,h}
           + \sum_{l=1}^{H_i-h}\parenv{e^{r,l}_{i,j}\alpha^r_{i,h+l} + e^{c,l}_{i,j}\alpha^c_{i,h+l}}}
       \Bigg\} + O\!\parenv{n^{\frac{\lambda_i^r-k}{\tau-1}-1}\ln^{H_i-h}n} \\
       &\quad = O\!\parenv{n^{\frac{\lambda_i^r-k}{\tau-1}}\ln^{H_i-h}n}, 
\end{align*} 
where the residuals $e^{r,h}_{i,j}$ and $e^{c,h}_{i,j}$ are defined recursively by 
\begin{align*}
    e^1_{i,j} &:= \sum_{t=j}^{n} \frac{1}{(t-1)(\tau-1) + m + \lambda_i - k}, \\
    e^h_{i,j} &:= \sum_{t=j}^{n-h+1} \frac{e^{h-1}_{\lambda,\,t+1}}{(t-1)(\tau-1) + m + \lambda_i - k}
            \quad \text{for } n-h+1 \ge j, \quad \text{and } 0 \text{ otherwise}, \\
    e^h_{i,j} &:= e^{r,h}_{i,j} + \bbi\, e^{c,h}_{i,j}.
\end{align*}
The associated error term satisfies 
\[\varepsilon_n^{(i)} = O\!\parenv{n^{\frac{\lambda_i^r-k}{\tau-1}}\ln^{H_i-1}n}.\]
\end{claim}

Applying analogous arguments to those in Corollary~\ref{t1 to v hat convergence}, we establish the following asymptotic equivalence:
\begin{align*}
    &\frac{1}{\sqrt{n}}\alpha_0^r X^{r(0)}_n
     + \frac{1}{\sqrt{n}}\sum_{i=1}^{p-1}\sum^{H_i}_{h=1}
       \parenv{\alpha_{i,h}^r X^{r(i,h)}_n + \alpha_{i,h}^c X^{c(i,h)}_n}
     \stackrel{L_1}{\sim}
    \\
    &\quad \sum_{j=1}^n \frac{1}{\sqrt{n}}
       \parenv{\alpha_0^r \nabla X^{r(0)}_j
         + \sum_{i=1}^{p-1}\sum^{H_i}_{h=1}
           \parenv{\beta^{r(i,h)}_{j,n}\overline{M}^{r(i,h)}_j
                 + \beta^{c(i,h)}_{j,n}\overline{M}^{c(i,h)}_j}},
\end{align*}
and thus the projection sum may be approximated by a sum of martingale differences,
allowing the same methodology as in the diagonalizable case to be applied here.

We observe that the error term is $O\!\parenv{n^{\frac{\lambda_i^r-k}{\tau-1}}\ln^{H_i-1}n}$
rather than $O\!\parenv{n^{\frac{\lambda_i^r-k}{\tau-1}}}$ as in the diagonalizable case.
This additional logarithmic factor poses no difficulty, since under our assumptions it remains $o\!\parenv{n^{1/2}}$
and therefore does not affect any of the subsequent proofs.

The remainder of the proof for the general case mirrors the structure of the diagonalizable case; the overall arguments proceed along the same lines.

\section{Conclusion}\label{sec:conc}
In this paper, we studied the asymptotic behavior of mutation systems by establishing a Central Limit Theorem (CLT) for the frequencies of $k$-tuples. 
Building upon prior work, our results fully characterize the stochastic fluctuations of sequence compositions around their deterministic expectations. 
By exploiting the spectral properties of the $k$-substitution matrix and projecting the centered count vectors, we successfully approximated the system using a martingale difference sequence. 
This framework allowed us to verify the classical martingale CLT conditions and to explicitly derive the limiting covariance matrix. 

The establishment of this CLT provides a mathematical foundation for statistical inference in evolving sequence systems. 
Moving beyond average-case behavior, our theorem enables the construction of confidence intervals and the formulation of hypothesis tests regarding sequence composition. 
In the context of \emph{in vivo} DNA-based data storage, quantifying these compositional fluctuations may be critical. 
Understanding that these variations are asymptotically normal allows for more accurate modeling of biological error channels, which is essential for evaluating the long-term reliability of stored information subject to continuous, probabilistic mutations. 
For example, since the limiting frequency depends entirely on the mutation process, one cannot expect perfect error protection over an indefinite number of mutations. Nevertheless, characterizing the convergence rate can guide the design of codes capable of protecting against mutation errors for a bounded duration (i.e., a bounded number of mutations), where the probability of decoding failure naturally increases over time. 

Several promising avenues for future research remain. First, while our analysis covers a broad class of mutation rules under standard spectral gap assumptions, extending this framework to systems where the generating matrix exhibits negative entries, analogous to generalized urn models without strict replacement, would capture scenarios where certain $k$-tuples actively degrade or are selected against over time. 
Second, generalizing these results to reducible substitution matrices is of particular interest, as it would significantly broaden the family of mutation rules known to exhibit this limiting behavior.  
Finally, investigating the exact rate of convergence and integrating the covariance calculations into the design and analysis of specialized error-correcting codes represent natural and important next step toward achieving highly robust DNA data storage.

\appendix

\subsection{Proof of Claim \ref{cl:beta_val}}
We start with finding the values for $\bbeta^r_{j,n},\bbeta^c_{j,n}$ and show that $\varepsilon_n$ is a small error term. 
The proof follows the same lines as in \cite{elishco2024longterm}. 
\begin{IEEEproof}[Proof of Claim \ref{cl:beta_val}] 
In \eqref{eq:betas} we defined the sum 
\begin{align}
\label{eq:v_hat_app2}
    \overline{v}_n &:= \sum^n_{j=1}\left(\bbeta^r_{j,n}\overline{M}^r_j + \bbeta^c_{j,n}\overline{M}^c_j \right) := \alpha^r X^r_n + \alpha^c X^c_n + \varepsilon_n
\end{align}

We will find $\bbeta^r_{j,n}$ and $\bbeta^c_{j,n}$ such that $\varepsilon_n$ is a small error term.
We begin by defining $\bbeta^r_{n,n}=\alpha^r, \bbeta^c_{n,n}=\alpha^c$ and choose the rest of the coefficients recursively to eliminate $X_{n-1},\dots,X_0$. 
We do that by substituting $\overline{M}^r_j$ and $\overline{M}^c_j$ with their values, isolate the parts of the equation with $X^r_{n-1}$ and $X^c_{n-1}$ and equate to 0.
\begin{align}
    &\beta^r_{n,n}\parenv{-X^r_{n-1}-\frac{1}{(n-1)(\tau-1)+Y_0}\parenv{(\lambda^r-k)X^r_{n-1}-\lambda^cX^c_{n-1} }} \nonumber
    \\ 
    &+\beta^c_{n,n}\parenv{-X^c_{n-1}-\frac{1}{(n-1)(\tau-1)+Y_0}\parenv{(\lambda^r-k)X^c_{n-1}+\lambda^cX^r_{n-1} }} \nonumber
    \\ 
    &+ \beta^r_{n-1,n}X^r_{n-1}+\beta^c_{n-1,n}X^c_{n-1}=0. \label{eq:reg case1}
\end{align} 
That way we find $\beta_{n-1,n}^r$ and $\beta_{n-1,n}^c$ that cancels out $X^r_{n-1}$ and $X^c_{n-1}$ 
\begin{align*}
    \beta_{n-1,n}^r&= \parenv{1+\frac{\lambda^r-k}{(n-1)(\tau-1)+Y_0}}\alpha^r + \frac{\lambda^c}{(n-1)(\tau-1)+Y_0}\alpha^c,
    \\ 
    \beta_{n-1,n}^c&= - \frac{\lambda^c}{(n-1)(\tau-1)+Y_0}\alpha^r + \parenv{1+\frac{\lambda^r-k}{(n-1)(\tau-1)+Y_0}}\alpha^c.
\end{align*}
Using the same method for all $j\leq n$ we get 
\begin{align*}
    \beta_{j-1,n}^r&=\parenv{1+\frac{\lambda^r-k}{(j-1)(\tau-1)+Y_0}}\beta^r_{j,n}+\frac{\lambda^c}{(n-1)(\tau-1)+Y_0}\beta^c_{j,n},
    \\
    \beta_{j-1,n}^c&= - \frac{\lambda^c}{(n-1)(\tau-1)+Y_0}\beta^r_{j,n}+\parenv{1+\frac{\lambda^r-k}{(j-1)(\tau-1)+Y_0}}\beta^c_{j,n}.
\end{align*}

To simplify the analysis, we write the recursion relation in vector form. 
Let $\bbeta_{j,n}=\bmat{\beta^r_{j,n}\\ \beta^c_{j,n}}$ and define $\bbeta_{n,n}=\bmat{\alpha^r\\ \alpha^c}$.
We get:  
\begin{align}
    \bbeta_{j-1,n}&= \parenv{\bbI+\frac{1}{(j-1)(\tau-1)+Y_0}\bbA}\bbeta_{j,n} \nonumber
    \\
    &= \prod_{t=j}^n\parenv{\bbI+\frac{1}{(t-1)(\tau-1)+Y_0}\bbA}\bmat{\alpha^r\\ \alpha^c}. \label{eq:beta app}
\end{align}
where $\bbI$ is the $2\times 2$ identity matrix and 
\[\bbA=\bmat{\lambda^r -k & \lambda^c \\ -\lambda^c & \lambda^r-k}.\] 

The eigenvalues of $\bbA$ are
\begin{align*}
    \rho_1&= \lambda^r-k+\bbi\lambda^c\\ 
    \rho_2&= \lambda^r-k-\bbi\lambda^c
\end{align*}
with eigenvectors 
\begin{align*}
    \bbv_1= \bmat{ 1\\ \bbi}\; \; \bbv_2= \bmat{1\\ -\bbi}
\end{align*} 
Notice that if $\lambda^c \neq 0$ then $\rho_1=\lambda-k$ and $\rho_2=\overline{\lambda}-k$, where $\overline{\lambda}$ is the complex conjugate of $\lambda$, and if $\lambda^c=0$ then $\rho_1=\rho_2=\lambda^r-k$, $\bbA$ is diagonal, and the eigenvectors are the standard basis vectors. 

We distinguish between the two cases. 
Let us assume first that $\lambda^c \neq 0$. 
Then 
\begin{align*}
    \bbv_1= \bmat{1\\ \bbi },\; \; \bbv_2= \bmat{1\\ -\bbi }.
\end{align*}
If $\bbV$ is the eigenvector matrix $\bmat{\bbv_1& \bbv_2}$ then we can write 
\[\bbA=\bbV\bmat{\lambda-k & 0 \\ 0 & \overline{\lambda}-k}\bbV^{-1}.\]
Plugging it back to \eqref{eq:beta app}, we get  
\begin{align}
    \bbeta_{j-1,n}&= \bbV \bmat{\prod_{t=j}^n\parenv{1+\frac{\lambda-k}{(t-1)(\tau-1)+Y_0}} & 0 \\ 0& \prod_{t=j}^n\parenv{1+\frac{\overline{\lambda}-k}{(t-1)(\tau-1)+Y_0}}} \bbV^{-1}\bmat{\alpha^r\\ \alpha^c}. \label{eq:beta mat} 
\end{align}
Using Stirling's approximation, we have
\begin{align*}
    \prod_{t=j}^n\parenv{1+\frac{\lambda-k}{(t-1)(\tau-1)+Y_0}}&=\prod_{t=j}^n\parenv{\frac{(t-1)+ 
    \frac{m+\lambda-k}{\tau-1}}{(t-1)+\frac{m}{\tau-1}}} = \parenv{\frac{n}{j}}^{\frac{\lambda-k}{\tau-1}}+O\parenv{n^{\frac{\lambda-k}{\tau-1}-1}}, 
\end{align*} 
and similarly 
\[\prod_{t=j}^n\parenv{1+\frac{\overline{\lambda}-k}{(t-1)(\tau-1)+Y_0}}=\parenv{\frac{n}{j}}^{\frac{\overline{\lambda}-k}{\tau-1}}+ 
O\parenv{n^{\frac{\overline{\lambda}-k}{\tau-1}-1}}.\]

Writing $\frac{n}{j}$ as $e^{\log (n/j)}$, we get 
\begin{align*}
\parenv{\frac{n}{j}}^{\frac{\lambda-k}{\tau-1}}&= 
\parenv{\frac{n}{j}}^{\frac{\lambda^r-k}{\tau-1}}\parenv{\frac{n}{j}}^{\bbi\frac{\lambda^c}{\tau-1}}
\\
&= \parenv{\frac{n}{j}}^{\frac{\lambda^r-k}{\tau-1}}\parenv{\cos\parenv{\frac{\lambda^c}{\tau-1}\log (n/j)} +\bbi\sin\parenv{\frac{\lambda^c}{\tau-1}\log (n/j)}},
\end{align*}

and similarly

\[\parenv{\frac{n}{j}}^{\frac{\overline{\lambda}-k}{\tau-1}}=\parenv{\frac{n}{j}}^{\frac{\lambda^r-k}{\tau-1}} \parenv{\cos\parenv{\frac{\lambda^c}{\tau-1}\log (n/j)} -\bbi\sin\parenv{\frac{\lambda^c}{\tau-1}\log (n/j)}}.\]

Plugging this back into \eqref{eq:beta mat}, and noticing that $\bbV^{-1}=\frac{1}{2}\bmat{1& -\bbi \\ 1& \bbi}$ we obtain the results:
\begin{align*}
    \bbeta_{j-1,n}&= \bmat{
    \parenv{\frac{n}{j}}^{\frac{\lambda^r-k}{\tau-1}} \parenv{\alpha^r\cos\parenv{\frac{\lambda^c}{\tau-1}\log (n/j)} +\alpha^c\sin\parenv{\frac{\lambda^c}{\tau-1}\log (n/j)}}+O\parenv{n^{\frac{\lambda^r-k}{\tau-1}-1}}\\ 
    \parenv{\frac{n}{j}}^{\frac{\lambda^r-k}{\tau-1}} \parenv{\alpha^c\cos\parenv{\frac{\lambda^c}{\tau-1}\log (n/j)} -\alpha^r\sin\parenv{\frac{\lambda^c}{\tau-1}\log (n/j)}}+O\parenv{n^{\frac{\lambda^r-k}{\tau-1}-1}}
    }.
\end{align*}

Similar calculations for the case of $\lambda^c=0$ yields 
\begin{align*}
    \bbeta_{j-1,n}&= \bmat{
    \alpha^r \parenv{\frac{n}{j}}^{\frac{\lambda^r-k}{\tau-1}} +O\parenv{n^{\frac{\lambda^r-k}{\tau-1}-1}}\\ 
    0 
    },
\end{align*}

because as we stated in Remark \ref{re:real_lambda_alpha_c}, in the case that $\lambda^c=0$, we set $\alpha^c$ to 0, meaning $\bbeta^c_{j-1,n} = 0$ for all $j$.
\\
For every $\lambda^c$ we get the upper bound:
\begin{align*}
|\bbeta_{j-1,n}|&\le \begin{bmatrix} \parenv{\alpha^r+\alpha^c} \parenv{\frac{n}{j}}^{\frac{\lambda^r-k}{\tau-1}} +O\parenv{n^{\frac{\lambda^r-k}{\tau-1}-1}}\\ \parenv{\alpha^r+\alpha^c} \parenv{\frac{n}{j}}^{\frac{\lambda^r-k}{\tau-1}} +O\parenv{n^{\frac{\lambda^r-k}{\tau-1}-1}} \end{bmatrix},
\end{align*}
which is crucial for most of the analysis in the work.
Plugging the results into \eqref{eq:v_hat_app2} we get for $\varepsilon_n$:

\[\varepsilon_n = -X^r_0\parenv{ \parenv{1+\frac{\lambda^r-k}{Y_0}}\beta^r_{1,n}+\frac{\lambda^c}{Y_0}\beta^c_{1,n}}-X^c_0\parenv{ \parenv{1+\frac{\lambda^r-k}{Y_0}}\beta^c_{1,n}-\frac{\lambda^c}{Y_0}\beta^r_{1,n} } = O\parenv{n^{\frac{\lambda^r-k}{\tau-1}}},\]

which is a small error term we will be able to neglect throughout the analysis.
\end{IEEEproof}

\subsection{Proof of Lemma \ref{lem:int_limit}}
The proof of Lemma \ref{lem:int_limit} is straightforward. 
\begin{IEEEproof}[Proof of Lemma \ref{lem:int_limit}]
We prove only \eqref{eq:int1}, the proof of \eqref{eq:int2} and \eqref{eq:int3} are analogous. 
First, we use the formula for the product of cosines to obtain 
\begin{align*} 
&\cos\parenv{a\log(t)}\cos\parenv{b\log(t)} \\ 
&= \frac{1}{2}\parenv{\cos\parenv{(a+b)\log(t)} +\cos\parenv{(a-b)\log(t)}}.
\end{align*} 
Thus, we are to calculate 
\begin{align}
\label{eq:int11}
\frac{1}{2}\int_0^1 t^{-r}\parenv{\cos\parenv{(a+b)\log(t)} +\cos\parenv{(a-b)\log(t)}}dt.
\end{align}
Let us focus on the first part. By taking $u=\log (t)$, we obtain 
\begin{align*}
\int_0^1 t^{-r}\cos\parenv{(a+b)\log(t)}dt &= \int_{-\infty}^0 e^{-ru}\cos\parenv{(a+b)u} e^u du\\ 
&= \int_{-\infty}^0 e^{(1-r)u}\cos\parenv{(a+b)u} du.
\end{align*} 
Notice that 
\[\int e^{(1-r)u}\cos\parenv{(a\pm b)u}du=\frac{e^{(1-r)u}}{(1-r)^2+(a\pm b)^2}\parenv{(1-r)\cos\parenv{(a\pm b)u} + (a\pm b)\sin\parenv{(a\pm b)u}}.\] 
Thus, since $r<1$, we get 
\begin{align*}
\int_{-\infty}^0 e^{(1-r)u}\cos\parenv{(a+b)u} du= \frac{1-r}{(1-r)^2+(a+b)^2},
\end{align*} 
and 
\begin{align*}
\int_{-\infty}^0 e^{(1-r)u}\cos\parenv{(a-b)u} du= \frac{1-r}{(1-r)^2+(a-b)^2}. 
\end{align*} 
Plugging everything together into \eqref{eq:int11}, we obtain \eqref{eq:int1}.
\end{IEEEproof}

\subsection{Proof of Lemma \ref{lem:finite_well_def}}
We now prove Lemma \ref{lem:finite_well_def} using Lemma \ref{lem:int_limit}.
\begin{IEEEproof}[Proof of Lemma \ref{lem:finite_well_def}]
Assume first that $\lambda_i,\lambda_l$ are purely real. 
We write explicitly 
\begin{align}
\beta_{j,n}^{r(i)}\beta_{j,n}^{r(l)}&= 
\parenv{\alpha_i^r \parenv{\frac{n}{j+1}}^{\frac{\lambda_i^r-k}{\tau-1}} +O\parenv{n^{\frac{\lambda_i^r-k}{\tau-1}-1}}} \parenv{\alpha_l^r \parenv{\frac{n}{j+1}}^{\frac{\lambda_l^r-k}{\tau-1}} +O\parenv{n^{\frac{\lambda_l^r-k}{\tau-1}-1}}} \nonumber \\ 
&= \parenv{\alpha_i^r\alpha_l^r \parenv{\frac{n}{j+1}}^{\frac{\lambda_i^r-k}{\tau-1} + \frac{\lambda_l^r-k}{\tau-1}} +\frac{1}{j+1}O\parenv{n^{\frac{\lambda_i^r-k}{\tau-1} + \frac{\lambda_l^r-k}{\tau-1}-1}}}. \label{eq:BB2}
\end{align}
Normalizing and summing over $j$, we use \eqref{eq:BB2} and our assumption on the eigenvalues to obtain 
\begin{align} 
&\lim_{n \to \infty}\frac{1}{n} \sum^n_{j=1} \beta^{r(i)}_{j,n} \beta^{r(l)}_{j,n} \nonumber \\  
&= \lim_{n \to \infty} \alpha_i^r \alpha_l^r \frac{1}{n} \sum^n_{j=1} \parenv{\frac{n}{j+1}}^{\frac{\lambda_i^r-k}{\tau-1} + \frac{\lambda_l^r-k}{\tau-1}} + o(1). \label{eq:BB1}
\end{align}

By the assumption on the non-leading eigenvalues, we may find the limit of \eqref{eq:BB1} 
\begin{align*}
\lim_{n \to \infty} \alpha_i^r \alpha_l^r \frac{1}{n} \sum^n_{j=1} \parenv{\frac{n}{j+1}}^{\frac{\lambda_i^r-k}{\tau-1} + \frac{\lambda_l^r-k}{\tau-1}} = \frac{\alpha_i^r \alpha_l^r}{1 - \frac{\lambda_i^r-k}{\tau-1} - \frac{\lambda_l^r-k}{\tau-1}}.
\end{align*}
This proves \eqref{eq:aa_r}. 

Now assume $\lambda_i,\lambda_l$ may by complex. 
Our goal is to show that the sum $\frac{1}{n} \sum^n_{j=1} \beta^{r(i)}_{j,n} \beta^{r(l)}_{j,n}$ converges as $n\to\infty$.  
To that end, we estimate the limit of the sum using integration. 
We first write $\beta_{j,n}^{r(i)}$ slightly different 
\begin{align*}
&\beta_{j,n}^{r(i)} = \parenv{\frac{1}{\frac{j+1}{n}}}^{\frac{\lambda_i^r-k}{\tau-1}} \parenv{\alpha_i^r\cos\parenv{\frac{\lambda_i^c}{\tau-1} 
\log \parenv{\frac{j+1}{n}}} -\alpha_i^c\sin\parenv{\frac{\lambda_i^c}{\tau-1}\log \parenv{\frac{j+1}{n}}}} +O\parenv{n^{\frac{\lambda_i^r-k}{\tau-1}-1}}.
\end{align*}
The product $\beta_{j,n}^{r(i)}\beta_{j,n}^{r(l)}$ comprises of several similar arguments. For simplicity, we estimate one at a time. 
We first consider the argument 
\begin{align*}
&\parenv{\frac{1}{\frac{j+1}{n}}}^{\frac{\lambda_i^r-k}{\tau-1}}\parenv{\frac{1}{\frac{j+1}{n}}}^{\frac{\lambda_l^r-k}{\tau-1}} 
\alpha_i^r\cos\parenv{\frac{\lambda_i^c}{\tau-1} \log \parenv{\frac{j+1}{n}}} 
\alpha_l^r\cos\parenv{\frac{\lambda_l^c}{\tau-1} \log \parenv{\frac{j+1}{n}}} \\ 
&=\parenv{\frac{1}{\frac{j+1}{n}}}^{\frac{\lambda_i^r+\lambda_l^r-2k}{\tau-1}}  
\alpha_i^r \alpha_l^r \cos\parenv{\frac{\lambda_i^c}{\tau-1} \log \parenv{\frac{j+1}{n}}} 
\cos\parenv{\frac{\lambda_l^c}{\tau-1} \log \parenv{\frac{j+1}{n}}}. 
\end{align*} 
When taking $n\to\infty$ and summing over $j$, we notice that the expression  
\begin{align}
\label{eq:temp1}
&\lim_{n\to\infty} \frac{1}{n}\sum_{j=1}^n \parenv{\frac{1}{\frac{j+1}{n}}}^{\frac{\lambda_i^r+\lambda_l^r-2k}{\tau-1}}  
\alpha_i^r \alpha_l^r \cos\parenv{\frac{\lambda_i^c}{\tau-1} \log \parenv{\frac{j+1}{n}}} 
\cos\parenv{\frac{\lambda_l^c}{\tau-1} \log \parenv{\frac{j+1}{n}}} 
\end{align} 
is the upper Riemann sum of the integral
\begin{align}
\label{eq:temp2}
&= \alpha_i^r \alpha_l^r \int_0^1 \parenv{\frac{1}{t}}^{\frac{\lambda_i^r+\lambda_l^r-2k}{\tau-1}} 
\parenv{\cos\parenv{\frac{\lambda_i^c}{\tau-1} \log \parenv{t}} 
\cos\parenv{\frac{\lambda_l^c}{\tau-1} \log \parenv{t}}} dt.
\end{align}
From Lemma \ref{lem:int_limit}, the integral exists, which implies that the limit in \eqref{eq:temp1} is indeed equal to the integral in \eqref{eq:temp2}. 
Using Lemma \ref{lem:int_limit} to solve the integral, we obtain 
\begin{align*}
&\lim_{n\to\infty} \frac{1}{n}\sum_{j=1}^n \parenv{\frac{1}{\frac{j+1}{n}}}^{\frac{\lambda_i^r+\lambda_l^r-2k}{\tau-1}}  
\alpha_i^r \alpha_l^r \cos\parenv{\frac{\lambda_i^c}{\tau-1} \log \parenv{\frac{j+1}{n}}} 
\cos\parenv{\frac{\lambda_l^c}{\tau-1} \log \parenv{\frac{j+1}{n}}} \\ 
&= \frac{\alpha_i^r \alpha_l^r\parenv{1-\frac{\lambda_i^r+\lambda_l^r-2k}{\tau-1}}}{2}\parenv{\frac{1}{\parenv{1-\frac{\lambda_i^r+\lambda_l^r-2k}{\tau-1}}^2+\parenv{\frac{\lambda_i^c-\lambda_l^c}{\tau-1}}^2}+\frac{1}{\parenv{1-\frac{\lambda_i^r+\lambda_l^r-2k}{\tau-1}}^2+\parenv{\frac{\lambda_i^c+\lambda_l^c}{\tau-1}}^2}}.
\end{align*} 
Similarly, we have 
\begin{align*}
&\lim_{n\to\infty} \frac{1}{n}\sum_{j=1}^n \parenv{\frac{1}{\frac{j+1}{n}}}^{\frac{\lambda_i^r+\lambda_l^r-2k}{\tau-1}}  
\alpha_i^r \alpha_l^c \cos\parenv{\frac{\lambda_i^c}{\tau-1} \log \parenv{\frac{j+1}{n}}} 
\sin\parenv{\frac{\lambda_l^c}{\tau-1} \log \parenv{\frac{j+1}{n}}} \\ 
&= \frac{\alpha_i^r \alpha_l^c}{2}\parenv{\frac{\frac{\lambda_i^c+\lambda_l^c}{\tau-1}}{\parenv{1-\frac{\lambda_i^r+\lambda_l^r-2k}{\tau-1}}^2+\parenv{\frac{\lambda_i^c+\lambda_l^c}{\tau-1}}^2}- 
\frac{\frac{\lambda_i^c-\lambda_l^c}{\tau-1}}{\parenv{1-\frac{\lambda_i^r+\lambda_l^r-2k}{\tau-1}}^2+\parenv{\frac{\lambda_i^c-\lambda_l^c}{\tau-1}}^2}}, 
\end{align*} 
and 
\begin{align*}
&\lim_{n\to\infty} \frac{1}{n}\sum_{j=1}^n \parenv{\frac{1}{\frac{j+1}{n}}}^{\frac{\lambda_i^r+\lambda_l^r-2k}{\tau-1}}  
\alpha_i^c \alpha_l^c \sin\parenv{\frac{\lambda_i^c}{\tau-1} \log \parenv{\frac{j+1}{n}}} 
\sin\parenv{\frac{\lambda_l^c}{\tau-1} \log \parenv{\frac{j+1}{n}}} \\ 
&= \frac{\alpha_i^c \alpha_l^c\parenv{1-\frac{\lambda_i^r+\lambda_l^r-2k}{\tau-1}}}{2}\parenv{\frac{1}{\parenv{1-\frac{\lambda_i^r+\lambda_l^r-2k}{\tau-1}}^2+\parenv{\frac{\lambda_i^c-\lambda_l^c}{\tau-1}}^2}-\frac{1}{\parenv{1-\frac{\lambda_i^r+\lambda_l^r-2k}{\tau-1}}^2+\parenv{\frac{\lambda_i^c+\lambda_l^c}{\tau-1}}^2}}.
\end{align*}
For simplicity, we denote $a=\parenv{1-\frac{\lambda_i^r+\lambda_l^r-2k}{\tau-1}}$, $b=\frac{\lambda_i^c-\lambda_l^c}{\tau-1}$ and $c=\frac{\lambda_i^c+\lambda_l^c}{\tau-1}$. 
Putting everything together, we obtain 
\begin{align*}
&\lim_{n\to\infty} \frac{1}{n}\sum_{j=1}^n\beta_{j,n}^{r(i)}\beta_{j,n}^{r(l)}\\ 
&= \frac{a\parenv{\alpha_i^r \alpha_l^r-\alpha_i^c \alpha_l^c}- c\parenv{\alpha_i^r \alpha_l^c+\alpha_i^c \alpha_l^r}}{2(a^2+c^2)} 
+\frac{a\parenv{\alpha_i^r \alpha_l^r+\alpha_i^c \alpha_l^c}+ b\parenv{\alpha_i^r \alpha_l^c-\alpha_i^c \alpha_l^r}}{2(a^2+b^2)}, 
\end{align*}
which is \eqref{eq:aa1}.

Similar calculations are used to prove \eqref{eq:aa2} and \eqref{eq:aa3}. 
\end{IEEEproof}

\subsection{Proof of Lemma \ref{lem:finite_well_def_x0}}
\begin{IEEEproof}[Proof of Lemma \ref{lem:finite_well_def_x0}]
We have already found the limit of the sum
\[\lim_{n\to\infty}\frac{1}{n} \sum_{j=1}^{n} \beta_{j,n}^{r(i)} \beta_{j,n}^{r(l)},\]
but we would like to replace $\beta_{j,n}^{r(i)}$ with $\alpha_0^r$. We know that
\[\beta^{r(i)}_{j,n} = \parenv{\frac{n}{j+1}}^{\frac{\lambda_i^r-k}{\tau-1}} \parenv{\alpha_i^r\cos\parenv{\frac{\lambda_i^c}{\tau-1}\log\parenv{\frac{n}{j+1}}} +\alpha_i^c\sin\parenv{\frac{\lambda_i^c}{\tau-1}\log\parenv{\frac{n}{j+1}}}}+O\parenv{n^{\frac{\lambda_i^r-k}{\tau-1}-1}}.\]
Meaning, if we set $\lambda_i^r=k$, $\lambda_i^c=0$, $\alpha_i^r=\alpha_0^r$, $\alpha_0^r=0$ we get:
\begin{align*}
    \lim_{n\to\infty}&\frac{1}{n} \sum_{j=1}^{n} \beta_{j,n}^{r(i)} \beta_{j,n}^{r(l)} = \lim_{n\to\infty}\frac{1}{n} \sum_{j=1}^{n} \parenv{\alpha_0^r + O\parenv{n^{-1}}} \beta_{j,n}^{r(l)} = \lim_{n\to\infty}\frac{1}{n} \sum_{j=1}^{n} \alpha_0^r \beta_{j,n}^{r(l)} + O\parenv{n^{\frac{\lambda_l^r-k}{\tau-1}-1}} = \lim_{n\to\infty}\frac{1}{n} \sum_{j=1}^{n} \alpha_0^r \beta_{j,n}^{r(l)}.
\end{align*}
Using \ref{lem:finite_well_def} and setting the appropriate variables, we conclude the proof.
\end{IEEEproof}

\subsection{Proof of Claim \ref{cl:beta_val_gen}}
We will now find $\beta^{r(i,h)}_{j,n}$ and $\beta^{c(i,h)}_{j,n}$ for the general case of non-diagonalizable substitution matrix $\bbM^{(k)}$. The process is similar to the regular case, so we follow the same idea as in the proof of Claim \ref{cl:beta_val}. 
\begin{IEEEproof}[Proof of Claim \ref{cl:beta_val_gen}]
For the general case presented in Subsection \ref{sec:non_diag}, we consider a Jordan chain
of left generalized eigenvectors associated with a single Jordan block in the Jordan normal
form of the substitution matrix $\bbM^{(k)}$
\[\parenv{\bbu_{i,1}, \bbu_{i,2}, \bbu_{i,3}, \dots, \bbu_{i,H_i}}\] 
such that $\bbu_{i,1}M^{(k)} = \lambda_i\bbu_{i,1}$ and $\bbu_{i,h}M^{(k)} = \lambda_i\bbu_{i,h} + \bbu_{i,h-1}$ for $h>1$. 
For convenience, we can define $\bbu_{i,0}$ to be the zero vector, $\bbu_{i,0} = \boldsymbol{0}$ so:
\[ \bbu_{i,h}M^{(k)} = \lambda_i\bbu_{i,h} + \bbu_{i,h-1} ~\text{ for all }~ h \ge 1.\]

We will now find the coefficients $\beta^{r(i,h)}_{j,n}$ and $\beta^{c(i,h)}_{j,n}$ such that:
\begin{align*}\overline{M}^{r(i,h)}_n
        &:= \nabla X^{r(i,h)}_n
            - \frac{1}{(n-1)(\tau-1)+Y_0}\!\parenv{(\lambda_i^r-k)X_{n-1}^{r(i,h)}
              - \lambda_i^c X_{n-1}^{c(i,h)} + X_{n-1}^{r(i,h-1)}},
    \\
    \overline{M}^{c(i,h)}_n
        &:= \nabla X^{c(i,h)}_n
            - \frac{1}{(n-1)(\tau-1)+Y_0}\!\parenv{(\lambda_i^r-k)X_{n-1}^{c(i,h)}
              + \lambda_i^c X_{n-1}^{r(i,h)} + X_{n-1}^{c(i,h-1)}},
    \\
    \overline{v}^{(i)}_n &:= \sum^{H_i}_{h=1}\sum^n_{j=1}\parenv{\beta^{r(i,h)}_{j,n}\overline{M}^{r(i,h)}_j + \beta^{c(i,h)}_{j,n}\overline{M}^{c(i,h)}_j} := \sum^H_{h=1} \alpha^r_{i,h} X^{r(i,h)}_n + \alpha^c_{i,h} X^{c(i,h)}_n + \varepsilon_n,
\end{align*}
with $\varepsilon_n$ being a small error term that we can neglect later.
We would find $\beta^{r(i,h)}_{j,n}$ and $\beta^{c(i,h)}_{j,n}$ by treating every value of $h$ separately. We start with $h = H_i$, we can define $\beta^{r(i,H_i)}_{j,n}$ and $\beta^{c(i,H_i)}_{j,n}$ as $\beta^{r(i)}_{j,n}$ and $\beta^{c(i)}_{j,n}$ of the regular case to get:

\begin{align*} &\sum^n_{j=1}\parenv{\beta^{r(i,H_i)}_{j,n}\overline{M}^{r(i,H_i)}_j + \beta^{c(i,H_i)}_{j,n}\overline{M}^{c(i,H_i)}_j} =
\\
&\alpha^r_{i,H_i} X^{r(i,H_i)}_n + \alpha^c_{i,H_i} X^{c(i,H_i)}_n + \varepsilon^{(i,H_i)}_n - \sum^n_{j=1} \frac{1}{(j-1)(\tau-1) + Y_0}\parenv{\beta^{r(i,H_i)}_{j,n}X^{r(i,H_i-1)}_{j-1} + \beta^{c(i,H_i)}_{j,n}X^{c(i,H_i-1)}_{j-1}}.
\end{align*}
We notice we got the output we expected with the small error term of $\varepsilon^{(i,H_i)}_n$, but with residuals from the projections on $\bbu_{i,H_i-1}$. We will carry those residuals to the sum of $h=H_i-1$, and find $\beta^{r(i,H_i-1)}_{j,n}$ and $\beta^{c(i,H_i-1)}_{j,n}$ such that:

\begin{align}
&\sum^n_{j=1}\parenv{\beta^{r(i,H_i-1)}_{j,n}\overline{M}^{r(i,H_i-1)}_j + \beta^{c(i,H_i-1)}_{j,n}\overline{M}^{c(i,H_i-1)}_j} \nonumber
\\
&- \sum^n_{j=1} \frac{1}{(j-1)(\tau-1) + Y_0}\parenv{\beta^{r(i,H_i)}_{j,n}X^{r(i,H_i-1)}_{j-1} + \beta^{c(i,H_i)}_{j,n}X^{c(i,H_i-1)}_{j-1}} = \label{eq:gen case1}
\\
&\alpha^r_{i,H_i-1} X^{r(i,H_i-1)}_n + \alpha^c_{i,H_i-1} X^{c(i,H_i-1)}_n + \varepsilon^{(i,H_i-1)}_n \nonumber
\\
&- \sum^n_{j=1} \frac{1}{(j-1)(\tau-1) + Y_0}\parenv{\beta^{r(i,H_i-1)}_{j,n}X^{r(i,H_i-2)}_{j-1} + \beta^{c(i,H_i-1)}_{j,n}X^{c(i,H_i-2)}_{j-1}}. \label{eq:gen case2}
\end{align}

We choose $\beta^{r(i,H_i-1)}_{n,n} := \alpha^r_{i,H_i-1}$ and $\beta^{c(i,H_i-1)}_{n,n} := \alpha^c_{i,H_i-1}$ and find  $\beta^{r(i,H_i-1)}_{n-1,n}$ and $\beta^{c(i,H_i-1)}_{n-1,n}$ by focusing on $X^{r(i,H_i-1)}_{n-1}$ and $X^{c(i,H_i-1)}_{n-1}$ in \eqref{eq:gen case1}, as we did in \eqref{eq:reg case1} in the regular case
\begin{align*}
    &\beta^{r(i,H_i-1)}_{n,n}\parenv{-X^{r(i,H_i-1)}_{n-1}-\frac{1}{(n-1)(\tau-1)+Y_0}\parenv{(\lambda^r-k)X^{r(i,H_i-1)}_{n-1}-\lambda^c X^{c(i,H_i-1)}_{n-1} }}
    \\ 
    +&\beta^{c(i,H_i-1)}_{n,n}\parenv{-X^{c(i,H_i-1)}_{n-1}-\frac{1}{(n-1)(\tau-1)+Y_0}\parenv{(\lambda^r-k)X^{c(i,H_i-1)}_{n-1}+\lambda^c X^{r(i,H_i-1)}_{n-1} }}
    \\ 
    +&\beta^{r(i,H_i-1)}_{n-1,n}X^{r(i,H_i-1)}_{n-1}+\beta^{c(i,H_i-1)}_{n-1,n}X^{c(i,H_i-1)}_{n-1}
    \\
    -&\frac{1}{(n-1)(\tau-1) + Y_0}\parenv{\beta^{r(i,H)}_{n,n}X^{r(i,H_i-1)}_{n-1} + \beta^{c(i,H)}_{n,n}X^{c(i,H_i-1)}_{n-1}}=0.
\end{align*}
Isolating the betas, we find their values to be
\begin{align*}
    \beta^{r(i,H_i-1)}_{n-1,n}&:= \parenv{1+\frac{\lambda^r-k}{(n-1)(\tau-1)+Y_0}} \alpha^r_{i,H_i-1} + \frac{\lambda^c}{(n-1)(\tau-1)+Y_0} \alpha^c_{i,H_i-1} + \frac{1}{(n-1)(\tau-1) + Y_0} \alpha^r_{i,H_i},
    \\ 
    \beta^{c(i,H_i-1)}_{n-1,n}&:= - \frac{\lambda^c}{(n-1)(\tau-1)+Y_0} \alpha^r_{i,H_i-1} + \parenv{1+\frac{\lambda^r-k}{(n-1)(\tau-1)+Y_0}} \alpha^c_{i,H_i-1} + \frac{1}{(n-1)(\tau-1) + Y_0} \alpha^c_{i,H_i}.
\end{align*}

In general, for all $0 < j \le n$ and $1 < h \le H_i$ we obtain 
\begin{align*}
    \beta^{r(i,h-1)}_{j-1,n}&= \parenv{1+\frac{\lambda^r-k}{(j-1)(\tau-1)+Y_0}} \beta^{r(i,h-1)}_{j,n} + \frac{\lambda^c}{(j-1)(\tau-1)+Y_0} \beta^{c(i,h-1)}_{j,n} + \frac{1}{(j-1)(\tau-1) + Y_0} \beta^{r(i,h)}_{j,n}
    \\ 
    \beta^{c(i,h-1)}_{j-1,n}&= - \frac{\lambda^c}{(j-1)(\tau-1)+Y_0} \beta^{r(i,h-1)}_{j,n} + \parenv{1+\frac{\lambda^r-k}{(j-1)(\tau-1)+Y_0}} \beta^{c(i,h-1)}_{j,n} + \frac{1}{(j-1)(\tau-1) + Y_0} \beta^{c(i,h)}_{j,n}
\end{align*}

To simplify the analysis, we write the recursion relation in vector form. 
Define the vector $\bbeta^{(i)}_{j,n}$ of length $2H_i$
\[ \bbeta^{(i)}_{j,n}:=\bmat{\beta^{r(i,1)}_{j,n}\\
\beta^{c(i,1)}_{j,n}\\ 
\vdots\\
\beta^{r(i,H_i-1)}_{j,n}\\ \beta^{c(i,H_i-1)}_{j,n}\\ 
\beta^{r(i,H_i)}_{j,n}\\ \beta^{c(i,H_i)}_{j,n}}, ~
\bbeta^{(i)}_{n,n}:=\bmat{\alpha^r_{i,1}\\
\alpha^c_{i,1}\\
\vdots\\
\alpha^r_{i,H_i-1}\\
\alpha^c_{i,H_i-1}\\
\alpha^r_{i,H_i}\\
\alpha^c_{i,H_i}\\}.\]

We get the recursive equation in a vector form:  
\begin{align}
    \bbeta^{(i)}_{j-1,n}&= \parenv{\bbI+\frac{1}{(j-1)(\tau-1)+Y_0}\bbA}\bbeta^{(i)}_{j,n} \nonumber
    \\
    &= \prod_{t=j}^n \parenv{\bbI+\frac{1}{(t-1)(\tau-1)+Y_0}\bbA}\bbeta^{(i)}_{n,n}, \label{eq: gen case 2}
\end{align}
as $\bbI$ is the identity matrix and $\bbA$ is a $2 H_i \times 2 H_i$ matrix
\[\bbA := \bmat{
&\lambda^r-k &\lambda^c &1 &0 &0 &0 &\dots \\
&-\lambda^c &\lambda^r-k &0 &1 &0 &0 &\dots \\
 &0 &0 &\lambda^r-k &\lambda^c &1 &0 &\dots \\
&0 &0 &-\lambda^c &\lambda^r-k &0 &1 &\dots \\
&\vdots &\vdots &\vdots &\vdots &\vdots &\vdots &\ddots}.\]

We compute a Jordan decomposition of $\bbA$ and get the following 
\[\bbV^{-1} \bbA \bbV = \bmat{&\bbJ_{\lambda-k} &\bold{0}\\
&\bold{0} &\bbJ_{\overline{\lambda}-k} }.\]

While block $\bbJ_\lambda$ is an $H_i \times H_i$ matrix in the form of
\[\bbJ_\lambda = \bmat{&\lambda &1 &0 &0 &\dots\\
&0 &\lambda &1 &0 &\dots\\
&0 &0 &\lambda &1 &\dots\\
&0 &0 &0 &\lambda &\dots\\
&\vdots &\vdots &\vdots &\vdots &\ddots},\]

$\overline{\lambda}$ is the complex conjugate of $\lambda$, and $\bbV$ is a $2 H_i \times 2 H_i$ matrix in the form of
\[\bbV = \bmat{&-\bbi &0 &0 &\dots &\bbi &0 &0 &\dots\\
&1 &0 &0 &\dots &1 &0 &0 &\dots\\
&0 &-\bbi &0 &\dots &0 &\bbi &0 &\dots\\
&0 &1 &0 &\dots &0 &1 &0 &\dots\\
&0 &0 &-\bbi &\dots &0 &0 &\bbi &\dots\\
&0 &0 &1 &\dots &0 &0 &1 &\dots\\
&\vdots &\vdots &\vdots &\vdots &\vdots &\vdots &\vdots &\ddots\\},\]

and the inverse of $\bbV$ is in the form of
\[\bbV^{-1} = \frac{1}{2}\bmat{&\bbi &1 &0 &0 &0 &0 &\dots\\
&0 &0 &\bbi &1 &0 &0 &\dots\\
&0 &0 &0 &0 &\bbi &1 &\dots\\
&\vdots &\vdots &\vdots &\vdots &\vdots &\vdots &\dots\\
&-\bbi &1 &0 &0 &0 &0 &\dots\\
&0 &0 &-\bbi &1 &0 &0 &\dots\\
&0 &0 &0 &0 &-\bbi &1 &\dots\\
&\vdots &\vdots &\vdots &\vdots &\vdots &\vdots &\ddots\\}.\]

Plugging it back to \eqref{eq: gen case 2} we get 
\begin{align}
    &\bbeta^{(i)}_{j-1,n} = \bbV \bmat{\prod_{t=j}^{n} \parenv{\bbI+\frac{1}{(t-1)(\tau-1)+Y_0}\bbJ_{\lambda-k}} & \bold{0} \\ \bold{0}& \prod_{t=j}^{n} \parenv{\bbI+\frac{1}{(t-1)(\tau-1)+Y_0}\bbJ_{\overline{\lambda}-k}}} \bbV^{-1}\bbeta^{(i)}_{n,n}. \label{eq: gen case 3}
\end{align}
Focusing on a single Jordan block, we have
\[\bbI+\frac{1}{(t-1)(\tau-1)+Y_0}\bbJ_{\lambda-k} = \frac{(t-1) + \frac{Y_0+\lambda-k}{\tau-1}}{(t-1) + \frac{Y_0}{\tau - 1}}
    \bmat{&1 & \frac{1}{(t-1)(\tau-1) + Y_0 +\lambda-k} &0 &0 &\dots\\
    & 0 &1 & \frac{1}{(t-1)(\tau-1) + Y_0 +\lambda-k} &0 &\dots\\
    &0 &0 &1 & \frac{1}{(t-1)(\tau-1) + Y_0 +\lambda-k} &\dots\\
    &\vdots &\vdots &\vdots &\ddots &\ddots\\},\]

Then we get:
\begin{align}
    \prod_{t=j}^{n} \parenv{\bbI+\frac{1}{(t-1)(\tau-1)+Y_0}\bbJ_{\lambda-k}} = \prod_{t=j}^{n} \frac{(t-1) + \frac{Y_0+\lambda-k}{\tau-1}}{(t-1) + \frac{Y_0}{\tau - 1}} \bmat{&1 &e^1_{\lambda,j} &e^2_{\lambda,j} &e^3_{\lambda,j} &\dots\\
    & 0 &1 & e^1_{\lambda,j} &e^2_{\lambda,j} &\dots\\
    &0 &0 &1 &e^1_{\lambda,j} &\dots\\
    &\vdots &\vdots &\vdots &\ddots &\ddots\\}, \label{eq:gen case 4}
\end{align}

with $e^h_{\lambda,j}$ defined recursively like so:
\begin{align*}
    &e^1_{\lambda,j}:=\sum_{t_1=j}^{n} \frac{1}{(t_1-1)(\tau-1) + Y_0 +\lambda-k},
    \\
    &e^2_{\lambda,j}:=\sum_{t_2=j}^{n-1} \parenv{\frac{1}{(t_2-1)(\tau-1) + Y_0 +\lambda-k} \sum_{t_1=t_2+1}^{n} \frac{1}{(t_1-1)(\tau-1) + Y_0 +\lambda-k}},
    \\
    &e^3_{\lambda,j}:=\sum_{t_3=j}^{n-2} \parenv{\frac{1}{(t_3-1)(\tau-1) + Y_0 +\lambda-k} \sum_{t_2=t_3+1}^{n-1} \parenv{\frac{1}{(t_2-1)(\tau-1) + Y_0 +\lambda-k} \sum_{t_1=t_2+1}^{n} \frac{1}{(t_1-1)(\tau-1) + Y_0 +\lambda-k}}},
    \\
    &e^h_{\lambda,j}:=\sum_{t=j}^{n-h+1} \frac{1}{(t-1)(\tau-1) + Y_0 +\lambda-k} e^{h-1}_{\lambda,t+1} ~\text{for}~ n-h+1 \ge j ~\text{else:}~ 0.
\end{align*}

We use the Stirling's approximation on \eqref{eq:gen case 4}, and get the following:
\[\prod_{t=j}^{n} \frac{(t-1) + \frac{Y_0+\lambda-k}{\tau-1}}{(t-1) + \frac{Y_0}{\tau - 1}} = 
\parenv{\frac{n}{j}}^{\frac{\lambda-k}{\tau-1}} + O\parenv{n^{\frac{\lambda-k}{\tau-1}-1}}. 
\]

Plugging it back to \eqref{eq: gen case 3} we get

\begin{align}
&\bbV \bmat{\prod_{t=j}^{n} \parenv{\bbI+\frac{1}{t(\tau-1)+Y_0}\bbJ_{\lambda-k}} & 0 \\ 0& \prod_{t=j}^{n} \parenv{\bbI+\frac{1}{t(\tau-1)+Y_0}\bbJ_{\overline{\lambda}-k}}} \bbV^{-1} = \nonumber
\\
&\parenv{\parenv{\frac{n}{j}}^{\frac{\lambda-k}{\tau-1}} + O\parenv{n^{\frac{\lambda-k}{\tau-1}-1}}} \bmat{&-\bbi &0 &0 &\dots\\
&1 &0 &0 &\dots\\
&0 &-\bbi &0 &\dots\\
&0 &1 &0 &\dots\\
&0 &0 &-\bbi &\dots\\
&0 &0 &1 &\dots\\
&\vdots &\vdots &\vdots &\ddots\\}
\bmat{&1 &e^1_{\lambda,j} &e^2_{\lambda,j} &e^3_{\lambda,j} &\dots\\
& 0 &1 & e^1_{\lambda,j} &e^2_{\lambda,j} &\dots\\
&0 &0 &1 &e^1_{\lambda,j} &\dots\\
&\vdots &\vdots &\vdots &\ddots &\ddots\\}
\frac{1}{2}\bmat{&\bbi &1 &0 &0 &0 &0 &\dots\\
&0 &0 &\bbi &1 &0 &0 &\dots\\
&0 &0 &0 &0 &\bbi &1 &\dots\\
&\vdots &\vdots &\vdots &\vdots &\vdots &\vdots &\ddots\\} + \nonumber
\\
&\parenv{\parenv{\frac{n}{j}}^{\frac{\overline{\lambda}-k}{\tau-1}} + O\parenv{n^{\frac{\overline{\lambda}-k}{\tau-1}-1}}}
\bmat{&\bbi &0 &0 &\dots\\
&1 &0 &0 &\dots\\
&0 &\bbi &0 &\dots\\
&0 &1 &0 &\dots\\
&0 &0 &\bbi &\dots\\
&0 &0 &1 &\dots\\
&\vdots &\vdots &\vdots &\ddots\\}
\bmat{&1 &e^1_{\overline{\lambda},j} &e^2_{\overline{\lambda},j} &e^3_{\overline{\lambda},j} &\dots\\
& 0 &1 & e^1_{\overline{\lambda},j} &e^2_{\overline{\lambda},j} &\dots\\
&0 &0 &1 &e^1_{\overline{\lambda},j} &\dots\\
&\vdots &\vdots &\vdots &\ddots &\ddots\\}
\frac{1}{2}
\bmat{&-\bbi &1 &0 &0 &0 &0 &\dots\\
&0 &0 &-\bbi &1 &0 &0 &\dots\\
&0 &0 &0 &0 &-\bbi &1 &\dots\\
&\vdots &\vdots &\vdots &\vdots &\vdots &\vdots &\ddots\\} = \nonumber
\\
&\parenv{\frac{1}{2}\parenv{\frac{n}{j}}^{\frac{\lambda-k}{\tau-1}} + O\parenv{n^{\frac{\lambda-k}{\tau-1}-1}}}
\bmat{&1 &-\bbi &e^1_{\lambda,j} &-\bbi e^1_{\lambda,j} &e^2_{\lambda,j} &-\bbi e^2_{\lambda,j} &e^3_{\lambda,j} &-\bbi e^3_{\lambda,j} &\dots\\
&\bbi &1 &\bbi e^1_{\lambda,j} &e^1_{\lambda,j} &\bbi e^2_{\lambda,j}  &e^2_{\lambda,j} &\bbi e^3_{\lambda,j} &e^3_{\lambda,j} &\dots\\
&0 &0 &1 &-\bbi &e^1_{\lambda,j} &-\bbi e^1_{\lambda,j} &e^2_{\lambda,j} &-\bbi e^2_{\lambda,j} &\dots\\
&0 &0 &\bbi &1 &\bbi e^1_{\lambda,j} &e^1_{\lambda,j} &\bbi e^2_{\lambda,j} &e^2_{\lambda,j} &\dots\\
&0 &0 &0 &0 &1 &-\bbi &e^1_{\lambda,j} &-\bbi e^1_{\lambda,j} &\dots\\
&0 &0 &0 &0 &\bbi &1 &\bbi e^1_{\lambda,j} &e^1_{\lambda,j} &\dots\\
&\vdots &\vdots &\vdots &\vdots &\vdots &\vdots &\vdots &\vdots &\ddots\\} + \nonumber
\\
&\parenv{\frac{1}{2}\parenv{\frac{n}{j}}^{\frac{\overline{\lambda}-k}{\tau-1}} + O\parenv{n^{\frac{\overline{\lambda}-k}{\tau-1}-1}}}
\bmat{&1 &\bbi &e^1_{\overline{\lambda},j} &\bbi e^1_{\overline{\lambda},j} &e^2_{\overline{\lambda},j} &\bbi e^2_{\overline{\lambda},j} &e^3_{\overline{\lambda},j} &\bbi e^3_{\overline{\lambda},j} &\dots\\
&-\bbi &1 &-\bbi e^1_{\overline{\lambda},j} &e^1_{\overline{\lambda},j} &-\bbi e^2_{\overline{\lambda},j}  &e^2_{\overline{\lambda},j} &-\bbi e^3_{\overline{\lambda},j} &e^3_{\overline{\lambda},j} &\dots\\
&0 &0 &1 &\bbi &e^1_{\overline{\lambda},j} &\bbi e^1_{\overline{\lambda},j} &e^2_{\overline{\lambda},j} &\bbi e^2_{\overline{\lambda},j} &\dots\\
&0 &0 &-\bbi &1 &\bbi e^1_{\overline{\lambda},j} &e^1_{\overline{\lambda},j} &-\bbi e^2_{\overline{\lambda},j} &e^2_{\overline{\lambda},j} &\dots\\
&0 &0 &0 &0 &1 &\bbi &e^1_{\overline{\lambda},j} &\bbi e^1_{\overline{\lambda},j} &\dots\\
&0 &0 &0 &0 &-\bbi &1 &-\bbi e^1_{\overline{\lambda},j} &e^1_{\overline{\lambda},j} &\dots\\
&\vdots &\vdots &\vdots &\vdots &\vdots &\vdots &\vdots &\vdots &\ddots\\}. \label{eq: gen case 4}
\end{align}

Before continuing on, we would now like to look deeper into the residual terms known as $e^h_{\lambda,j}$. To investigate their asymptotic behavior, we take $n\to\infty$.
We would also assume that $j$ is linearly related to $n$, $j \propto n$, meaning that as $n\to\infty$, we have $j\to\infty$, later we will remove that assumption.

\begin{align*}
    &e^1_{\lambda,j} = \sum_{t_1=j}^{n} \frac{1}{(t_1-1)(\tau-1) + Y_0 +\lambda-k} = \frac{1}{\tau-1} \sum_{t_1=j}^{n} \frac{1}{t_1 + \frac{Y_0 +\lambda-k}{\tau-1}-1}.
\end{align*}
As $j\to\infty$, $\frac{Y_0 +\lambda-k}{\tau-1}-1$ becomes negligible and can be ignored.

\begin{align*}
    &\frac{1}{\tau-1} \sum_{t_1=j}^{n} \frac{1}{t_1 + \frac{Y_0 +\lambda-k}{\tau-1}-1} \sim \frac{1}{\tau-1} \sum_{t_1=j}^{n} \frac{1}{t_1} = \frac{1}{\tau-1} \frac{1}{n} \sum_{t_1=j}^{n} \parenv{\frac{t_1}{n}}^{-1}.
\end{align*}
Notice we got a Riemann sum that can be approximated as an integral like so:

\begin{align*}
    &\frac{1}{\tau-1} \frac{1}{n} \sum_{t_1=j}^{n} \parenv{\frac{t_1}{n}}^{-1} \sim \frac{1}{\tau-1} \int_{\frac{j}{n}}^1 \frac{1}{t} \,dt =  \frac{1}{\tau-1}\log{\frac{n}{j}}.
\end{align*}
We got an asymptotic logarithmic behavior that depends on the relation $j/n$. In the case where $j$ is constant and $n \to \infty$, $e^1_{\lambda,j}$ does not converge, thus the asymptotic behavior still holds

\begin{align*}
    &e^1_{\lambda,j} \sim \frac{1}{\tau-1}\log{\frac{n}{j}}.
\end{align*}

For $h>1$ we get a similar behavior,

\begin{align*}
    &e^h_{\lambda,j} = \sum_{t=j}^{n-h+1} \frac{1}{(t-1)(\tau-1) + Y_0 +\lambda-k} e^{h-1}_{\lambda,t+1} \sim \sum_{t=j}^{n-h+1} \frac{1}{(t-1)(\tau-1) + Y_0 +\lambda-k} \parenv{\frac{1}{(h-1)!(\tau-1)^{h-1}}\log^{h-1}\frac{n}{t+1}} \sim
    \\
    &\frac{1}{(h-1)!(\tau-1)^h}\sum_{t=j}^{n-h+1} \frac{1}{t} \log^{h-1}\frac{n}{t+1} \sim \frac{1}{(h-1)!(\tau-1)^h} \frac{1}{n}\sum_{t=j}^{n} \frac{n}{t} \log^{h-1}\frac{n}{t},
\end{align*}
as $h$ is negligible as $n \to \infty$. We notice we get a Riemann sum that can be approximated as an integral.

\begin{align*}
    &\frac{1}{(h-1)!(\tau-1)^h} \frac{1}{n}\sum_{t=j}^{n} \frac{n}{t} \log^{h-1}\frac{n}{t} \sim \frac{1}{(h-1)!(\tau-1)^h}  \int_{\frac{j}{n}}^1 \frac{1}{t} \log^{h-1}\frac{1}{t} \,dt = \frac{1}{h!(\tau-1)^h} \log^{h}\frac{n}{j}.
\end{align*}
We have shown $e^h_{\lambda,j}$ exhibits asymptotically logarithmic behavior of order $h$:

\begin{align*}
    &e^h_{\lambda,j} \sim \frac{1}{h!(\tau-1)^h} \log^{h}\frac{n}{j}.
\end{align*}
To continue with \eqref{eq: gen case 4}, we would like to define the following terms, which are the real and imaginary parts of the residual term $e^h_{\lambda,j}$:

\begin{align*}
    e^{r,h}_{\lambda,j} :=& \frac{e^h_{\lambda,j}+e^h_{\overline{\lambda},j}}{2},
    \\
    e^{c,h}_{\lambda,j} :=& \frac{e^h_{\lambda,j}-e^h_{\overline{\lambda},j}}{2 \bbi}.
\end{align*}
As $e^{r,h}_{\lambda,j}$ and $e^{c,h}_{\lambda,j}$ are the real and imaginary parts of $e^h_{\lambda,j}$, it is easy to show their asymptotical behavior.

\begin{align*}
    &e^{r,h}_{\lambda,j} \sim \frac{1}{h!(\tau-1)^h} \log^{h}\frac{n}{j},
    \\
    &e^{c,h}_{\lambda,j} = o(1).
\end{align*}

We now continue with \eqref{eq: gen case 4} by separating the terms into the real and imaginary parts and summing individually:
\begin{align}
&\bbV \bmat{\prod_{t=j}^{n} \parenv{\bbI+\frac{1}{t(\tau-1)+Y_0}\bbJ_{\lambda-k}} & 0 \\ 0& \prod_{t=j}^{n} \parenv{\bbI+\frac{1}{t(\tau-1)+Y_0}\bbJ_{\overline{\lambda}-k}}} \bbV^{-1} = \nonumber
\\
&\parenv{\frac{1}{2}\parenv{\frac{n}{j}}^{\frac{\lambda-k}{\tau-1}} + O\parenv{n^{\frac{\lambda-k}{\tau-1}-1}}}
\bmat{&1 &-\bbi &e^1_{\lambda,j} &-\bbi e^1_{\lambda,j} &e^2_{\lambda,j} &-\bbi e^2_{\lambda,j} &e^3_{\lambda,j} &-\bbi e^3_{\lambda,j} &\dots\\
&\bbi &1 &\bbi e^1_{\lambda,j} &e^1_{\lambda,j} &\bbi e^2_{\lambda,j}  &e^2_{\lambda,j} &\bbi e^3_{\lambda,j} &e^3_{\lambda,j} &\dots\\
&0 &0 &1 &-\bbi &e^1_{\lambda,j} &-\bbi e^1_{\lambda,j} &e^2_{\lambda,j} &-\bbi e^2_{\lambda,j} &\dots\\
&0 &0 &\bbi &1 &\bbi e^1_{\lambda,j} &e^1_{\lambda,j} &\bbi e^2_{\lambda,j} &e^2_{\lambda,j} &\dots\\
&0 &0 &0 &0 &1 &-\bbi &e^1_{\lambda,j} &-\bbi e^1_{\lambda,j} &\dots\\
&0 &0 &0 &0 &\bbi &1 &\bbi e^1_{\lambda,j} &e^1_{\lambda,j} &\dots\\
&\vdots &\vdots &\vdots &\vdots &\vdots &\vdots &\vdots &\vdots &\ddots\\} + \nonumber
\\
&\parenv{\frac{1}{2}\parenv{\frac{n}{j}}^{\frac{\overline{\lambda}-k}{\tau-1}} + O\parenv{n^{\frac{\overline{\lambda}-k}{\tau-1}-1}}}
\bmat{&1 &\bbi &e^1_{\overline{\lambda},j} &\bbi e^1_{\overline{\lambda},j} &e^2_{\overline{\lambda},j} &\bbi e^2_{\overline{\lambda},j} &e^3_{\overline{\lambda},j} &\bbi e^3_{\overline{\lambda},j} &\dots\\
&-\bbi &1 &-\bbi e^1_{\overline{\lambda},j} &e^1_{\overline{\lambda},j} &-\bbi e^2_{\overline{\lambda},j}  &e^2_{\overline{\lambda},j} &-\bbi e^3_{\overline{\lambda},j} &e^3_{\overline{\lambda},j} &\dots\\
&0 &0 &1 &\bbi &e^1_{\overline{\lambda},j} &\bbi e^1_{\overline{\lambda},j} &e^2_{\overline{\lambda},j} &\bbi e^2_{\overline{\lambda},j} &\dots\\
&0 &0 &-\bbi &1 &\bbi e^1_{\overline{\lambda},j} &e^1_{\overline{\lambda},j} &-\bbi e^2_{\overline{\lambda},j} &e^2_{\overline{\lambda},j} &\dots\\
&0 &0 &0 &0 &1 &\bbi &e^1_{\overline{\lambda},j} &\bbi e^1_{\overline{\lambda},j} &\dots\\
&0 &0 &0 &0 &-\bbi &1 &-\bbi e^1_{\overline{\lambda},j} &e^1_{\overline{\lambda},j} &\dots\\
&\vdots &\vdots &\vdots &\vdots &\vdots &\vdots &\vdots &\vdots &\ddots\\} = \nonumber
\\
& \parenv{\parenv{\frac{n}{j}}^{\frac{\lambda^r-k}{\tau-1}} \cos\parenv{\frac{\lambda^c}{\tau-1}\log\parenv{\frac{n}{j}}} + O\parenv{n^{\frac{\lambda^r-k}{\tau-1}-1}}}
\bmat{&1 &0 &e^{r,1}_{\lambda,j} &e^{c,1}_{\lambda,j} &e^{r,2}_{\lambda,j} &e^{c,2}_{\lambda,j} &e^{r,3}_{\lambda,j} &e^{c,3}_{\lambda,j} &\dots\\
&0 &1 &-e^{c,1}_{\lambda,j} &e^{r,1}_{\lambda,j} &-e^{c,2}_{\lambda,j}  &e^{r,2}_{\lambda,j} &-e^{c,3}_{\lambda,j} &e^{r,3}_{\lambda,j} &\dots\\
&0 &0 &1 &0 &e^{r,1}_{\lambda,j} &e^{c,1}_{\lambda,j} &e^{r,2}_{\lambda,j} &e^{c,2}_{\lambda,j} &\dots\\
&0 &0 &0 &1 &-e^{c,1}_{\lambda,j} &e^{r,1}_{\lambda,j} &-e^{c,2}_{\lambda,j} &e^{r,2}_{\lambda,j} &\dots\\
&0 &0 &0 &0 &1 &0 &e^{r,1}_{\lambda,j} &e^{c,1}_{\lambda,j} &\dots\\
&0 &0 &0 &0 &0 &1 &-e^{c,1}_{\lambda,j} &e^{r,1}_{\lambda,j} &\dots\\
&\vdots &\vdots &\vdots &\vdots &\vdots &\vdots &\vdots &\vdots &\ddots\\} + \nonumber
\\
& \parenv{\parenv{\frac{n}{j}}^{\frac{\lambda^r-k}{\tau-1}} \sin\parenv{\frac{\lambda^c}{\tau-1}\log\parenv{\frac{n}{j}}} + O\parenv{n^{\frac{\lambda^r-k}{\tau-1}-1}}}
\bmat{&0 &1 &-e^{c,1}_{\lambda,j} &e^{r,1}_{\lambda,j} &-e^{c,2}_{\lambda,j} &e^{r,2}_{\lambda,j} &-e^{c,3}_{\lambda,j} &e^{r,3}_{\lambda,j} &\dots\\
&-1 &0 &-e^{r,1}_{\lambda,j} &-e^{c,1}_{\lambda,j} &-e^{r,2}_{\lambda,j}  &-e^{c,2}_{\lambda,j} &-e^{r,3}_{\lambda,j} &-e^{c,3}_{\lambda,j} &\dots\\
&0 &0 &0 &1 &-e^{c,1}_{\lambda,j} &e^{r,1}_{\lambda,j} &-e^{c,2}_{\lambda,j} &e^{r,2}_{\lambda,j} &\dots\\
&0 &0 &-1 &0 &-e^{r,1}_{\lambda,j} &-e^{c,1}_{\lambda,j} &-e^{r,2}_{\lambda,j} &-e^{c,2}_{\lambda,j} &\dots\\
&0 &0 &0 &0 &0 &1 &-e^{c,1}_{\lambda,j} &e^{r,1}_{\lambda,j} &\dots\\
&0 &0 &0 &0 &-1 &0 &-e^{r,1}_{\lambda,j} &-e^{c,1}_{\lambda,j} &\dots\\
&\vdots &\vdots &\vdots &\vdots &\vdots &\vdots &\vdots &\vdots &\ddots\\}. \label{eq: gen case 5}
\end{align}
Finally, plugging \eqref{eq: gen case 5} back into \eqref{eq: gen case 3}, we get a close term for $\beta^{r(i,h)}_{j,n}$ and $\beta^{c(i,h)}_{j,n}$:

\begin{align}
\beta^{r(i,h)}_{j-1,n} = \parenv{\frac{n}{j}}^{\frac{\lambda^r-k}{\tau-1}} \Bigg\{&\cos\parenv{\frac{\lambda^c}{\tau-1}\log \parenv{\frac{n}{j}}}\parenv{\alpha^r_{i,h} + \sum_{l=1}^{H_i-h}\parenv{e^{r,l}_{\lambda,j} \alpha^r_{i,h+l} + e^{c,l}_{\lambda,j} \alpha^c_{i,h+l}}} + \nonumber
\\
&\sin\parenv{\frac{\lambda^c}{\tau-1}\log \parenv{\frac{n}{j}}}\parenv{\alpha^c_{i,h} + \sum_{l=1}^{H_i-h}\parenv{e^{r,l}_{\lambda,j} \alpha^c_{i,h+l} - e^{c,l}_{\lambda,j} \alpha^r_{i,h+l}}}\Bigg\} + O\parenv{n^{\frac{\lambda^r-k}{\tau-1}-1}\log^{H_i-h}n}
\\
= O\parenv{n^{\frac{\lambda^r-k}{\tau-1}}\log^{H_i-h}n}&, \nonumber
\\
\nonumber
\\
\beta^{c(i,h)}_{j-1,n} = \parenv{\frac{n}{j}}^{\frac{\lambda^r-k}{\tau-1}} \Bigg\{&\cos\parenv{\frac{\lambda^c}{\tau-1}\log \parenv{\frac{n}{j}}}\parenv{\alpha^c_{i,h} + \sum_{l=1}^{H_i-h}\parenv{e^{r,l}_{\lambda,j} \alpha^c_{i,h+l} - e^{c,l}_{\lambda,j} \alpha^r_{i,h+l}}} - \nonumber
\\
&\sin\parenv{\frac{\lambda^c}{\tau-1}\log \parenv{\frac{n}{j}}}\parenv{\alpha^r_{i,h} + \sum_{l=1}^{H_i-h}\parenv{e^{r,l}_{\lambda,j} \alpha^r_{i,h+l} + e^{c,l}_{\lambda,j} \alpha^c_{i,h+l}}}\Bigg\} + O\parenv{n^{\frac{\lambda^r-k}{\tau-1}-1}\log^{H_i-h}n}
\\
= O\parenv{n^{\frac{\lambda^r-k}{\tau-1}}\log^{H_i-h}n}&. \nonumber
\end{align}
We find the small error term to be:

\begin{align*}
\varepsilon^{(i,H_i)}_n =&\sum_{h=1}^{H_i}-X_0^{r(i,h)}\parenv{\parenv{1+\frac{\lambda^r-k}{Y_0}} \beta^{r(i,h)}_{1,n} + \frac{\lambda^c}{Y_0} \beta^{c(i,h)}_{1,n}}-X_0^{c(i,h)}\parenv{- \frac{\lambda^c}{Y_0} \beta^{r(i,h)}_{1,n} + \parenv{1+\frac{\lambda^r-k}{Y_0}} \beta^{c(i,h)}_{1,n}} + 
\\
&\sum_{h=1}^{H_i-1}-X_0^{r(i,h)}\frac{1}{Y_0} \beta^{r(i,h+1)}_{1,n}-X_0^{c(i,h)}\frac{1}{Y_0} \beta^{c(i,h+1)}_{1,n} = O\parenv{n^{\frac{\lambda^r-k}{\tau-1}}\ln^{H_i-1}n}.
\end{align*}

As we can see, we got results similar to the diagonalizable case, but with logarithmic residuals $e^{h}_{\lambda,j}$.
Those residuals make the analysis slightly more involved, but the core ideas remain the same. 
\end{IEEEproof}

\bibliographystyle{IEEEtranS}
\bibliography{allbib.bib}
\end{document}